\documentclass[10pt, english]{article}

\usepackage{babel,latexsym,graphicx,amsmath,amsfonts,amsthm,amssymb}
\usepackage{dsfont, fullpage, pstricks}
\usepackage[round]{natbib}

\newtheoremstyle{localrem}
	{5pt} 
	{5pt} 
	{\rm} 
	{} 
	{\bf} 
	{{\rm.}} 
	{.7em} 
	{} 

\theoremstyle{localrem}
\newtheorem{Definition}{Definition}[section]
\newtheorem{Remark}[Definition]{Remark}
\newtheorem{Example}[Definition]{Example}

\newtheoremstyle{localthm}
	{5pt} 
	{5pt} 
	{\sl} 
	{} 
	{\bf} 
	{{\rm.}} 
	{.7em} 
	{} 

\theoremstyle{localthm}

\newtheorem{Lemma}[Definition]{Lemma}

\newcommand{\R}{\mathbb{R}}
\newcommand{\V}{\mathbb{V}}

\newcommand{\DD}{\mathcal{D}}
\newcommand{\XX}{\mathcal{X}}
\newcommand{\eps}{\varepsilon}

\DeclareMathOperator*{\argmax}{arg\,max}

\def\DD{\mathcal{D}}
\def\LL{\mathcal{L}}
\def\NN{\mathcal{N}}

\def\VV{\mathcal{V}}

\def\Ex{\mathop{\rm I\!E}\nolimits}
\def\Pr{\mathop{\rm I\!P}\nolimits}

\def\bs{\boldsymbol}

\def\thb{\bs{\theta}}
\def\taub{\bs{\tau}}

\begin{document}

\addtolength{\baselineskip}{0.4\baselineskip}

\title{Active Set Algorithms for\\
	Estimating Shape-Constrained Density Ratios}
\author{Lutz D\"{u}mbgen$^1$, Alexandre M\"{o}sching$^2$ and Christof Str\"ahl$^1$\\
	$^1$University of Bern and $^2$University of G\"ottingen}
\date{December 2020}

\maketitle

\begin{abstract}
In many instances, imposing a constraint on the shape of a density is a reasonable and flexible assumption. It offers an alternative to parametric models which can be too rigid and to other nonparametric methods requiring the choice of tuning parameters. This paper treats the nonparametric estimation of log-concave or log-convex density ratios by means of active set algorithms in a unified framework. In the setting of log-concave densities, the new algorithm is similar to but substantially faster than previously considered active set methods. Log-convexity is a less common shape constraint which is described by some authors as ``tail inflation''. The active set method proposed here is novel in this context. As a by-product, new goodness-of-fit tests of single hypotheses are formulated and are shown to be more powerful than higher criticism tests in a simulation study.
\end{abstract}

\section{Introduction}
\label{sec:introduction}

Suppose we observe independent random variables $X_1, X_2, \ldots, X_n$ with unknown distributions $P_1, P_2, \ldots, P_n$ on the real line. This paper discusses the estimation of the marginal (average) distribution $P := n^{-1} \sum_{i=1}^n P_i$ under certain shape constraints on $P$. This framework includes the case of i.i.d.\ observations from a single distribution $P$, of course.

Within the broad field of nonparametric statistics, inference about $P$ under shape-constraints is a well-established alternative to the assumption of quantitative smoothness properties, e.g.\ certain bounds on the maximum modulus of some higher order derivative of the density of $P$ (w.r.t.\ Lebesgue measure). While estimation under smoothness assumptions involves typically tuning parameters, e.g.\ bandwidths of kernel density estimators, maximum likelihood estimation under shape constraints is often possible without any further specifications. For a thorough discussion of the benefits of shape-constraints we refer to \cite{Groeneboom_Jongbloed_2014}.

One particular example of a shape constraint is log-concavity of the density of $P$. A broad overview of statistical methods with such densities, including the multivariate case, is given by \cite{samworth2018}. A second example of a shape constraint is convexity of the density of $P$ on the positive half-line, see \cite{groeneboom2001}. In the present paper we reconsider the estimation of log-concave densities and a less familiar setting which is related to the estimation of convex densities:

\paragraph{Setting 1: Log-concave densities.}
We assume that $P$ has a log-concave density $f$ with respect to Lebesgue measure, that means, $\log f : \R \to [-\infty,\infty)$ is concave.

\paragraph{Setting 2: Tail inflation.}
For a given continuous reference distribution $P_o$ on $\R$, we assume that $P$ has a log-convex density $f$ with respect to $P_o$, that means, $\log f : \R \to \R$ is convex.

\medskip

The notion of tail inflation has been introduced by \cite{mccullagh2012,McCullagh_Polson_2017} to investigate statistical sparsity. They consider the case of observations $X_i > 0$, the reference distribution $P_o$ being the chi-squared distribution with one degree of freedom, and $\log f$ is assumed to be convex and isotonic (non-decreasing). Setting~2 is also related to multiple hypothesis testing. There, $X_1, X_2, \ldots, X_n$ represent test statistics for given null hypotheses $H_1, H_2, \ldots, H_n$, where $X_i$ has distribution $P_o$ whenever $H_i$ is true. In image analysis, the random variables $X_i$ could be measured intensities at different pixels of a digital image, and $P_o$ describes pure background noise or measurement errors.

Primary goals are to estimate $P$ or to test the null hypothesis that all $P_i$ are equal to $P_o$. The assumption of log-convexity of $f = dP/dP_o$ may seem a bit arbitrary at first sight. But note, for instance, that the testing problems considered by \cite{Donoho_Jin_2004} may be viewed as a special case of Setting~2, with $P_o$ being the standard Gaussian distribution $\NN(0,1)$. Indeed, the latter authors considered i.i.d.\ observations with distribution $P$ being a mixture $(1 - \eps) \NN(0,1) + \eps \NN(\mu,1)$ with unknown parameters $\eps \in [0,1]$ and $\mu \ge 0$. As shown later, if each $P_i$ is a mixture of Gaussian distributions with standard deviation at least $1$, then each $P_i$ as well as the marginal distribution $P$ has a log-convex density with respect to $P_o$. Consequently, if we estimate the log-density $\theta := \log f$ of $P$, this gives rise to a new likelihood ratio test statistic for the null hypothesis that all $P_i$ are equal to $P_o$.

\paragraph{Outline of the paper.}
Our main goals are to establish existence and uniqueness of the nonparametric maximum likelihood estimator $\hat{\theta}$ of $\theta := \log f$ in Setting~2 and to devise explicit algorithms for its computation. Since Settings~1 and 2 are closely related, it is worthwhile to treat both of them simultaneously, highlighting similarities and differences. In Section~\ref{sec:settings}, the specific estimation problems are described in more detail, and it is shown that under certain assumptions, the maximizer $\hat{\theta}$ exists and is unique.

In Section~\ref{sec:active.sets}, we describe a general active set method for the computation of $\hat{\theta}$. The starting point is the active set method described by \cite{Duembgen_etal_2011} and \cite{Duembgen_Rufibach_2011}, which is similar to the support reduction algorithm of \cite{Groeneboom_etal_2008}. The new version is more efficient in that all single Newton steps take shape constraints on $\theta$ into account. We also adopt the proposal of \cite{Liu_Wang_2018} to deactivate occasionally more than one constraint in one step, but other than the latter authors, we do not resort to quadratic programming routines within the algorithm. In Setting~2, we explore the full infinite-dimensional parameter space rather than using ad hoc finite-dimensional approximations.

Numerical examples illustrating the estimation method are given in Section~\ref{sec:numerical.examples}. For Setting~1, we demonstrate the benefits of the new method in a small simulation study. We also show that our estimator for Setting~2 leads to a promising goodness-of-fit test. Simulations show that its power can exceed the power of higher criticism methods as proposed by \cite{Donoho_Jin_2004} and \cite{Gontscharuk_etal_2016}. 

Section~\ref{sec:proofs} provides proofs for the existence, uniqueness and special properties of $\hat{\theta}$, while Appendix~\ref{sec:technical.details} provides technical details for specific applications and a proof of convergence which generalises and simplifies a previous proof of \cite{Sommer-Simpson_2019}. The algorithms have been implemented in the statistical langage R \citep{R2016} and are available from the authors.

\section{General considerations, existence and uniqueness}
\label{sec:settings}

In what follows, we consider an arbitrary discrete distribution
\[
	\hat{P} \ := \ \sum_{i=1}^n w_i \delta_{x_i}
\]
with $n \ge 2$ probability weights $w_1, \ldots, w_n > 0$ and real support points $x_1 < \cdots < x_n$. In Settings~1 and 2, these points $x_1,\ldots,x_n$ are the order statistics of the observations $X_1,\ldots,X_n$ while $w_i = n^{-1}$. The general form of $\hat{P}$ covers also the situation of $N \ge n$ raw observations from $P$ which are recorded with rounding errors. Then $x_1,\ldots,x_n$ are the different recorded values, and $w_i$ is the relative frequency of $x_i$ in the sample.

\subsection{Parameter spaces and target functional}

In general, we assume that $\hat{P}$ estimates an unknown distribution $P$ which has a density $f$ with respect to a given continuous measure $M$ on $\R$. Precisely,
\[
	f(x) \ = \ f_\theta(x) := e_{}^{\theta(x)}
\]
with an unknown function parameter $\theta : \R \to [-\infty,\infty)$ in a given family $\Theta$ reflecting the particular shape constraints to be specified later. Then $\theta$ is estimated by a function $\hat{\theta} \in \Theta$ maximizing the normalized log-likelihood
\[
	\ell(\theta)
	\ := \ \int \theta \, d\hat{P}
	\ = \ \sum_{i=1}^n w_i \theta(x_i)
\]
under the constraint that $\int e^{\theta} \, dM = 1$.

In the specific settings we have in mind, all functions $\theta \in \Theta$ satisfy $0 < \int e^\theta \, dM \le \infty$ and $\theta + c \in \Theta$ for arbitrary real constants $c$. Thus we may apply the Lagrange trick of \cite{silverman1982} and rewrite $\hat{\theta}$ as
\[
	\hat{\theta}
	\ = \ \argmax_{\theta \in \Theta} \, L(\theta)
\]
with
\[
	L(\theta) \ := \ \int \theta \, d\hat{P} - \int e_{}^\theta \, dM + 1
	\ \in \ [-\infty,\infty) .
\]
Indeed, for $\theta \in \Theta$ with $L(\theta) > -\infty$ and $c \in \R$, the derivative $\partial L(\theta + c)/\partial c$ equals $1 - e^c \int e_{}^\theta \, dM$. Hence, a function $\hat{\theta} \in \Theta$ with $L(\hat{\theta}) > -\infty$ maximizes $L(\theta)$ over all $\theta \in \Theta$ if and only if it maximises $\ell(\theta)$ under the constraint that $\int e^\theta \, dM = 1$. Note also that $L(\theta) = \ell(\theta)$ if and only if $\int e^\theta \, dM = 1$.

\paragraph{Setting~1.}
$M$ is Lebesgue measure on $\R$, and the parameter space $\Theta$ consists of all concave, upper semicontinuous functions $\theta : \R \to [-\infty,\infty)$ such that $\int e^\theta \, dM > 0$.

\paragraph{}
For Setting~2 from the introduction, we distinguish between two versions, where the second one covers the framework of \cite{mccullagh2012}.

\paragraph{Setting~2A.}
$M$ stands for the reference distribution $P_o$. We assume that $P_o$ is continuous with $P_o(B) > 0$ for any non-degenerate interval $B \subset \R$, and
\[
	\Bigl\{ \lambda \in \R : \int e^{\lambda x} \, P_o(dx) < \infty \Bigr\}
	\ = \ \bigl( \lambda_{\ell}(P_o), \lambda_r(P_o) \bigr)
\]
for certain numbers $-\infty \le \lambda_\ell(P_o) < 0 < \lambda_r(P_o) \le \infty$. The extended parameter space $\Theta$ consists of all convex functions $\theta : \R \to \R$.

\begin{Example}[Gaussian mixtures]
Let $P_o = \NN(0,1)$. Suppose that $P$ is a mixture of Gaussian distributions with standard deviation at least $1$, i.e.\ $P = \int \NN(\mu,\sigma^2) \, Q(d\mu, d\sigma)$ for some probability distribution $Q$ on $\R \times [1,\infty)$. Then $\theta := \log dP/dP_o$ is given by
\[
	\theta(x)
	\ = \ \log \int e_{}^{\theta(x,\mu,\sigma)} \, Q(d\mu,d\sigma)
\]
with
\[
	\theta(x,\mu,\sigma)
	\ := \ \log \frac{d\NN(\mu,\sigma^2)}{d\NN(0,1)}(x)
		\ = \ - \log \sigma
		+ \frac{(\sigma^2 - 1) x^2 + 2\mu x - \mu^2}{2\sigma^2} .
\]
Obviously, $\theta(\cdot,\mu,\sigma)$ is a convex function for arbitrary $\mu \in \R$ and $\sigma \ge 1$, so the log-mixture density $\theta$ is convex, too. This can be deduced from H\"older's inequality or Artin's theorem, see Section~D.4 of \cite{Marshall_Olkin_1979}.
\end{Example}

\begin{Example}[Student distributions]
Let $P_o = \NN(0,\sigma^2)$ and $P = t_k$ with $\sigma, k > 0$. Tedious but elementary calculations show that $\theta = \log(dP/dP_o)$ is convex if and only if $\sigma^2 \le k/(k+1)$.
\end{Example}

\begin{Example}[Logistic distributions]
Let $P_o = \NN(0,1)$, and let $P$ be the logistic distribution with scale parameter $\sigma > 0$, i.e.\ with Lebesgue density $p(x) = \sigma^{-1} (e^{x/\sigma} + e^{-x/\sigma} + 2)^{-1}$. Here one can show that $\theta = \log(dP/dP_o)$ is convex if and only if $\sigma \ge 2^{-1/2}$.
\end{Example}

\paragraph{Setting~2B.}
$M$ stands for the reference distribution $P_o$. We assume that $P_o$ is continuous such that $P_o((-\infty,0]) = 0$ and $P_o(B) > 0$ for any non-degenerate interval $B \subset (0,\infty)$, and
\[
	\Bigl\{ \lambda \in \R : \int e^{\lambda x} \, P_o(dx) < \infty \Bigr\}
	\ = \ \bigl( -\infty, \lambda_r(P_o) \bigr)	
\]
for some number $\lambda_r(P_o) \in (0,\infty]$. Now the extended parameter space $\Theta$ consists of all convex functions $\theta : \R \to \R$ such that $\theta \equiv \theta(0)$ on $(-\infty,0]$. In particular, all $\theta \in \Theta$ are isotonic.

\begin{Example}[Scale mixtures of Gamma distributions]
Let $P_o = \mathrm{Gamma}(\alpha,\beta)$, the gamma distribution with given shape parameter $\alpha > 0$ and rate parameter $\beta > 0$. Suppose that $P$ is a scale mixture of gamma distributions with the same shape parameter, i.e.\ $P = \int \mathrm{Gamma}(\alpha,\beta/s) \, Q(ds)$ for some probability measure $Q$ on $(0,\infty)$. Then $\theta := \log dP/dP_o$ is given by
\[
	\theta(x)
	\ = \ \log \int e^{\theta(x,s)} \, Q(ds)
\]
with $\theta(x,s) := \beta (1 - 1/s) x - \alpha \log s$. The latter expression is linear in $x$, whence $\theta$ is convex. If $Q([1,\infty)) = 1$, then $\theta$ is also isotonic.

A special instance of this setting are raw observations $\tilde{X}_i = S_i G_i$, $1 \le i \le n$, with independent random variables $S_1, \ldots, S_n \ge 1$ and $G_1,\ldots,G_n \sim \NN(0,1)$. With $P_o := \chi_1^2 = \mathrm{Gamma}(1/2,1/2)$, the marginal distribution $P$ of the observations $X_i := \tilde{X}_i^2$ has the log-density $\theta = \log \int e^{\theta(\cdot,s)} \, Q(ds)$ with respect to $P_o$, where $Q := n^{-1} \sum_{i=1}^n \LL(S_i)$.
\end{Example}

\subsection{Existence and uniqueness of the estimator}

In Settings~1 and 2A-B, the target functional $L$ is strictly concave on the convex set $\{\theta \in \Theta : L(\theta) > - \infty\}$. This follows easily from strict convexity of the exponential function. Precisely, there exists a unique maximizer $\hat{\theta} \in \Theta$ of $L$ which is piecewise linear and satisfies further properties as summarized in the following three lemmas. The first one has been proved by \cite{Walther_2002}, see also \cite{Duembgen_etal_2011} or \cite{Cule_etal_2010}:

\begin{Lemma}
\label{lem:existence.uniqueness.1}
In Setting~1, there exists a unique maximizer $\hat{\theta}$ of $L$ over $\Theta$. Precisely, there exist $m \ge 2$ points $\tau_1 < \cdots < \tau_m$ in $\{x_1,x_2,\ldots,x_n\}$ with $\tau_1 = x_1$, $\tau_m = x_n$, with the following properties:
\[
	\hat{\theta} \ \begin{cases}
		\text{is linear on} \ [\tau_j,\tau_{j+1}], \ 1 \le j < m , \\
		\text{equals} \ -\infty \ \text{on} \ \R \setminus [x_1,x_n] ,
	\end{cases}
\]
and the slope $\hat{\theta}'(\tau_j\,+) = \bigl( \hat{\theta}(\tau_{j+1}) - \hat{\theta}(\tau_j) \bigr)/(\tau_{j+1} - \tau_j)$ is strictly decreasing in $j \in \{1,\ldots,m-1\}$.
\end{Lemma}

\begin{Lemma}
\label{lem:existence.uniqueness.2A}
In Setting~2A, there exists a unique maximizer $\hat{\theta}$ of $L$ over $\Theta$. Precisely, either $\hat{\theta}$ is linear, or there exist $m \in \{1,\ldots,n-1\}$ points $\tau_1 < \cdots < \tau_m$ in $[x_1,x_n] \setminus \{x_1,\ldots,x_n\}$ with the following properties:
\[
	\hat{\theta} \ \text{is linear on} \ \begin{cases}
		\XX_0 := (-\infty, \tau_1], \\
		\XX_j := [\tau_j,\tau_{j+1}], \ 1 \le j < m , \\
		\XX_m := [\tau_m,\infty) ,
	\end{cases}
\]
and the sequence of slopes of $\hat{\theta}$ on these $m+1$ intervals is strictly increasing. Furthermore, each interval $(x_i, x_{i+1})$, $1 \le i < n$, contains at most one point $\tau_j$.
\end{Lemma}

\begin{Lemma}
\label{lem:existence.uniqueness.2B}
In Setting~2B, there exists a unique maximizer $\hat{\theta}$ of $L$ over $\Theta$. Precisely, either $\hat{\theta} \equiv 0$, or there exist $m \in \{1,\ldots,n-1\}$ points $\tau_1 < \cdots < \tau_m$ in $\{0\} \cup [x_1,x_n] \setminus \{x_1,\ldots,x_n\}$ with the following properties:
\[
	\hat{\theta} \ \text{is} \ \begin{cases}
		\text{constant on} \ (-\infty, \tau_1] , \\
		\text{linear on} \ \XX_j := [\tau_j,\tau_{j+1}], \ 1 \le j < m-1 , \\
		\text{linear on} \ \XX_m :=	[\tau_m,\infty) ,
	\end{cases}
\]
and the slope $\hat{\theta}'(\tau_j\,+)$ is strictly positive and strictly increasing in $j \in \{1,\ldots,m\}$. Furthermore, each interval $(x_i, x_{i+1})$, $1 \le i < n$, contains at most one point $\tau_j$.
\end{Lemma}

Note that the number $m$ in Lemma~\ref{lem:existence.uniqueness.2B} could be $1$, meaning that $\hat{\theta}$ is constant on $[0,\tau_1]$ and linear on $[\tau_1,\infty)$ with slope $\hat{\theta}'(\tau_1\,+) \in (0,\lambda_r(P_o))$.

\section{A general active set strategy}
\label{sec:active.sets}

\subsection{The space of relevant functions}

In view of Lemmas~\ref{lem:existence.uniqueness.1}, \ref{lem:existence.uniqueness.2A} and \ref{lem:existence.uniqueness.2B}, it suffices to consider continuous, piecewise linear functions $\theta$ on
\[
	\XX \ := \ \begin{cases}
		[x_1,x_n]  & \text{in Setting~1} \\
		\R         & \text{in Setting~2A} \\
		[0,\infty) & \text{in Setting~2B}	
	\end{cases}
\]
with changes of slope only in
\[
	\DD \ := \ \begin{cases}
		\{x_i : 1 < i < n\}  & \text{in Setting~1} , \\
		(x_1,x_n)            & \text{in Setting~2A} , \\
		\{0\} \cup (x_1,x_n) & \text{in Setting~2B} .
	\end{cases}
\]
In Setting~2B, a change of slope at $0$ means that $\theta'(0\,+) \ne 0$. The linear space of all such functions $\theta$ is denoted by $\V$. One particular basis is given by the functions
\begin{align*}
	x \ &\mapsto \ 1 , \\
	x \ &\mapsto \ x \quad (\text{in Settings~1 and 2A})
\intertext{and}
	x \ &\mapsto \ V_\tau(x) \ := \ \xi (x - \tau)^+, \quad \tau \in \DD ,
\end{align*}
where
\[
	\xi \ := \ \begin{cases}
		-1 & \text{in Setting~1} , \\
		+1 & \text{in Settings~2A-B} .
	\end{cases}
\]
That means, $\dim(\V)$ equals $n$ in Setting~1 and $\infty$ in Settings~2A-B. Any $\theta \in \V$ may be written as
\begin{equation}
\label{eq:representation.theta}
	\theta(x) \ = \ \left\{\!\!\begin{array}{l}
		\alpha_0 \\[0.5ex]
		\quad + \ \alpha_1 x	\quad (\text{in Settings~1 and 2A}) \\[0.5ex]
		\displaystyle
		\quad\quad + \ \sum_{\tau \in \DD} \beta_\tau V_\tau(x)
	\end{array}\!\!\right\}
\end{equation}
with real coefficients $\alpha_0, \alpha_1, \beta_\tau$ such that $\beta_\tau \ne 0$ for at most finitely many $\tau \in \DD$. Note that $\xi \beta_\tau$ is equal to the change of slope, $\theta'(\tau\,+) - \theta'(\tau\,-)$, whence
\[
	\theta \in \Theta
	\quad\text{if and only if} \quad
	\beta_\tau \ge 0 \ \ \text{for all} \ \tau \in \DD .
\]

\subsection{Properties of $L$}
\label{subsec:properties.L}

On the set $\V$, the functional $L$ is continuous with respect to the norm
\begin{equation}
\label{eq:norm.on.V}
	\|\theta\| \ := \ \begin{cases}
		\max_{x \in [x_1,x_n]} \, |\theta(x)|
			& \text{in Setting~1}, \\
		\max_{x \in [x_1,x_n]} \, |\theta(x)| + |\theta'(x_1)| + |\theta'(x_n)|
			& \text{in Settings~2A-B} .
	\end{cases}		
\end{equation}
For Setting~1, $\|\cdot\|$ quantifies uniform convergence on $\XX$. For Settings~2A-B, convergence with respect to $\|\cdot\|$ is equivalent to uniform convergence on arbitrary bounded subsets of $\XX$. Moreover, in Setting~1, $L$ is real-valued, whereas in Settings~2A-B it follows from our assumptions on $P_o$ that
\[
	\{\theta \in \V : L(\theta) > - \infty\}
	\ = \ \begin{cases}
		\{\theta \in \V : \theta'(x_1) > \lambda_\ell(P_o)
				\ \text{and} \ \theta'(x_n) < \lambda_r(P_o)\}
			& \text{in Setting~2A} , \\
		\{\theta \in \V : \theta'(x_n) < \lambda_r(P_o)\}
			& \text{in Setting~2B} .
	\end{cases}
\]
Finally, on the set $\{\theta \in \V : L(\theta) > - \infty\}$, the functional $L$ is strictly concave. Precisely, for $\theta, v \in \V$ with $L(\theta) > -\infty$,
\begin{align*}
	DL(\theta,v) \
	&:= \ \frac{d}{dt} \Big|_{t = 0} \, L(\theta + tv)
		\ = \ \int v \, d\hat{P} - \int_{\XX} v e_{}^\theta \, dM , \\
	H(\theta,v) \
	&:= \ - \frac{d^2}{dt^2} \Big|_{t = 0} \, L(\theta + tv)
		\ = \ \int_{\XX} v^2 e_{}^\theta \, dM .
\end{align*}
These derivatives $DL(\theta,v)$ and $H(\theta,v)$ are well-defined, because $\int_{\XX} e^{\theta(x) + \eps |x|} \, M(dx) < \infty$ for sufficiently small $\eps > 0$. Note that $H(\theta,v) > 0$ unless $\|v\| = 0$.

\subsection{Characterizing $\hat{\theta}$}
\label{subsec:Characterization}

The properties of $L$ imply that a function $\theta \in \V \cap \Theta$ with $L(\theta) > - \infty$ equals $\hat{\theta}$ if and only if
\begin{equation}
\label{eq:characterization.0}
	DL(\theta,v) \ \le \ 0
	\quad\text{for any} \ v \in \V \
		\text{such that} \ \theta + tv \in \Theta \ \text{for some} \ t > 0 .
\end{equation}
Representing $\theta$ as in \eqref{eq:representation.theta} and $v$ analogously, one can easily verify that \eqref{eq:characterization.0} is equivalent to the following four conditions:
\begin{align}
\label{eq:characterization.1a}
	\int_{\XX} e_{}^\theta \, dM \
	&= \ 1 , \\
\label{eq:characterization.1b}
	\int_{\XX} x e_{}^{\theta(x)} \, M(dx) \
	&= \ \hat{\mu} \quad (\text{in Settings~1 and 2A}) , \\
\label{eq:characterization.1c}
	\int_{\XX} V_\tau^{} e_{}^{\theta} \, dM \
	&= \ \int V_\tau \, d\hat{P}
		\quad \text{whenever} \ \beta_\tau > 0 , \\
\label{eq:characterization.1d}
	\int_{\XX} V_\tau e^\theta \, dM \
	&\ge \ \int V_\tau \, d\hat{P}
		\quad \text{whenever} \ \beta_\tau = 0 ,
\end{align}
where $\hat{\mu}$ denotes the empirical mean $\hat{\mu} := \int x \, \hat{P}(dx) = \sum_{i=1}^n w_i x_i$.

\paragraph{Local optimality.}
Requirements (\ref{eq:characterization.1a}--\ref{eq:characterization.1c}) can be interpreted as follows: For $\theta \in \V$ let $D(\theta) \subset \DD$ be the finite set of its ``deactivated (equality) constraints''. That means,
\[
	D(\theta) \ := \ \bigl\{ \tau \in \DD : \theta'(\tau\,-) \ne \theta'(\tau\,+) \bigr\} .
\]
For an arbitrary finite set $D \subset \DD$ we define
\[
	\V_D \ := \ \bigl\{ \theta \in \V : D(\theta) \subset D \bigr\} .
\]
This is a linear subspace of $\V$ with dimension $2 + \#D$ (in Settings~1 and 2A) or $1 + \# D$ (in Setting~2B). Then requirements (\ref{eq:characterization.1a}--\ref{eq:characterization.1c}) are equivalent to saying that $\int_{\XX} v e_{}^\theta \, dM = \int v \, d\hat{P}$ for all $v \in \V_{D(\theta)}$, that means,
\begin{equation}
\label{eq:optimality.1}
	DL(\theta, v) \ = \ 0
	\quad\text{for all} \ v \in \V_{D(\theta)} .
\end{equation}
In other words, $\theta$ is ``locally optimal'' in the sense that
\[
	\theta \ = \ \argmax_{\eta \in \V_{D(\theta)}} \, L(\eta) .
\]

\paragraph{Checking global optimality.}
Requirement \eqref{eq:characterization.1d} is equivalent to
\begin{equation}
\label{eq:optimality.2}
	h_\theta(\tau) := DL(\theta, V_\tau) \ \le \ 0
	\quad\text{for all} \ \tau \in \DD \setminus D(\theta) .
\end{equation}
Thus a function $\theta \in \V \cap \Theta$ with $L(\theta) > - \infty$ is equal to $\hat{\theta}$ if and only if it is locally optimal in the sense of \eqref{eq:optimality.1} and satisfies \eqref{eq:optimality.2}. As explained in Section~\ref{subsec:Localised.kinks}, for computational efficiency and numerical accuracy it is advisable to replace the simple kink functions $V_\tau$ with localised versions $V_{\tau,\theta} = V_\tau - \eta_{\tau,\theta}$, where $\eta_{\tau,\theta} \in \V_{D(\theta)}$, but the general description of our methods is easier in terms of the $V_\tau$.

\subsection{Basic procedures}

Our active set method involves a candidate $\theta \in \Theta \cap \V$ for the function $\hat{\theta}$ such that $f_\theta$ defines a probability density w.r.t.\ $M$ and a finite set $D \subset \DD$ such that $D(\theta) \subset D$.

\subsubsection*{Basic step~1: Obtaining a proposal $\theta_{\rm new}$ via Newton's method}

Recall that the functional $L$ is continuous and concave on the finite-dimensional space $\V_D$. Moreover, on $\{\eta \in \V_D : L(\eta) > -\infty\}$ it is twice continuously differentiable with negative definite Hessian operator. Thus we may perform a standard Newton step to obtain a function $\theta_{\rm new} \in \V_D$ such that
\[
	\delta := DL(\theta, \theta_{\rm new} - \theta) \ \ge \ 0
\]
with equality if and only if
\[
	\theta = \theta_{\rm new} \ = \ \argmax_{\eta \in \V_D} \, L(\eta) .
\]
Even in case of $\delta > 0$, it may happen that $L(\theta_{\rm new}) \le L(\theta)$. To guarantee a real improvement, we apply a standard Armijo--Goldstein step size correction and replace $\theta_{\rm new}$ with $\theta + 2^{-n} (\theta_{\rm new} - \theta)$, where $n$ is the smallest nonnegative integer such that
\[
	L(\theta + 2^{-n}(\theta_{\rm new} - \theta)) - L(\theta)
	\ \ge \ 2^{-n} DL(\theta, \theta_{\rm new} - \theta) / 3 .
\]
(A theoretical justification of this step size correction can be found, for instance, in \cite{duembgen2017}.) In algorithmic language, as long as $L(\theta_{\rm new}) < L(\theta) + \delta/3$, we replace $(\theta_{\rm new}, \delta)$ with the pair $\bigl( (\theta + \theta_{\rm new})/2, \delta/2 \bigr)$. After finitely many steps, the new pair $(\theta_{\rm new}, \delta)$ will satisfy $L(\theta_{\rm new}) \ge L(\theta) + \delta/3$ and $\delta = DL(\theta, \theta_{\rm new} - \theta) > 0$. In the pseudocode provided later, this Newton--Armijo--Goldstein step is abbreviated as ``$(\theta_{\rm new},\delta) \leftarrow \text{Newton}(\theta, D)$''.

\subsubsection*{Basic step 2: Modification of $\theta$ or reduction of $D$}

Having computed a new proposal $\theta_{\rm new}$ as in basic step~1, where $\delta = DL(\theta, \theta_{\rm new} - \theta) > 0$, we first check whether it belongs to $\Theta$ or at least satisfies
\[
	(1 - t) \theta + t \theta_{\rm new} \ \in \ \Theta
	\quad\text{for some} \ t > 0 .
\]
If we represent $\theta$ and $\theta_{\rm new}$ as in \eqref{eq:representation.theta} with coefficients $\alpha_0, \alpha_1, \beta_\tau$ for $\theta$ and $\alpha_{0,{\rm new}}, \alpha_{1,{\rm new}}, \beta_{\tau,{\rm new}}$ for $\theta_{\rm new}$, then the latter requirement is satisfied if
\begin{equation}
\label{eq:valid.direction}
	\beta_{\tau,{\rm new}} > 0
	\quad\text{whenever} \ \tau \in D \setminus D(\theta) .
\end{equation}
If \eqref{eq:valid.direction} is violated, we leave $\theta$ unchanged, but we replace $D$ with $D \setminus \{\tau_o\}$, where $\tau_o$ is an index in $D \setminus D(\theta)$ such that $\beta_{\tau_o,{\rm new}}$ is minimal.
If \eqref{eq:valid.direction} is satisfied, we perform a second step size correction and replace $\theta$ with $(1 - t_o) \theta + t_o \theta_{\rm new}$, where $t_o \in (0,1]$ is the largest number such that the latter convex combination belongs to $\Theta$. An explicit expression for $t_o$ is given by
\[
	t_o \
	:= \ \max \bigl\{ t \in (0,1] :
		(1 - t)\theta + t \theta_{\rm new} \in \Theta \bigr\}
	\ = \ \min \Bigl( \{1\} \cup
	 	\Bigl\{ \frac{\beta_\tau}{\beta_\tau - \beta_{\tau,{\rm new}}}
			: \tau \in D(\theta), \beta_{\tau,{\rm new}} < 0 \Bigr\} \Bigr) .
\]
In addition, we then replace $\theta$ with $\theta - c$ for some constant $c$ such that $f_\theta$ defines a probability density. Finally, we replace $D$ with $D(\theta)$ for the modified candidate $\theta$. Note that $L(\theta)$ increases strictly, and in case of $t_o < 1$, the new set $D$ is a proper subset of the previous set $D$.

All in all, we obtain a new pair $(\theta,D)$ such that $L(\theta)$ has increased strictly or $D$ is a proper subset of the former set $D$. Moreover, the new $\theta$ differs from the previous one if and only if the new value $L(\theta)$ is strictly larger than the previous one. In the pseudocode provided later, this whole modification of $(\theta,D)$ is written as ``$(\theta,D) \leftarrow \text{StepForward}(\theta,D,\theta_{\rm new})$''.

\subsubsection*{Local search}

If we start from a pair $(\theta,D)$ with $\theta \in \Theta \cap \V$, $L(\theta) > -\infty$ and $D \supset D(\theta)$, a local search means to iterate basic steps~1 and 2 with a certain threshold $\delta_{\rm Newton} \ge 0$ as follows:
\[
	\begin{array}{l}
	\hline
	(\theta_{\rm new},\delta) \ \leftarrow \ \text{Newton}(\theta, D) \\
	\text{while} \ \delta > \delta_{\rm Newton} \ \text{do} \\
	\strut\quad (\theta,D)
		\leftarrow \text{StepForward}(\theta,D,\theta_{\rm new}) \\
	\strut\quad (\theta_{\rm new},\delta)
		\leftarrow \text{Newton}(\theta,D,\theta_{\rm new}) \\
	\text{end while} \\
	\hline
	\end{array}
\]
Imagine for the moment that $\delta_{\rm Newton} = 0$. After finitely many iterations, the set $D$ would remain unchanged and be equal to $D(\theta)$, while the first assignment within the while-loop would amount to $\theta \leftarrow \theta_{\rm new}$. That means, eventually, a local search leads to a standard Newton procedure and a locally optimal function $\theta$.

Note also that after finitely many steps, $L(\theta)$ is strictly larger than the original value unless the starting point $\theta$ was already locally optimal while the set $D \supsetneq D(\theta)$ has been chosen poorly in the sense that basic step~2 leaves $\theta$ unchanged and results in a stepwise reduction of $D$ until $D = D(\theta)$ again.

In practice, of course, we run a local seach with a small threshold $\delta_{\rm Newton} > 0$. The resulting $\theta$ is called \textsl{almost locally optimal}.

\subsubsection*{Basic step 3: Deactivating constraints}

Suppose that $\theta \in \Theta \cap \V$ is (almost) locally optimal, but \eqref{eq:optimality.2} is violated. More precisely, suppose that $\max_{\tau \in \DD} h_\theta(\tau)$ is strictly larger than a given threshold $\delta_{\rm Knot} \ge 0$. Then we choose a nonempty finite set $D_o \subset \DD \setminus D(\theta)$ such that
\begin{equation}
\label{ineq:theta.not.yet.optimal}
	h_\theta(\tau_o) \ > \ \delta_{\rm Knot}
	\quad\text{for all} \ \tau_o \in D_o .
\end{equation}
Thereafter we start a new local search with $D = D(\theta) \cup D_o$.

The obvious question is whether such a choice of $D$ is reasonable. It may happen that during the first iterations of the local search, $\theta$ remains unchanged while elements of the set $D_o$ are removed again. But eventually, at least one of its elements will be retained and $\theta$ will be modified. To prove this claim, we write $\theta_{\rm new} = \theta + v + \sum_{\tau_o \in D_o} \beta_{\tau_o,{\rm new}} V_{\tau_o}$ with some function $v \in \V_{D(\theta)}$. Then it follows from \eqref{eq:optimality.1} that
\[
	0 \ < \ DL(\theta, \theta_{\rm new} - \theta)
	\ = \ DL(\theta, v)
		+ \sum_{\tau_o \in D_o} \beta_{\tau_o,{\rm new}} DL(\theta, V_{\tau_o})
	\ = \ \sum_{\tau_o \in D_o} \beta_{\tau_o,{\rm new}} h_\theta(\tau_o) ,
\]
and because of \eqref{ineq:theta.not.yet.optimal}, at least one coefficient $\beta_{\tau_o,{\rm new}}$, $\tau_o \in D_o$, has to be strictly positive. Consequently, starting a local search with this choice of $D$ yields a strict improvement of $L(\theta)$ after at most $\# D_o$ iterations.

Our explicit construction of $D_o$ depends on the current set $D(\theta)$ and follows essentially the proposal of \cite{Liu_Wang_2018}. Suppose first that $D(\theta) = \emptyset$. Then we choose $D_o = \{\tau_o\}$ with a point $\tau \in \DD$ such that $h_\theta(\tau_o) = \max_{\tau \in \DD} h_\theta(\tau)$. Otherwise, let $\tau_1 < \cdots < \tau_m$ be the $m \ge 1$ different elements of $D(\theta)$. With $\tau_0 := -\infty$ and $\tau_{m+1} := \infty$, we set $\DD_j := \DD \cap (\tau_j, \tau_{j+1})$. For each $0 \le j \le m$ with $\DD_j \ne \emptyset$, we determine a point $\tau_o \in \argmax_{\tau \in \DD_j} h_\theta(\tau)$. If $h_\theta(\tau_o)$ is greater than both $\delta_{\rm Knot}$ and $10^{-3} \max_{\tau \in \DD} h_\theta(\tau)$, then $\tau_o$ is added to $D_o$. The latter condition on $h_\theta(\tau_o)$ prevents us from deactivating too many constraints early on, which would increase the dimensionality unnecessarily.

All in all, basic step~3 amounts to a procedure ``$(h_o,D_o) \leftarrow \text{NewKnots}(\theta)$''. It returns $h_o := \max_{\tau \in \DD} h_\theta(\tau)$ and, in case of $h_o > \delta_{\rm Knot}$, a nonempty finite set $D_o \subset \DD$ such that $DL(\theta,V_{\tau_o}) > \max(10^{-3} h_o, \delta_{\rm Knot})$ for all $\tau_o \in D_o$.

\subsubsection*{Explicit maximisation of $h_\theta$}

In Setting~1, maximizing $h_\theta$ over subsets of $\DD$ is straightforward, because $\DD$ is finite. In Settings~2A-B, suppose that $\theta \in \Theta \cap \V$ is (almost) locally optimal, and that $P_\theta(dx) := e_{}^{\theta(x)} \, P_o(dx)$ defines a probability measure on $\XX$. Here,
\[
	h_\theta(\tau) \ = \ \int V_\tau \, d(\hat{P} - P_\theta)
	\ = \ \int (x - \tau)^+ \, (\hat{P} - P_\theta)(dx) .
\]
Note that for any probability measure $Q$ on $\R$ with $\int |x| \, Q(dx) < \infty$ and $\tau \in \R$,
\[
	H_Q(\tau) \ := \ \int (x - \tau)^+ \, Q(dx)
\]
defines a convex and non-increasing function $H_Q : \R \to [0,\infty)$ with derivatives
\begin{align*}
	H_Q'(\tau \, -) \ &= \ - Q([\tau,\infty)) \ = \ Q((-\infty,\tau)) - 1 , \\
	H_Q'(\tau \, +) \ &= \ - Q((\tau,\infty)) \ = \ Q((-\infty,\tau]) - 1 .
\end{align*}
Hence $h_\theta = H_{\hat{P}} - H_{P_\theta}$ is a Lipschitz-continuous function on $\R$ with derivatives
\[
	h_\theta'(\tau\,\pm)
	\ = \ \hat{F}(\tau\,\pm) - F_\theta(\tau) ,
\]
where $\hat{F}$ and $F_\theta$ denote the cumulative distribution functions of $\hat{P}$ and $P_\theta$, respectively. Note that $\hat{F}$ is constant on the intervals $(-\infty,x_1)$, $[x_1,x_2)$, \ldots, $[x_{n-1},x_n)$, $[x_n,\infty)$ whereas $F_\theta$ is continuous on $\R$ and strictly increasing on $\XX$. Consequently,\\[1ex]
(i) \ $h_\theta$ is strictly concave on each interval $[x_i,x_{i+1}]$, $1 \le i < n$,\\[0.5ex]
(ii) \ $h_\theta$ is concave and non-increasing on $(-\infty,x_1]$,\\[0.5ex]
(iii) \ $h_\theta$ is concave and non-decreasing on $[x_n,\infty)$ with $\lim_{\tau \to \infty} h_\theta(\tau) = 0 > h_\theta(x_n)$.\\[1ex]
The limit in (iii) follows from dominated convergence together with the fact that $(x - x_n)^+ \ge (x - \tau)^+ \to 0$ as $x_n \le \tau \to \infty$. The strict inequality for $h_\theta(x_n)$ follows from $\hat{P}((x_n,\infty)) = 0 < P_\theta((x_n,\infty))$. Hence any $\tau$ with $h_\theta(\tau) > 0$ has to satisfy $\tau < x_n$.

In Setting~2A one may even conclude from local optimality of $\theta$ that\\[1ex]
(ii') \ $h_\theta$ is concave and non-increasing on $(-\infty,x_1]$ with limit $\lim_{\tau \to - \infty} h_\theta(\tau) = 0 > h_\theta(x_1)$,\\[1ex]
because $\int (x - \tau) \, (\hat{P} - P_\theta)(dx) = 0$, so the equality $(x - \tau)^+ = x - \tau + (\tau - x)^+$ leads to the alternative representation $h_\theta(\tau) = \int (\tau - x)^+ \, (\hat{P} - P_\theta)(dx)$. Consequently, it suffices to search for local maximizers of $h_\theta$ on $(x_1,x_n)$.

In Setting~2B, (ii) implies that the maximizer of $h_\theta$ on $[0,x_1]$ is $0$. Hence it suffices to search for local maximizers of $h_\theta$ on $\{0\} \cup (x_1,x_n)$.

If we want to maximize $h = h_\theta$ on an interval $[a,b] = [x_i,x_{i+1}]$ for some $1 \le i < n$, we could proceed as follows: First we check whether $h'(a\,+) \le 0$ or $h'(b\,-) \ge 0$. In these cases, $h(a) = \max_{\tau \in [a,b]} h(\tau)$ or $h(b) = \max_{\tau \in [a,b]} h(\tau)$, respectively. In case of $h'(a\,+) > 0 > h'(b\,-)$, we determine the unique point $\tau \in (a,b)$ satisfying $h_\theta'(\tau) = 0$. In general, this leads to a numerical approximation of $\tau$, but in our specific examples for Settings~2A-B, $\tau$ may be computed explicitly by means of the standard Gaussian or gamma quantile functions, see Sections~\ref{subsec:Details.2A} and \ref{subsec:Details.2B}.

\subsubsection*{Finding a starting point $\theta$}

One possibility to determine a starting point $\theta$ is to activate all constraints initially and find an optimal function in $\V_{\emptyset} \subset \Theta$. In Setting~2A, we are then looking for a function $\theta(x) = \hat{\kappa}x - c(\hat{\kappa})$ with $c(\kappa) := \log \int_{\XX} e^{\kappa x} \, P_o(dx)$, and $\hat{\kappa} \in \R$ is the unique real number such that $c'(\hat{\kappa}) = \hat{\mu}$. Specifically, if $P_o = \NN(0,1)$, then $c(\kappa) = \kappa^2/2$, whence $\hat{\kappa} = \hat{\mu}$.

In Setting~2B, activating all constraints would lead to the trivial space $\V_{\emptyset} = \{0\}$. Alternatively, one could determine an optimal function in $\V_{\{0\}} \cap \Theta$. With $\hat{\kappa}$ as before, i.e.\ $c'(\hat{\kappa}) = \hat{\mu}$, the optimal function $\theta$ is given by $\theta(x) = \hat{\kappa}^+ x - c(\hat{\kappa}^+)$. Specifically, if $P_o = \mathrm{Gamma}(\alpha,\beta)$, then $c(\kappa) = - \alpha \log((1 - \kappa/\beta)^+)$, so that $\hat{\kappa} = \beta - \alpha/\hat{\mu}$.

All in all, for Settings~2A-B we obtain a starting point $\theta \in \Theta$ depending only on $\hat{\mu}$ which is locally optimal, indicated as ``$\theta \leftarrow \text{Start}(\hat{\mu})$''.

In Setting~1, finding an optimal function in $\V_{\emptyset}$ would amount to solving a nonlinear equation numerically. Alternatively, we start with the MLE $\theta$ of a Gaussian log-density up to an additive constant, i.e.
\[
	\theta_0(x) \
	:= \ - (x - \hat{\mu})^2/(2 \hat{\sigma}^2)
\]
with $\hat{\mu} = \sum_{i=1}^n w_i x_i$ and $\hat{\sigma}^2 := \sum_{i=1}^n w_i (x_i - \hat{\mu})^2$. Next we fix a nonempty set $D_0 \subset \DD$ and replace $\theta_0$ with the unique linear spline $\theta \in \V_{D_0}$ such that $\theta \equiv \theta_0$ on $D_0 \cup \{x_1,x_n\}$. Then we normalize it via $\theta \leftarrow \theta - \log \bigl( \int_{x_1}^{x_n} e^{\theta(x)} \, dx \bigr)$. All these operations are hidden behind ``$\theta \leftarrow \text{Start}(\hat{\mu},\hat{\sigma},D_0)$'' in the subsequent pseudocode. Note that this starting point $\theta$ is not locally optimal in general.

\subsection{Complete algorithms}
\label{subsec:Complete.algorithms}

In Settings~2A-B, where a locally optimal starting point is easily found, our complete algorithm works as follows:
\[
	\begin{array}{l}
	\hline
	\theta \leftarrow \text{Start}(\hat{\mu}) \\
	(h_o,D_o) \leftarrow \text{NewKnots}(\theta) \\
	\text{while} \ h_o > \delta_{\rm Knot} \ \text{do} \\
	\strut\quad
		D \leftarrow D(\theta) \cup D_o \\
	\strut\quad \text{\sl \# Local search:} \\
	\strut\quad
		(\theta_{\rm new},\delta) \leftarrow \text{Newton}(\theta,D) \\
	\strut\quad
		\text{while} \ \delta > \delta_{\rm Newton} \ \text{do} \\
	\strut\quad\quad
			(\theta,D) \leftarrow \text{StepForward}(\theta,D,\theta_{\rm new}) \\
	\strut\quad\quad
			(\theta_{\rm new},\delta) \leftarrow \text{Newton}(\theta,D) \\
	\strut\quad
		\text{end while} \\
	\strut\quad \text{\sl \# Check global optimality:} \\
	\strut\quad	(h_o,D_o) \leftarrow \text{NewKnots}(\theta) \\
	\text{end while} \\
	\hline
	\end{array}
\]
In Setting~1, our algorithm has a slightly different beginning, because the starting point $\theta$ is not locally optimal:
\[
	\begin{array}{l}
	\hline
	\theta \leftarrow
		\text{Start}(\hat{\mu},\hat{\sigma},D_0) \\
	(D_o,h_o) \leftarrow (\emptyset,\infty) \\
	\text{while} \ h_o > \delta_{\rm Knot} \ \text{do} \\
	\strut\quad \ldots \\
	\text{end while} \\
	\hline
	\end{array}
\]

Note that in Setting~1, an affine transformation $x \mapsto a + bx$ of our data with $b > 0$ would result in new directional derivatives $DL(\theta,V_{\tau_o})$ which differ from the original values by this factor $b$. By way of contrast, the output $\delta$ of $\text{Newton}(\theta,D)$ is invariant under such transformations. Hence it is advisable to distinguish the stopping thresholds $\delta_{\rm Newton}$ and $\delta_{\rm Knots}$, where $\delta_{\rm Knot} > 0$ is chosen to be a small constant times $\hat{\sigma}$. In Settings~2A-B the parameter $\delta_{\rm Knot}$ should reflect the spread of the reference distribution $P_o$.

\subsection{Convergence}
\label{subsec:Convergence}

After circulating a first version of the present paper, \cite{Sommer-Simpson_2019} provided a proof of convergence of our algorithm in Setting~2B. Lemma~\ref{lem:convergence} below implies that in all three settings, the output of our algorithm is arbitrarily close to $\hat{\theta}$ if $\delta_{\rm Knot}$ and $\delta_{\rm Newton}$ are sufficiently small. Our proof generalizes and simplifies the arguments of \cite{Sommer-Simpson_2019}.

To formulate the result, let $\theta \in \Theta \cap \V$ with $L(\theta) > -\infty$. To check local optimality of $\theta$, we perform a Newton step for $L$ on the parameter space $\V_{D(\theta)}$. This yields a function $\theta_{\rm new} \in \V_{D(\theta)}$ maximizing a second order Taylor approximation of $L$ on $\V_{D(\theta)}$ and the directional derivative
\[
	\delta_{\rm Newton}(\theta) \ := \ DL(\theta, \theta_{\rm new} - \theta) .
\]
In our algorithm, $\theta$ is viewed as approximately locally optimal if $\delta_{\rm Newton}(\theta)$ is smaller than a given number $\delta_{\rm Newton}$. If that is the case, we check whether
\[
	\delta_{\rm Knot}(\theta) \ := \ \max_{\tau \in \DD} \, DL(\theta,V_{\tau,\theta})
\]
is smaller than a given number $\delta_{\rm Knot}$. Note also that during our algorithm the value $L(\theta)$ never decreases.

\begin{Lemma}
\label{lem:convergence}
In all Settings and for any constant $L_o \in (-\infty,L(\hat{\theta}))$, there exist constants $C_{\rm Newton}$ and $C_{\rm Knot}$ such that for all $\theta \in \Theta \cap \V$ with $L(\theta) \ge L_o$,
\[
	L(\hat{\theta}) - L(\theta)
	\ \le \ C_{\rm Newton} \sqrt{\delta_{\rm Newton}(\theta)}
		+ C_{\rm Knot} \delta_{\rm Knot}(\theta) .
\]
\end{Lemma}

\begin{Remark}
\label{rem:from.L.to.theta}
For $\theta \in \Theta \cap \V$, it follows from $L(\theta) \to L(\hat{\theta})$ that $\|\theta - \hat{\theta}\| \to 0$, where $\|\cdot\|$ is the norm in \eqref{eq:norm.on.V}.
\end{Remark}

\section{Numerical examples, simulations and an application}
\label{sec:numerical.examples}

\subsection{Comparisons in Setting~1}

An obvious question is how much better the new algorithm for Setting~1 is in comparison to the active set method of \cite{Duembgen_Rufibach_2011}. To enable a fair comparison, we implemented the latter method as follows:
\[
	\begin{array}{l}
	\hline
	\theta \leftarrow
		\text{Start}(\hat{\mu},\hat{\sigma},D_0) \\
	h_o \leftarrow \infty \\
	D \leftarrow D_0 \\
	\text{while} \ h_o > \delta_{\rm Knot} \ \text{do} \\
	\strut\quad
		(\theta_{\rm new},\delta) \leftarrow
			\text{Newton}(\theta,D) \\
	\strut\quad
		\text{while} \ \delta > \delta_{\rm Newton} \ \text{do} \\
	\strut\quad\quad
			\theta_{\rm new} \leftarrow \text{Newton}(\theta,D) \\
	\strut\quad\quad
			\theta_{\rm new} \leftarrow \text{Normalize}(\theta_{\rm new}) \\
	\strut\quad
		\text{end while} \\
	\strut\quad
		\text{if} \ \theta_{\rm new} \in \Theta \ \text{then}\\
	\strut\quad\quad
			\theta \leftarrow \theta_{\rm new} \\
	\strut\quad\quad
			(h_o,\tau_o) \leftarrow \text{NewKnot}(\theta) \\
	\strut\quad\quad
			D \leftarrow D(\theta) \cup \{\tau_o\} \\
	\strut\quad
		\text{else} \\
	\strut\quad\quad
			(\theta,D) \leftarrow \text{StepForward}(\theta,D,\theta_{\rm new}) \\
	\strut\quad
		\text{end if} \\
	\text{end while} \\
	\hline
	\end{array}
\]
Here $\theta \leftarrow \text{Normalize}(\theta)$ stands for replacing $\theta$ with $\theta - c$ such that $f_\theta$ defines a probability density. And ``$(h_o,\tau_o) \leftarrow \text{NewKnot}(\theta)$'' returns only one point $\tau_o \in \DD$ with maximal directional derivative $h_o = DL(\theta,\tau_o)$. This is the first main difference between the old and the new algorithm. The second main difference is that a full Newton procedure is run on $\V_D$ without checking and enforcing the shape constraint that $\theta \in \Theta$. An advantage of omitting the shape constraint is that the Newton search runs a bit faster. A disadvantage is that we sometimes iterate and optimize in a region far from $\Theta$, whereas in the subsequent $\text{StepForward}(\theta,D,\theta_{\rm new})$, only a rather small step is performed.

Concerning $D_0$, extensive numerical experiments showed that the choice $D_0 = \{x_{j(1)}, x_{j(2)}, x_{j(3)}\}$ with approximately equispaced indices $1 < j(1) < j(2) < j(3) < n$ is a good choice for a broad range of sample sizes $n$. With this choice, we simulated $200$ times a random sample of size $n$ from the standard Gaussian distribution and fitted a log-concave density with the old and the new method. Figure~\ref{fig:S1ratios} shows boxplots of the running time with the old method divided by the running time with the new method. One sees clearly, that the improvement is substantial, particularly for large sample sizes. It is similar in magnitude to the improvements reported by \cite{Wang_2018} for the algorithm of \cite{Liu_Wang_2018}. Table~\ref{tab:S1} reports the means of these relative efficiencies as well as the mean absolute running times. The methods have been implemented in pure R code, and the simulations have been performed on a MacBook Pro (2.6 GHz 6-Core Intel Core i7), the stopping thresholds being $\delta_{\rm Newton} = 10^{-7}/n$ and $\delta_{\rm Knot} = 10^{-7} \hat{\sigma}/n$.

\begin{figure}
\centering
\includegraphics[width=0.7\textwidth]{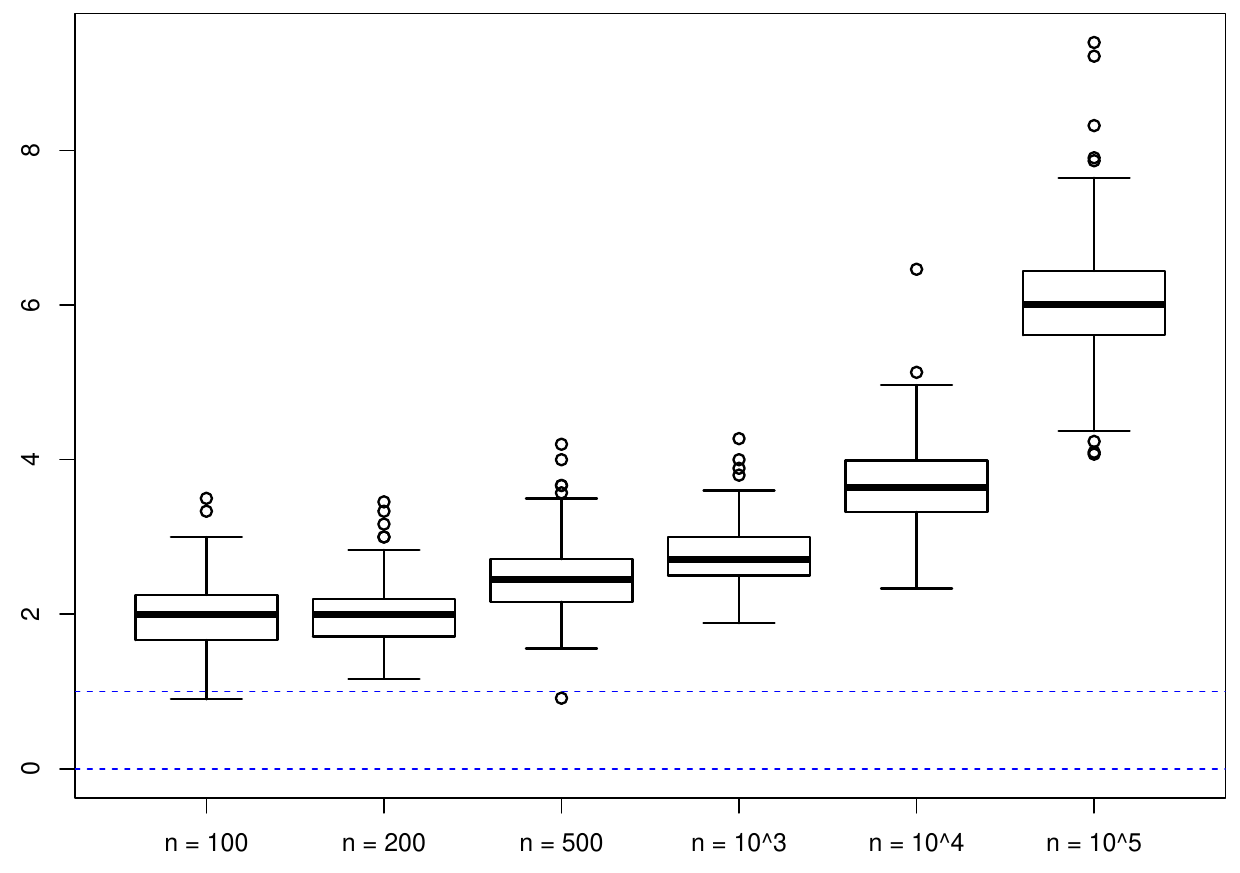}
\caption{Relative efficiencies of the new algorithm for Setting~1 with Gaussian samples.}
\label{fig:S1ratios}
\end{figure}

\begin{table}
\[
	\begin{array}{lcccccc}
	\hline
	\text{Sample size}
		& 100 & 200 & 500 & 10^3 & 10^4 & 10^5 \\
	\hline
	\text{Rel.\ efficiency}
		& 2.003 & 1.988 & 2.467 & 2.749 & 3.615 & 6.067 \\
	\hline
	\text{Running time (s)}
		& 4.425 \cdot 10^{-3}
		& 6.310 \cdot 10^{-3}
		& 8.450 \cdot 10^{-3}
		& 0.0115
		& 0.1029
		& 1.4805 \\
	\hline
	\end{array}
\]
\caption{Mean relative efficiencies and running times of the new algorithm for Setting~1 with Gaussian samples.}
\label{tab:S1}
\end{table}

\subsection{Numerical examples for Settings~2A-B}

\paragraph{Setting~2A.}
Inspired by the testing problem described in Section~\ref{subsec:GoF.test}, we simulated $n = 400$ independent observations $X_i$ with distribution $P_i = \NN(0,1)$ for $i > 20$ and $P_i = \NN(1.5,1)$ for $i \le 20$. With the reference distribution $P_o = \NN(0,1)$, the corresponding log-density ratio equals
\[
	\theta(x) \ = \ \log \frac{dP}{dP_o}(x)
	\ = \ \log(0.95 + 0.05 \,e^{1.5x-1.125} ) .
\]
The estimator $\hat{\theta}$ turned out to have $m = 5$ knots, and Figure~\ref{fig:HplotA} depicts the function
\[
	t \ \mapsto \ h(t) = DL(\hat{\theta},V_t) ,
\]
where the knots of $\hat{\theta}$ are indicated by vertical lines. As predicted by theory, $h(t) \le 0$ for all $t$ with equality in case of $t \in D(\hat{\theta})$. Figure~\ref{fig:ThetaplotA} depicts the true and estimated tail inflation functions $\theta$ and $\hat{\theta}$. Figure~\ref{fig:PplotA} shows the corresponding densities $p_o = \phi$, $p = e^\theta p_o$ and $\hat{p} = e^{\hat{\theta}} p_o$. Note that the estimator $\hat{p}$ captures the heavier right tail of $p$ in comparison to $p_o$. Applying the goodness-of-fit test described in Section~\ref{subsec:GoF.test} to this particular data set yielded a Monte-Carlo p-value smaller than $10^{-3}$ (with $10^5-1$ simulations) for the null hypothesis that all $400$ observations are standard Gaussian.

\begin{figure}
\centering
\includegraphics[width=0.7\textwidth]{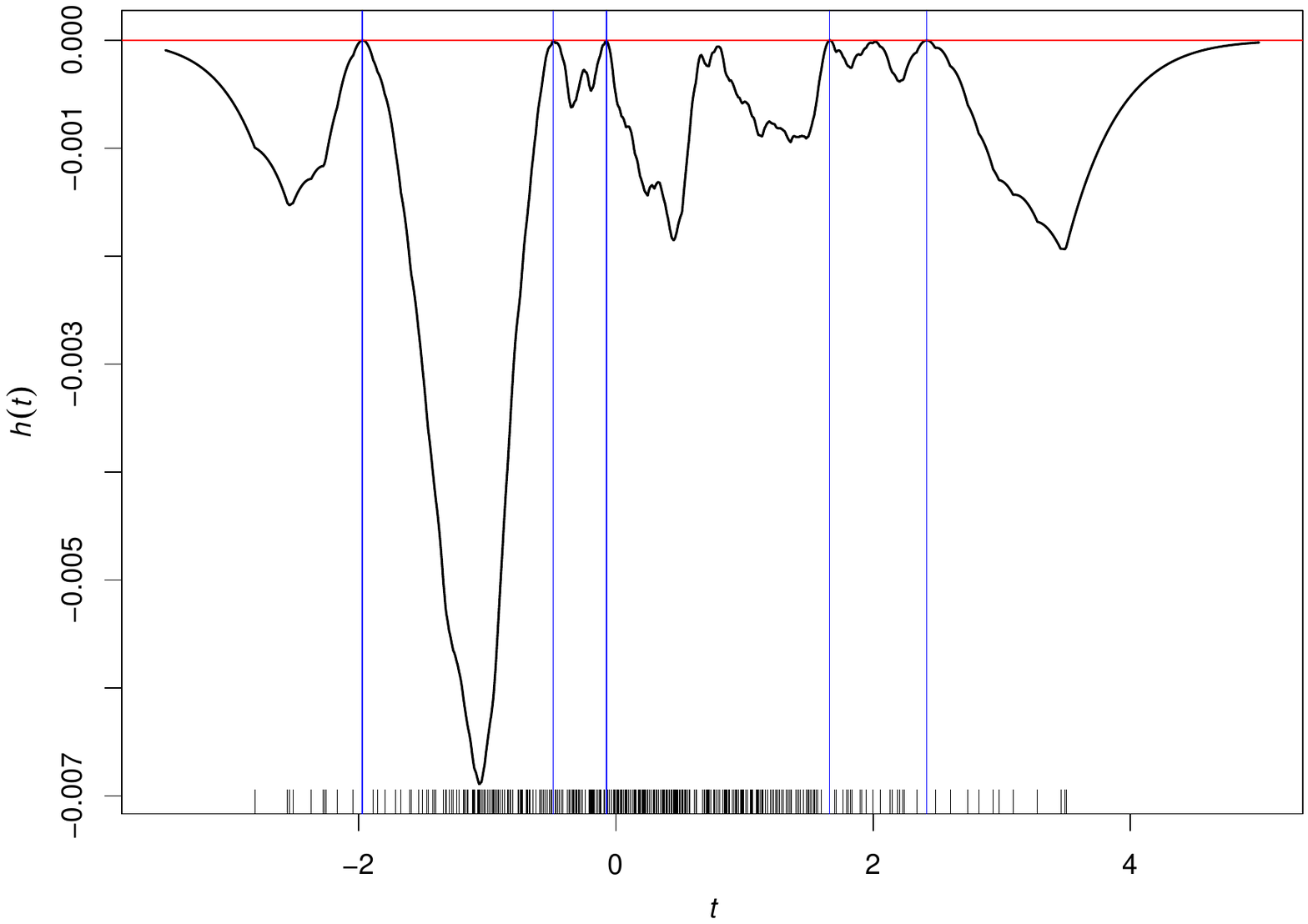}
\caption{Directional derivatives $h(t) = DL(\hat{\theta},V_t)$ for data example in Setting~2A.}
\label{fig:HplotA}
\end{figure}

\begin{figure}
\centering
\includegraphics[width=0.7\textwidth]{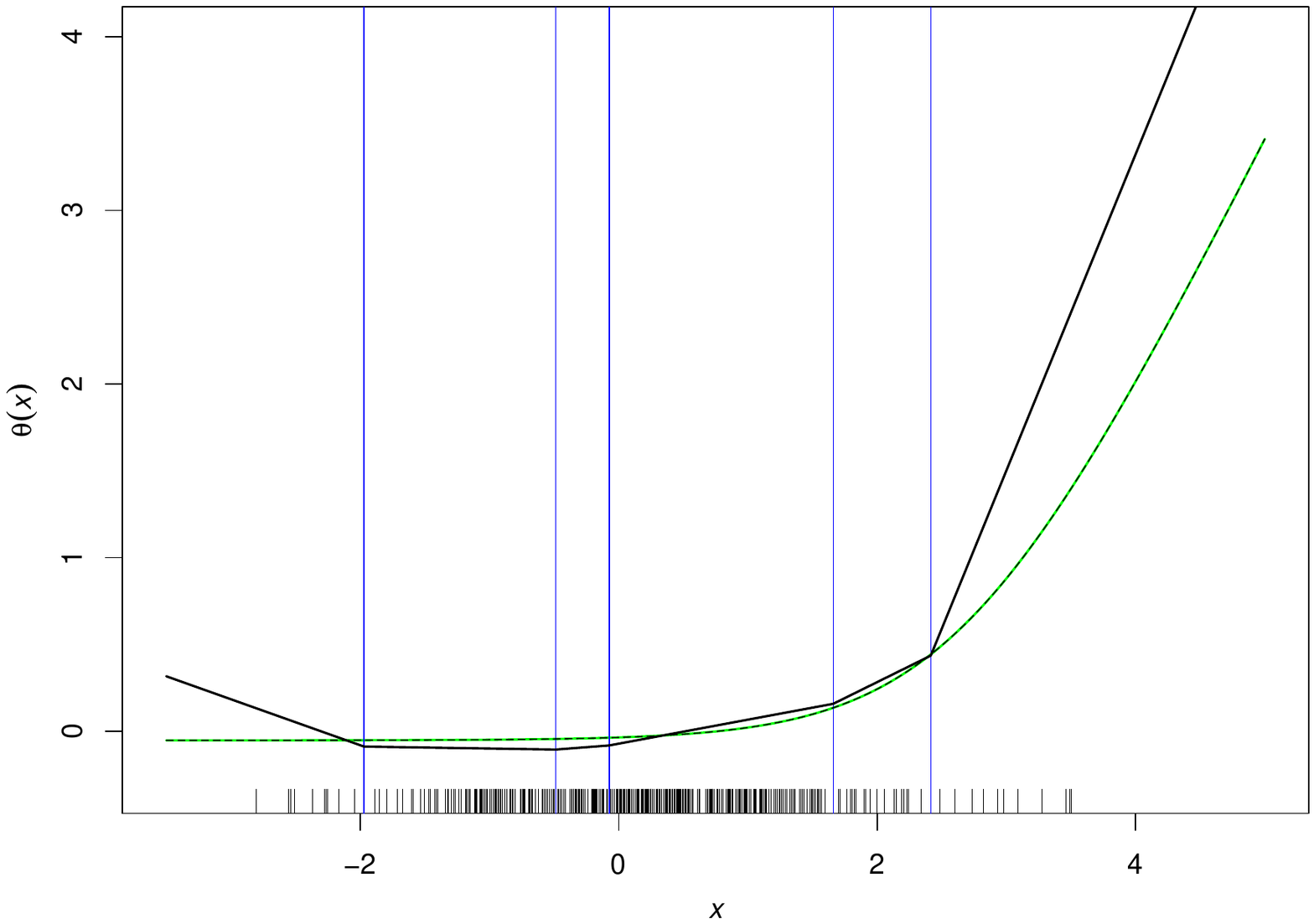}
\caption{True (green, dashed) and estimated (black) tail inflation functions $\theta$ and $\hat{\theta}$ for data example in Setting~2A.}
\label{fig:ThetaplotA}
\end{figure}

\begin{figure}
\centering
\includegraphics[width=0.7\textwidth]{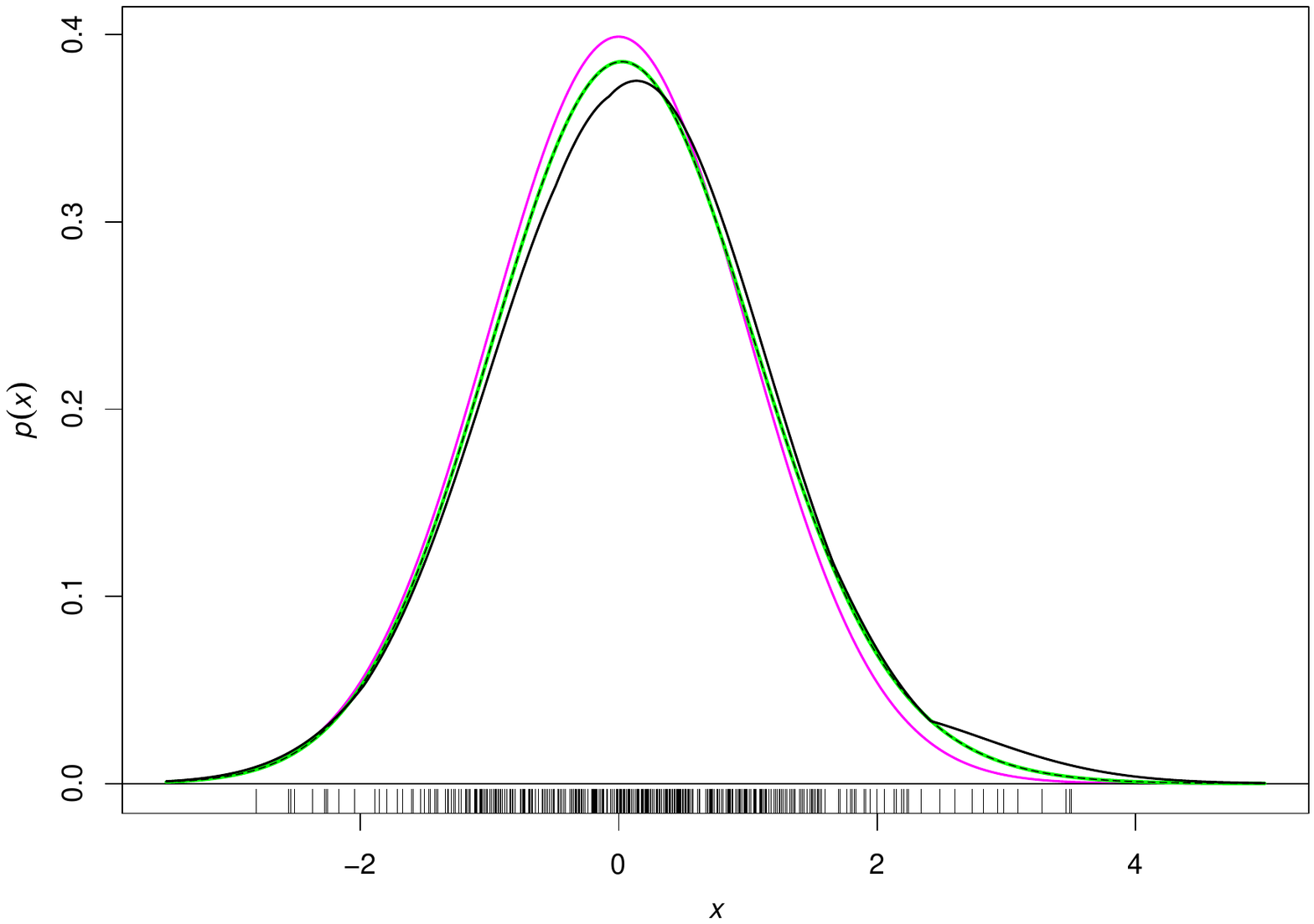}
\caption{Lebesgue densities $p_o$ (magenta), $p = e^\theta p_o$ (green, dashed) and $\hat{p} = e^{\hat{\theta}} p_o$ (black) for data example in Setting~2A.}
\label{fig:PplotA}
\end{figure}

\paragraph{Setting 2B.}
We simulated $n=1000$ independent observations $X_i$ such that $X_i \sim \chi_1^2$ for $i > 200$, $X_i/1.4 \sim \chi_1^2$ for $100 < i \le 200$ and $X_i/2 \sim \chi_1^2$ for $i \le 100$. With the reference distribution $P_o = \chi_1^2$, the corresponding log-density ratio equals
\[
	\theta(x) 
	\ = \ \log \bigl( 8 + 1.4^{-1/2} e^{x/7} + 2^{-1/2} e^{x/4} \bigr) - \log 10 .
\]
The estimator $\hat{\theta}$ turned out to have $m=6$ knots, and Figures~\ref{fig:HplotB} and \ref{fig:ThetaplotB} are analogous to the displays for Setting~2A, showing the directional derivatives $h(\tau) = DL(\hat{\theta},V_\tau)$ and the log-density ratios $\theta, \hat{\theta}$, respectively. Applying the goodness-of-fit test described in Section~\ref{subsec:GoF.test} to this particular data set yielded a Monte-Carlo p-value of $10^{-5}$ (with $10^5-1$ simulations) for the null hypothesis that all $400$ observations have distribution $\chi_1^2$.

\begin{figure}
\centering
\includegraphics[width=0.7\textwidth]{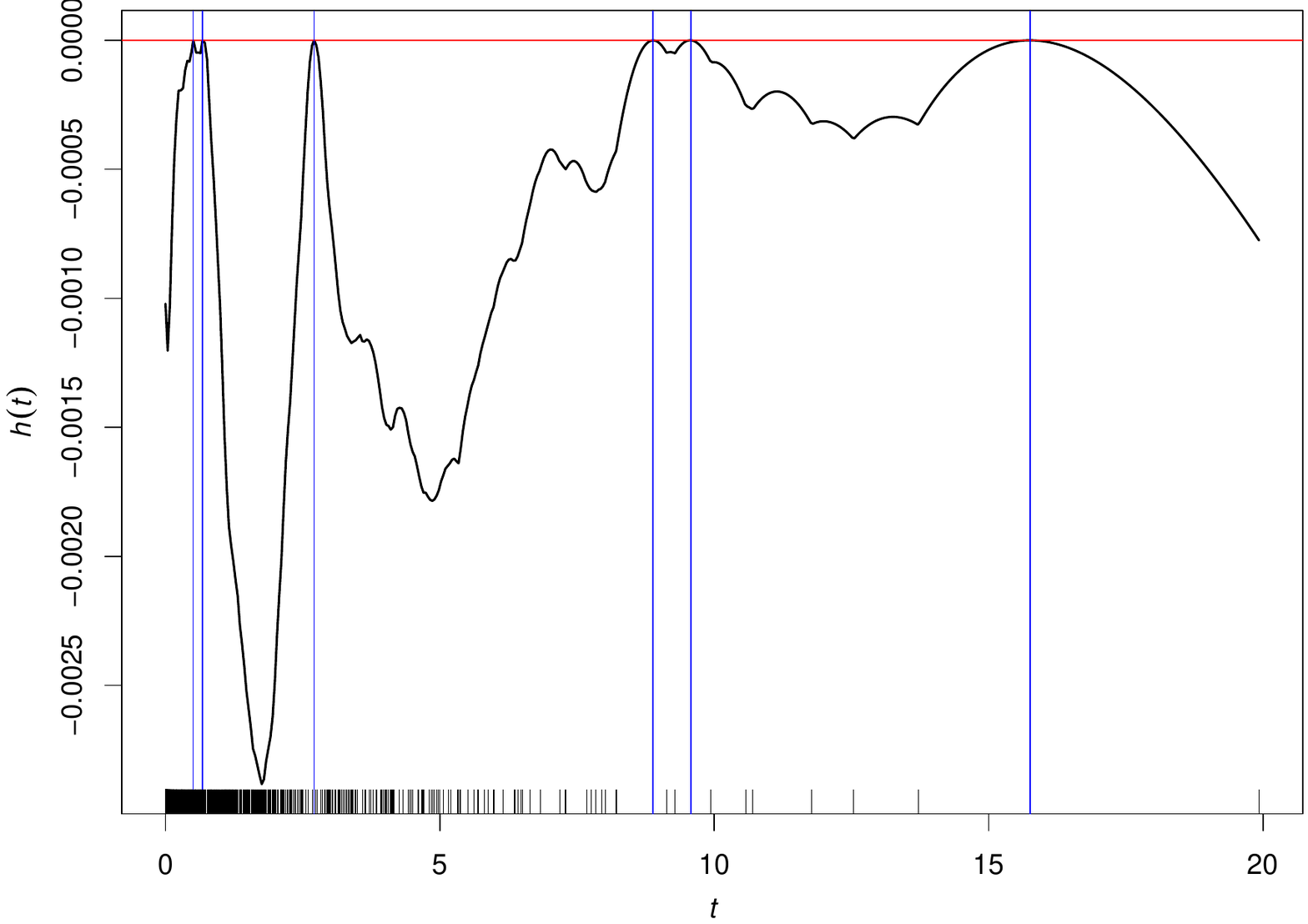}
\caption{Directional derivatives $h(t) = DL(\hat{\theta},V_t)$ for data example in Setting~2B.}
\label{fig:HplotB}
\end{figure}

\begin{figure}
\centering
\includegraphics[width=0.7\textwidth]{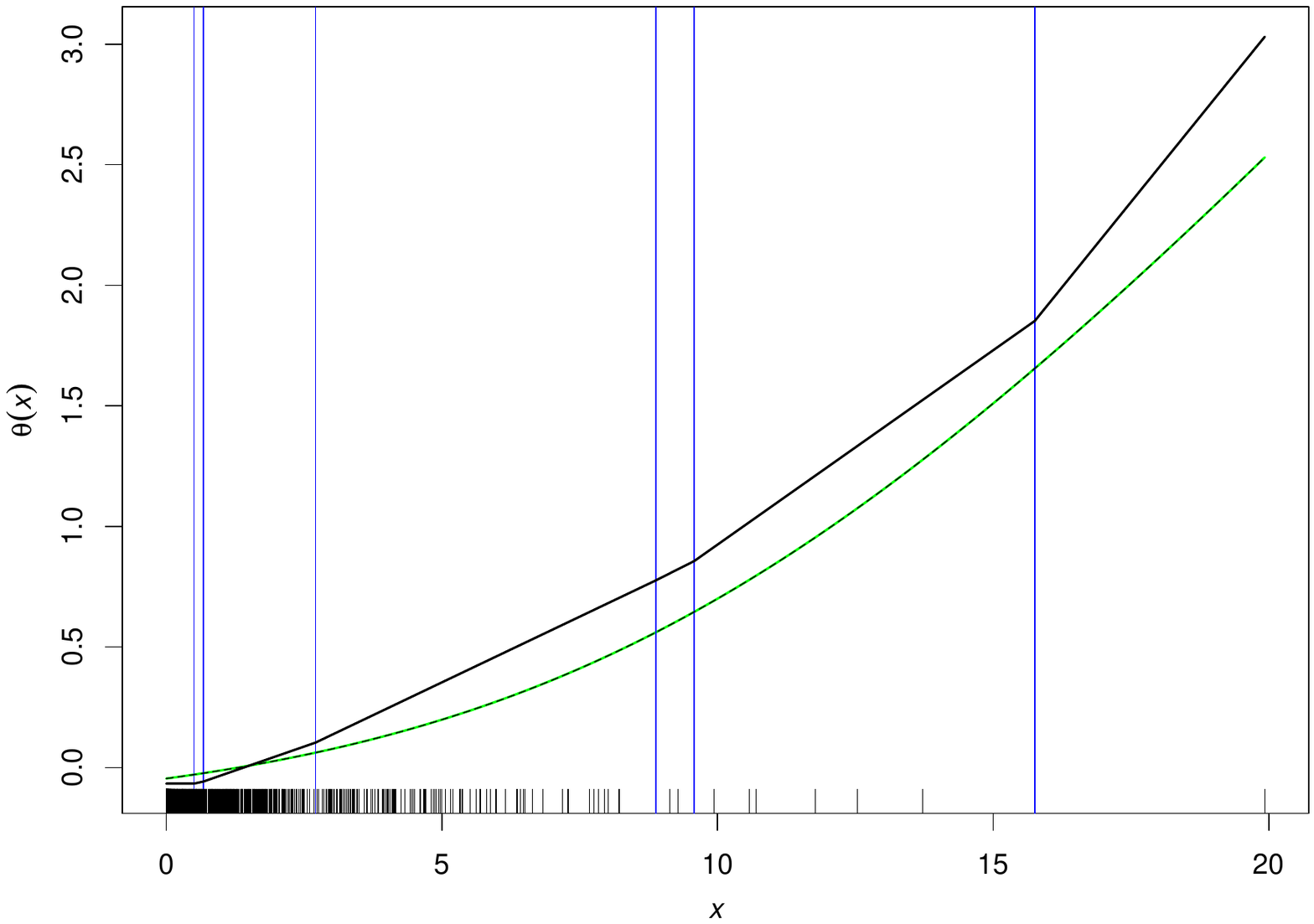}
\caption{True (green, dashed) and estimated (black) tail inflation functions $\theta$ and $\hat{\theta}$ for data example in Setting~2B.}
\label{fig:ThetaplotB}
\end{figure}

\subsection{Data-driven goodness-of-fit tests}
\label{subsec:GoF.test}

With the estimator $\hat{\theta}$ at hand, one may use the likelihood ratio statistic
\[
	T_{LR}(X_1,\ldots,X_n) \ := \ \sum_{i=1}^n \hat{\theta}(X_i)
\]
to test the null hypothesis that all distributions $P_i$ are equal to $P_o$ versus the alternative hypothesis that the marginal $P$ has a convex log-density $\theta \not\equiv 0$ with respect to $P_o$. Large values of $T_{LR}$ indicate a violation of the null hypothesis. The distribution of this test statistic under the null hypothesis is unknown but can be easily estimated via Monte Carlo simulations.

Specifically, consider Setting~2A with $P_o = \NN(0,1)$. As mentioned before, if each distribution $P_i$ and thus the marginal $P$ is a mixture of Gaussian distributions with standard deviation at least $1$, then $\theta = \log(dP/dP_o)$ is convex. This renders $T_{LR}$ an interesting alternative to higher criticism statistics as introduced by \cite{Donoho_Jin_2004} and \cite{Gontscharuk_etal_2016}. In the subsequent power simulations, we focus on a particular union-intersection test similar to those considered by the latter authors: With the order statistics $X_{(1)} < \cdots < X_{(n)}$ of the $X_i$, note that under $H_o$, the random variables $\Phi(X_{(1)}), \ldots, \Phi(X_{(n)})$ are distributed like the order statistics of a sample from the uniform distribution on $[0,1]$. In particular, $\Phi(X_{(i)})$ follows the beta distribution with parameters $i$ and $n+1-i$. Denoting the corresponding distribution function with $B_{i,n+1-i}$, a union-intersection test statistic of $H_o$ is given by
\begin{align*}
	T_{UI}(X_1,\ldots,X_n) \
	:= \ &\min \Bigl(
		\min_{i < (n+1)/2} B_{i,n+1-i}(\Phi(X_{(i)})),
		\min_{i > (n+1)/2} \bigl(1 - B_{i,n+1-i}(\Phi(X_{(i)})) \bigr) \Bigr) \\
	= \ &\min \Bigl(
		\min_{i < (n+1)/2} B_{i,n+1-i}(\Phi(X_{(i)})),
		\min_{i > (n+1)/2} B_{n+1-i,i}(\Phi(-X_{(i)})) \Bigr) ,
\end{align*}
small values indicating a violation of $H_o$. The rationale behind this test statistic is as follows: If $H_o$ is violated and $\theta$ is convex, then the left tail of $P$ is heavier than the one of $P_o$, leading to smaller order statistics $X_{(1)}, X_{(2)}, \ldots$, or the right tail of $P$ is heavier than the one of $P_o$, leading to larger order statistics $X_{(n)}, X_{(n-1)}, \ldots$. We also use the identity $1 - B_{i,n+1-i}(\Phi(x)) = B_{n+1-i,i}(\Phi(-x))$ for numerical reasons.

In a large simulation study involving different sample sizes $n$, we estimated the $(1 - \alpha)$-quantile of the null distribution of $T_{LR}(X_1,\ldots,X_n)$ and the $\alpha$-quantile of $T_{UI}(X_1,\ldots,X_n)$ in $10^{5} - 1$ Monte Carlo simulations, where $\alpha = 1\%, 5\%$. With these critical values, we estimated the power of the two tests at level $\alpha$ under the following distribution of the sample: For a fixed distribution $P_*$ on the real line and a subset $J \subset \{1,2,\ldots,n\}$ with $k \ge 0$ elements, the distributions $P_i$ of the random variables $X_i$ are given by
\[
	P_i \ = \ \begin{cases}
		P_* & \text{if} \ i \in J, \\
		P_o & \text{otherwise} .
	\end{cases}
\]
Specifically, we used $P_* = \NN(1.5, 1)$ and $P_* = \NN(0, 3)$. This setting is similar to the setting of \cite{Donoho_Jin_2004} with $P_i = (1 - k/n) P_o + (k/n) P_*$ for all $i$. The latter setting corresponds to a random set $J$ with $\# J$ having binomial distribution $\mathrm{Bin}(n, k/n)$.

For these two choices of $P_*$, Figures~\ref{fig:Power100}, \ref{fig:Power400} and \ref{fig:Power1000} show the power $\Pr(\text{reject} \ H_o \ \text{at level} \ \alpha)$ of both tests as a function of $k = \#J$. Clearly, the test based on $T_{LR}$ has higher power than the one based on $T_{UI}$. The difference in case of $P_* = \NN(0,3)$ is stronger than in case of the simple shift altervative $P_* = \NN(1.5, 1)$.

\begin{figure}
\strut \hfill $P_* = \NN(1.5, 1)$ \hfill\hfill $P_* = \NN(0,3)$ \hfill \strut

\includegraphics[width=0.45\textwidth]{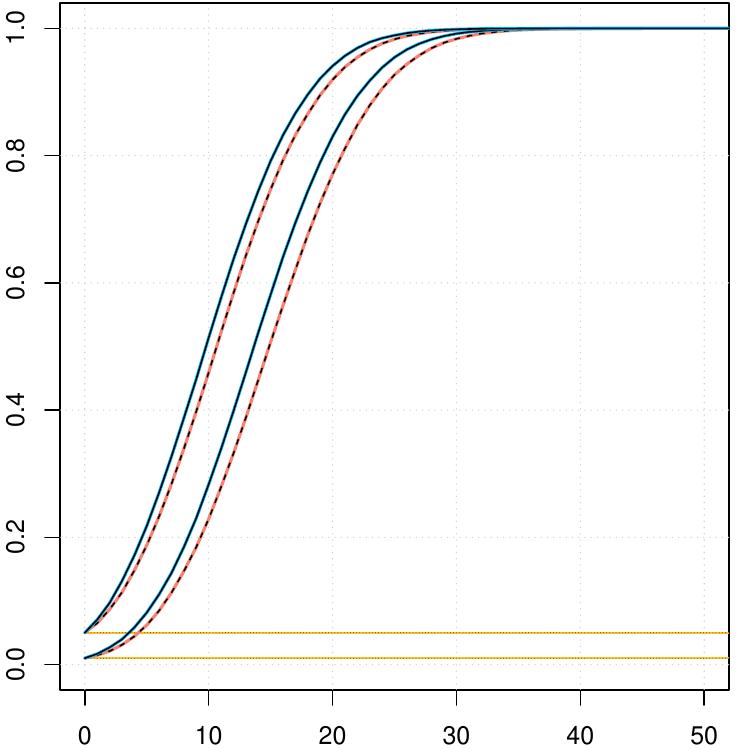}
\hfill
\includegraphics[width=0.45\textwidth]{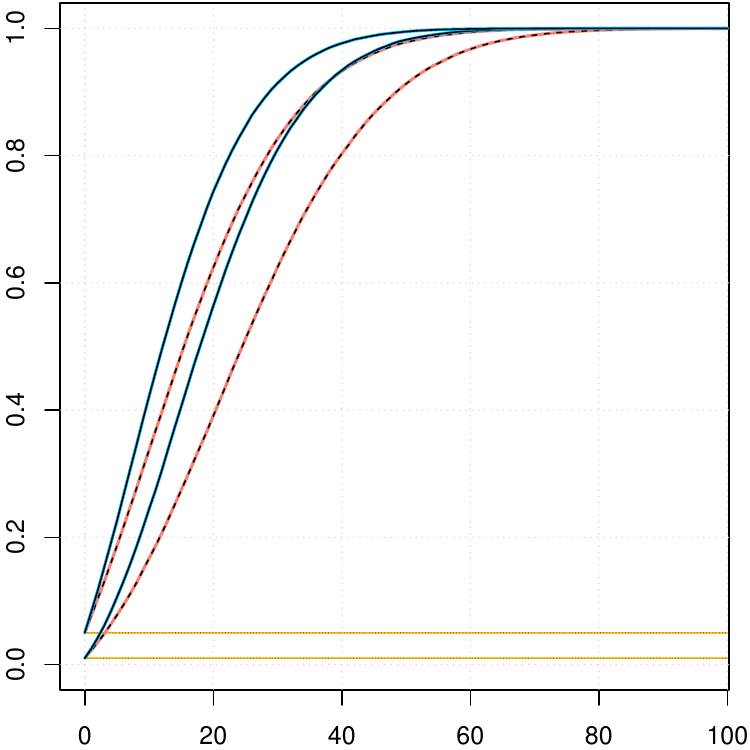}

\caption{Power of goodness-of-fit tests based on $T_{LR}$ (blue, solid) and $T_{UI}$ (red, dashed) as a function of $k$ for two distributions $P_*$ and sample size $n = 100$. The test levels $\alpha$ are $5\%$ and $1\%$.}
\label{fig:Power100}
\end{figure}

\begin{figure}
\strut \hfill $P_* = \NN(1.5, 1)$ \hfill\hfill $P_* = \NN(0,3)$ \hfill \strut

\includegraphics[width=0.45\textwidth]{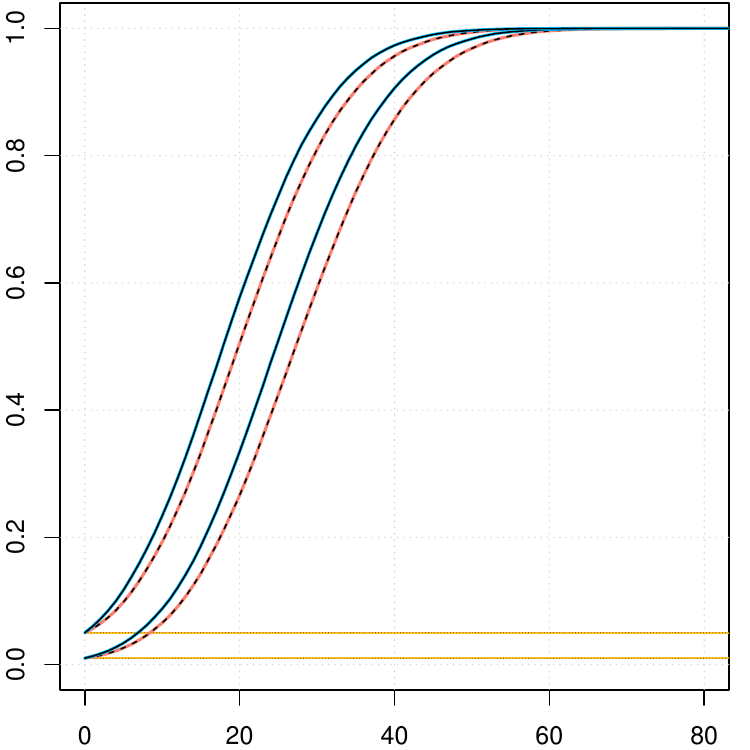}
\hfill
\includegraphics[width=0.45\textwidth]{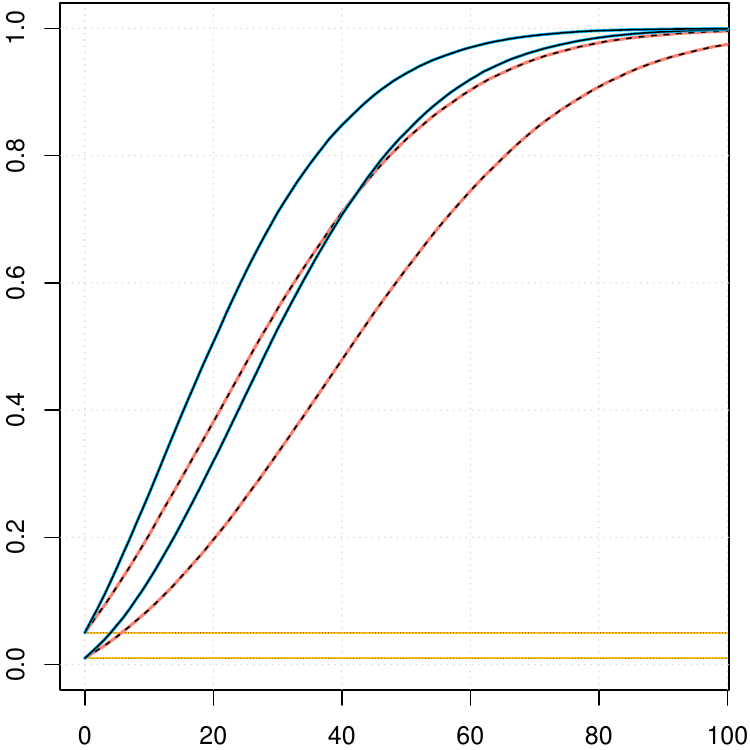}

\caption{Power comparison for sample size $n = 400$.}
\label{fig:Power400}
\end{figure}

\begin{figure}
\strut \hfill $P_* = \NN(1.5, 1)$ \hfill\hfill $P_* = \NN(0,3)$ \hfill \strut

\includegraphics[width=0.45\textwidth]{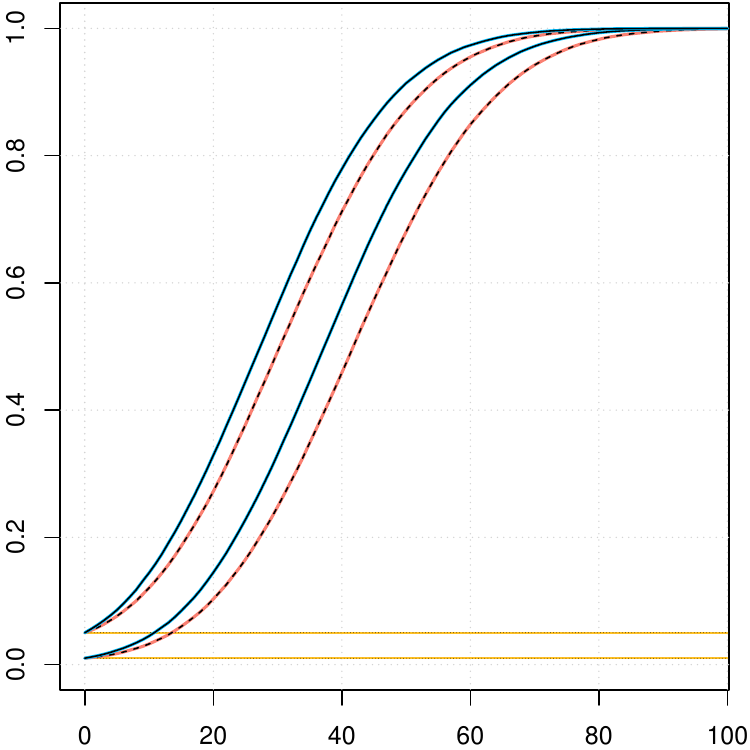}
\hfill
\includegraphics[width=0.45\textwidth]{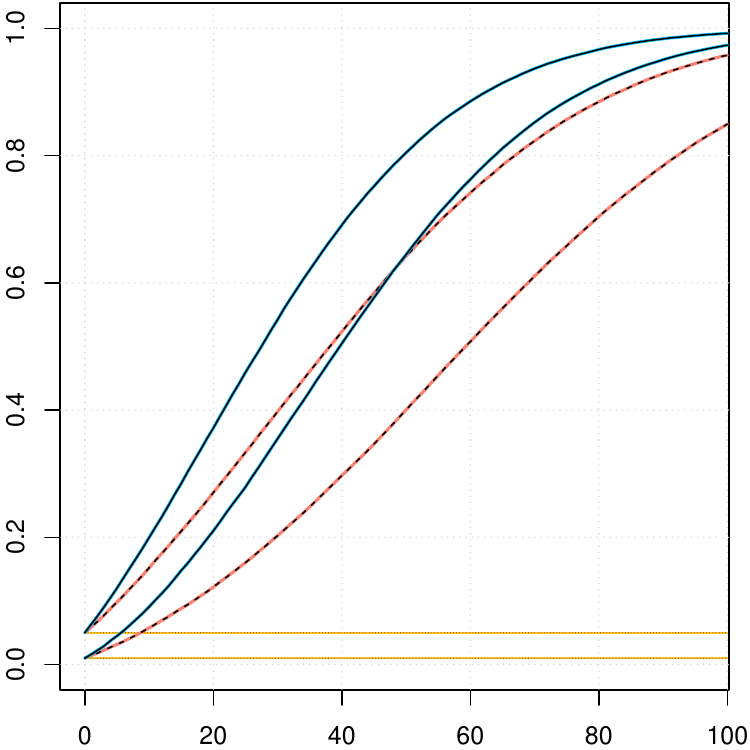}

\caption{Power comparison for sample size $n = 1000$.}
\label{fig:Power1000}
\end{figure}

Section~\ref{subsec:null.distributions} contains further information about the null distribution of $T_{LR}$ for different sample sizes and $P_o = \NN(0,1)$ or $P_o = \chi_1^2$.

\section{Proofs}
\label{sec:proofs}

An essential ingredient for the proof of Lemmas~\ref{lem:existence.uniqueness.1}, \ref{lem:existence.uniqueness.2A} and \ref{lem:existence.uniqueness.2B} is the following coercivity result.

\begin{Lemma}
\label{lem:coercivity}
Let $M$ be a measure on $\R$, and let $L(\theta) := \int \theta \, d\hat{P} - \int_{} e^\theta \, dM + 1$ for measurable functions $\theta : \R \to \R$.

\noindent
{\bf(a)} Suppose that $M(B) = \mathrm{Leb}(B \cap [x_1,x_n])$. Then for concave functions $\theta : \R \to \R$,
\[
	L(\theta) \ \to \ - \infty
	\quad\text{as}\quad
	\max_{x \in [x_1,x_n]} \, |\theta(x)| \ \to \ \infty .
\]

\noindent
{\bf(b)} Suppose that the three numbers $M((-\infty,x_1))$, $M([x_1,x_n])$ and $M((x_n,\infty))$ are strictly positive. Then for convex functions $\theta$,
\[
	L(\theta) \ \to \ - \infty
	\quad\text{as}\quad
	\max_{x \in [x_1,x_n]} \, |\theta(x)|
		+ \max \bigl\{ - \theta'(x_1\,-), \theta'(x_n\,+) \bigr\}
	\ \to \ \infty .
\]
\end{Lemma}

Part~(a) is known from \cite{Duembgen_etal_2011}, but for the reader's convenience and later reference, a simplified argument is also given here.

\begin{proof}[\bf Proof of Lemma~\ref{lem:coercivity}]
Let $i(\theta) := \min_{x \in [x_1,x_n]} \theta(x)$, $s(\theta) := \max_{x \in [x_1,x_n]} \theta(x)$ and $r(\theta) := s(\theta) - i(\theta)$.

As to part~(a), note first that
\[
	L(\theta) \ \le \ s(\theta) - e^{i(\theta)} (x_n - x_1) + 1
	\ = \ i(\theta) - e^{i(\theta)} (x_n - x_1) + r(\theta) + 1 .
\]
The right-hand side converges to $-\infty$ if either $s(\theta) \to - \infty$ or $i(\theta) \to \infty$ while $r(\theta)$ stays bounded. Thus it suffices to show that $L(\theta) \to -\infty$ as $r(\theta) \to \infty$. By concavity of $\theta$, the difference $\theta - i(\theta)$ is bounded from below on $[x_1,x_n]$ by a piecewise linear function with values in $[0,r(\theta)]$, and the value $0$ is attained at $x_1$ or at $x_n$. Hence, with $w_{\rm min} := \min(w_1,w_n)$, we may conclude that
\begin{align*}
	L(\theta) \
	&\le \ i(\theta) + (1 - w_{\rm min}) r(\theta)
		- e^{i(\theta)} \int_{x_1}^{x_n} e^{\theta(x) - i(\theta)} \, dx + 1 \\
	&\le \ i(\theta) + (1 - w_{\rm min}) r(\theta)
		- e^{i(\theta)} (x_n - x_1) \int_0^1 e^{r(\theta) t} \, dt + 1 \\
	&\le \ i(\theta) + (1 - w_{\rm min}) r(\theta)
		- e^{i(\theta)} (x_n - x_1) (e^{r(\theta)} - 1)/r(\theta) + 1 .
\end{align*}
For fixed $r(\theta)$, the maximum of the latter bound with respect to $i(\theta)$ equals
\[
	- \log(x_n - x_1) - \log(1 - e^{-r(\theta)})
		+ \log r(\theta) - w_{\rm min} r(\theta) ,
\]
and this converges to $-\infty$ as $r(\theta) \to \infty$.

As to part~(b), convexity of $\theta$ implies that either
\begin{equation}
\label{eq:boundaries.left}
	s(\theta) \ = \ \theta(x_1) \ > \ \theta(x_n) , \quad
	- \theta'(x_1\,-) \ \ge \ \frac{r(\theta)}{x_n - x_1}
	\quad\text{and}\quad
	\theta(x) \ \ge \ s(\theta) + \theta'(x_1\,-) (x - x_1) \ \
		\text{for} \ x \le x_1 ,
\end{equation}
or
\begin{equation}
\label{eq:boundaries.right}
	s(\theta) \ = \ \theta(x_n) \ \ge \ \theta(x_1) , \quad
	\theta'(x_n\,+) \ \ge \ \frac{r(\theta)}{x_n - x_1}
	\quad\text{and}\quad
	\theta(x) \ \ge \ s(\theta) + \theta'(x_n\,+) (x - x_n) \ \
		\text{for} \ x \ge x_n .
\end{equation}
Hence with $\XX_\ell := (-\infty,x_1)$ and $\XX_r := (x_n,\infty)$,
\[
	L(\theta)
	\ \le \ s(\theta)
		- e_{}^{s(\theta)} \min \bigl\{ M(\XX_\ell), M(\XX_r) \bigr\} + 1
	\ \to \ - \infty \quad \text{as} \ |s(\theta)| \to \infty ,
\]
because $M(\XX_\ell), M(\XX_r) > 0$. Moreover,
\[
	L(\theta)
	\ \le \ s(\theta) - e_{}^{s(\theta)} \int e_{}^{\theta - s(\theta)} \, dM + 1
	\ \le \ \sup_{s \in \R} \Bigl( s
		- e_{}^{s} \int e_{}^{\theta - s(\theta)} \, dM \Bigr) + 1
	\ = \ - \log \int e_{}^{\theta - s(\theta)} \, dM ,
\]
and the right-hand side is not larger than
\begin{align*}
	&\begin{cases}
		\displaystyle
		- \log \int_{\XX_\ell} e_{}^{\theta'(x_1\,-)(x - x_1)} \, dM - 1
			& \text{in case of \eqref{eq:boundaries.left}} \\
		\displaystyle
		- \log \int_{\XX_r} e_{}^{\theta'(x_n\,+)(x - x_n)} \, dM - 1
			& \text{in case of \eqref{eq:boundaries.right}}
		\end{cases} \\
	&\le \ - \min \biggl\{
		\log \int_{\XX_\ell} e_{}^{- r(\theta)(x - x_1)/(x_n - x_1)} \, dM , \
		\log \int_{\XX_r} e_{}^{r(\theta)(x - x_n)/(x_n - x_1)} \, dM
		\biggr\} - 1 .
\end{align*}
Hence these inequalities show that
\[
	L(\theta) \ \to \ - \infty
	\quad\text{as} \ r(\theta)
		+ \max \bigl\{ - \theta'(x_1\,-), \theta'(x_n\,+) \bigr\} \to \infty .
	\qedhere
\]
\end{proof}

\begin{proof}[\bf Proof of Lemmas~\ref{lem:existence.uniqueness.2A} and \ref{lem:existence.uniqueness.2B}]
We first consider Setting~2A. For an arbitrary function $\theta \in \Theta$ let
\[
	\tilde{\theta}(x) \ := \ \begin{cases}
		\theta(x_1) + (x - x_1) \theta'(x_1\,+)
			& \text{if} \ x \le x_1 , \\
		\theta(x)
			& \text{if} \ x \in [x_1,x_n] , \\
		\theta(x_n) + (x - x_n) \theta'(x_n\,-)
			& \text{if} \ x \ge x_n .
	\end{cases}
\]
Then $\tilde{\theta} \le \theta$, $\tilde{\theta} \equiv \theta$ on $[x_1,x_n]$, and $L(\tilde{\theta}) \ge L(\theta)$ with equality if, and only if $\tilde{\theta} \equiv \theta$. Thus we may restrict our attention to convex functions $\theta$ such that $\theta' \equiv \theta'(x_1\,+)$ on $(-\infty,x_1]$ and $\theta' \equiv \theta'(x_n\,-)$ on $[x_n,\infty)$.

Let $(\theta_k)_k$ be a sequence of such functions such that $\lim_{k \to \infty} L(\theta_k) = \sup_{\theta \in \Theta} L(\theta)$. By Lemma~\ref{lem:coercivity},
\[
	\sup_k \Bigl( \sup_{x \in [x_1,x_n]} \, |\theta_k(x)| + 
		\max \bigl\{ -\theta_k'(x_1), \theta_k'(x_n) \bigr\} \Bigr)
	\ < \ \infty .
\]
Consequently, the sequence $(\theta_k)_k$ is uniformly bounded on $[x_1,x_n]$ and uniformly Lipschitz continuous on $\R$. Hence we may apply the theorem of Arzela--Ascoli and replace $(\theta_k)_k$ with a subsequence, if necessary, such that $\theta_k \to \theta \in \Theta$ pointwise and uniformly on any compact set as $k \to \infty$. By Fatou's lemma, $L(\theta) \ge \lim_{k \to \infty} L(\theta_k)$, so $\theta$ is a maximizer of $L$ over $\Theta$.

One can easily deduce from strict convexity of $\exp(\cdot)$ that $L$ is strictly concave on $\Theta$. Hence there exists a unique maximizer $\hat{\theta}$ of $L$ over $\Theta$.

Let
\[
	\check{\theta}(x)
	\ := \ \max_{i=1,\ldots,n} \bigl( \hat{\theta}(x_i) + \hat{\theta}'(x_i)(x - x_i) \bigr)
\]
with $\hat{\theta}'(x_i\,-) \le \hat{\theta}'(x_i) \le \hat{\theta}'(x_i\,+)$ for $2 \le i < n$. This defines another function $\check{\theta} \in \Theta$ such that $(\check{\theta}(x_i))_{i=1}^n = (\hat{\theta}(x_i))_{i=1}^n$ and $\check{\theta} \le \hat{\theta}$. Thus we may conclude that $\hat{\theta} \equiv \check{\theta}$, a function with at most $n-1$ changes of slope, all of which are within $(x_1,x_n)$.

Suppose that $\hat{\theta}$ changes slope at two points $\tau_1 < \tau_2$ but $(\tau_1,\tau_2)$ contains no observation $x_i$. Then we could redefine
\[
	\hat{\theta}(x)
	\ := \ \max \bigl( \hat{\theta}(\tau_1) + \hat{\theta}'(\tau_1\,-) (x - \tau_1),
		\hat{\theta}(\tau_2) + \hat{\theta}'(\tau_2\,+) (x - \tau_2) \bigr)
\]
for $x \in (\tau_1,\tau_2)$. This modification would not change the vector $(\hat{\theta}(x_i))_{i=1}^n$ but decrease strictly the integral $\int e_{}^{\hat{\theta}(x)} \, P_o(dx)$, a contradiction to optimality of $\hat{\theta}$. Hence any interval $[x_i, x_{i+1}]$, $1 \le i < n$, contains at most one point $\tau$ such that $\hat{\theta}'(\tau\,-) < \hat{\theta}'(\tau\,+)$.

Finally, as argued in Section~\ref{subsec:Characterization}, $\hat{\theta}$ satisfies the (in)equalities
\[
	h(\tau) \ := \ \int (x - \tau)^+ \, (\hat{P} - P_{\hat{\theta}})(dx)
	\ \begin{cases}
		\le \ 0 & \text{for all} \ \tau \in (x_1,x_n) , \\
		= \ 0 & \text{if} \ \hat{\theta}'(\tau\,-) < \hat{\theta}'(\tau\,+) .
	\end{cases}
\]
But $h(\cdot)$ itself is continuous with one-sided derivatives
\[
	h'(\tau\,\pm) \ = \ \hat{F}(\tau\,\pm) - F_{\hat{\theta}}(\tau) ,
\]
where $\hat{F}$ and $F_{\hat{\theta}}$ are the distribution functions of $\hat{P}$ and $P_{\hat{\theta}}$, respectively. If $\hat{\theta}$ changes slope at some point $\tau$, then it follows from $h \le 0 = h(\tau)$ that $h'(\tau\,-) \ge 0 \ge h'(\tau\,+)$, so
\[
	0 \ \ge \ h'(\tau\,+) - h'(\tau\,-) \ = \ \hat{P}(\{\tau\}) .
\]
Hence $\tau$ cannot be an observation $x_i$.

These arguments prove Lemma~\ref{lem:existence.uniqueness.2A}. The same arguments apply to Setting~2B without essential changes, because the functions $\tilde{\theta}, \check{\theta}$ and $\theta = \lim_{k \to \infty} \theta_k$ above are automatically isotonic. The only difference, merely notational, is that in case of $\hat{\theta}'(0\,+) > 0$ we interpret $0$ as a first knot $\tau_1$. Hence Lemma~\ref{lem:existence.uniqueness.2B} is also true.
\end{proof}

\begin{proof}[\bf Proof of Lemma~\ref{lem:convergence}]
We prove the lemma for Setting~2A. The arguments for Setting~2B and Setting~1 are very similar, see Section~\ref{subsec:further.proofs}. Let $\Theta_o$ be the set of all functions $\theta \in \Theta \cap \V$ such that $L(\theta) \ge L_o$. Obviously, the target function $\hat{\theta}$ belongs to $\Theta_o$. It follows from Lemma~\ref{lem:coercivity} that
\[
	C_o \
	:= \ \sup_{\theta \in \Theta_o} \, \sup_{x \in [x_1,x_n]} \, |\theta(x)|
	\ < \ \infty ,
\]
and
\[
	C_{\ell} \ := \ \inf_{\theta \in \Theta_o} \, \theta'(x_1\,-)
		\ > \ \lambda_\ell(P_o) ,
	\quad
	C_{r} \ := \ \sup_{\theta \in \Theta_o} \, \theta'(x_n\,+)
		\ < \ \lambda_r(P_o) .
\]

For arbitrary $\theta \in \Theta_o$, let $\theta_{\rm new} \in \V_{D(\theta)}$ be the subsequent Newton proposal. Precisely, $\theta_{\rm new} - \theta$ maximizes the second order Taylor approximation
\[
	L(\theta) + DL(\theta,v) - 2^{-1} H(\theta,v)
\]
of $L(\theta + v)$ over all $v \in \V_{D(\theta)}$, and elementary considerations show that
\[
	DL(\theta, \theta_{\rm new} - \theta)
	\ = \ \max_{v \in \V_{D(\theta)} \setminus \{0\}} \,
		\frac{DL(\theta,v)^2}{H(\theta,v)} .
\]
Now let $\VV$ be the set of basis functions $v_0(x) := 1$, $v_1(x) := x - x_1$ and $V_\tau(x) = (x - \tau)^+$, $\tau \in \DD$. Then for any $v \in \VV$,
\[
	H(\theta,v) \ = \ \int v^2 e^\theta \, dP_o
	\ \le \ C_{\rm N} \ := \ \int v_{\rm max}(x)^2 e^{\theta_{\rm max}(x)} \, P_o(dx)
	\ < \ \infty ,
\]
where $v_{\rm max}(x) := \max(1,|x - x_1|)$ is an upper bound for $|v(x)|$, $v \in \VV$, and $\theta_{\rm max}(x) := C_o - C_\ell (x - x_1)^- + C_r (x - x_n)^+$ is an upper bound for $\theta(x)$, $\theta \in \Theta_o$. That $C_{\rm N}$ is finite follows from the fact that $\int e_{}^{\theta_{\rm max}(x) + \eps |x|} \, P_o(dx) < \infty$ for sufficiently small $\eps > 0$. Consequently,
\[
	DL(\theta,v) \ \le \ \sqrt{C_{\rm N} \delta_{\rm Newton}(\theta)}
	\quad\text{for all} \ v \in \VV \cap \V_{D(\theta)} .
\]

After these preparations, let us compare $\theta$ with $\hat{\theta}$. By concavity of $L(\cdot)$,
\[
	L(\hat{\theta}) - L(\theta)
	\ \le \ DL(\theta, \hat{\theta} - \theta) .
\]
Now we write $\hat{\theta} - \theta = \alpha_0 v_0 + \alpha_1 v_1 + \sum_{\tau \in \DD} \beta_\tau V_\tau$ with parameters satisfying
\begin{align*}
	|\alpha_0| \ &= \ \bigl| \hat{\theta}(x_1) - \theta(x_1) \bigr|
		\ \le \ 2C_o , \\
	|\alpha_1| \ &= \ \bigl| \hat{\theta}'(x_1) - \theta'(x_1) \bigr|
		\le \ C_r - C_\ell
		\quad\text{and} \\
	\beta_\tau \ &= \ \hat{\theta}'(\tau\,+) - \hat{\theta}'(\tau\,-)
		- \bigl( \theta'(\tau\,+) - \theta'(\tau\,-) \bigr)
		\ \begin{cases}
		\le \ \ \hat{\theta}'(\tau\,+) - \hat{\theta}'(\tau\,-) , \\
		\ge \ - \bigl( \theta'(\tau\,+) - \theta'(\tau\,-) \bigr) .
	\end{cases}
\end{align*}
In particular,
\[
	\sum_{\tau \in \DD} \beta_\tau^+ \ \le \ \hat{\theta}'(x_n) - \hat{\theta}'(x_1)
		\ \le \ C_r - C_\ell , \quad
	\sum_{\tau \in \DD} \beta_\tau^- \ \le \ \theta'(x_n) - \theta'(x_1)
		\ \le \ C_r - C_\ell .
\]
If $\beta_\tau^- > 0$, then $\tau \in D(\theta)$. And if $\tau \in D(\theta)$, then $V_\tau \in \V_{D(\theta)}$ and $|DL(\theta,V_\tau)| \le \sqrt{C_{\rm N} \delta_{\rm Newton}(\theta)}$. For $\tau \in \DD \setminus D(\theta)$, we know that $\beta_\tau = \beta_\tau^+$ and
\[
	DL(\theta,V_\tau)
	\ = \ DL(\theta,V_{\tau,\theta}) + DL(\theta,\eta_{\tau,\theta})
	\ \le \ \delta_{\rm Knot}(\theta)
		+ (1 + x_n - x_1) \sqrt{C_{\rm N} \delta_{\rm Newton}(\theta)} .
\]
Here $V_{\tau,\theta} = V_\tau - \eta_{\tau,\theta}$ is the localised kink function with $\eta_{\tau,\theta} \in \V_{D(\theta})$ as described in Section~\ref{subsec:Localised.kinks}. The explicit construction of $\eta_{\tau,\theta}$ shows that it is a linear combination of at most two basis functions in $\VV \cap \V_{D(\theta)}$ with coefficients whose absolute values sum to less than $1 + x_n - x_1$. This explains the upper bound $(1 + x_n - x_1) \sqrt{C_{\rm N} \delta_{\rm Newton}(\theta)}$ for $DL(\theta,\eta_{\tau,\theta})$. Consequently,
\begin{align*}
	DL(\theta, \hat{\theta} - \theta) \
	\le \ &\alpha_0 DL(\theta, v_0) + \alpha_1 DL(\theta, v_1)
		+ \sum_{\tau \in \DD} \beta_\tau^+ DL(\theta, V_\tau)^+
		+ \sum_{\tau \in \DD} \beta_\tau^- DL(\theta, V_\tau)^- \\
	\le \ &2C_o \sqrt{C_{\rm N} \delta_{\rm Newton}(\theta)}
		+ (C_r - C_\ell) \sqrt{C_{\rm N} \delta_{\rm Newton}(\theta)} \\
		&+ \ (C_r - C_\ell)
			\bigl( \delta_{\rm Knot}(\theta)
				+ (1 + x_n - x_1) \sqrt{C_{\rm N} \delta_{\rm Newton}(\theta)} \bigr)
		+ (C_r - C_\ell) \sqrt{C_{\rm N} \delta_{\rm Newton}(\theta)} ,
\end{align*}
so the assertion is true with $C_{\rm Newton} = \bigl( 2C_o + (C_r - C_\ell)(3 + x_n - x_1) \bigr) \sqrt{C_{\rm N}}$ and $C_{\rm Knot} = C_r - C_\ell$.
\end{proof}

\paragraph{Acknowledgements.}
This work was supported by Swiss National Science Foundation. We owe thanks to Peter McCullagh for drawing our attention to the nonparametric tail inflation model of \cite{mccullagh2012}, to Jon Wellner for the hint to Artin's theorem and Gaussian mixtures, and to Jasha Sommer-Simpson for sharing his MSc thesis. Constructive comments of two referees are gratefully acknowledged.


\begin{thebibliography}{20}
\providecommand{\natexlab}[1]{#1}
\providecommand{\url}[1]{\texttt{#1}}
\expandafter\ifx\csname urlstyle\endcsname\relax
  \providecommand{\doi}[1]{doi: #1}\else
  \providecommand{\doi}{doi: \begingroup \urlstyle{rm}\Url}\fi

\bibitem[Cule et~al.(2010)Cule, Samworth, and Stewart]{Cule_etal_2010}
Cule, Madeleine, Samworth, Richard, and Stewart, Michael.
\newblock Maximum likelihood estimation of a multi-dimensional log-concave
  density.
\newblock \emph{J. R. Stat. Soc. Ser. B Stat. Methodol.}, 72\penalty0
  (5):\penalty0 545--607, 2010.
\newblock URL \url{http://dx.doi.org/10.1111/j.1467-9868.2010.00753.x}.

\bibitem[Donoho and Jin(2004)]{Donoho_Jin_2004}
Donoho, David and Jin, Jiashun.
\newblock Higher criticism for detecting sparse heterogeneous mixtures.
\newblock \emph{Ann. Statist.}, 32\penalty0 (3):\penalty0 962--994, 2004.

\bibitem[D\"{u}mbgen(2017)]{duembgen2017}
D\"{u}mbgen, Lutz.
\newblock Optimization methods -- with applications in statistics.
\newblock Lecture notes, University of Bern, 2017.

\bibitem[D{\"u}mbgen and Rufibach(2011)]{Duembgen_Rufibach_2011}
D{\"u}mbgen, Lutz and Rufibach, Kaspar.
\newblock logcondens: Computations related to univariate log-concave density
  estimation.
\newblock \emph{J. Statist. Software}, 39\penalty0 (6):\penalty0 1--28, 2011.
\newblock \doi{10.18637/jss.v039.i06}.
\newblock URL \url{http://www.jstatsoft.org/v39/i06}.

\bibitem[D\"{u}mbgen et~al.(2007/2011)D\"{u}mbgen, H\"{u}sler, and
  Rufibach]{Duembgen_etal_2011}
D\"{u}mbgen, Lutz, H\"{u}sler, Andr\'{e}, and Rufibach, Kaspar.
\newblock Active set and {EM} algorithms for log-concave densities based on
  complete and censored data.
\newblock Technical report~61, University of Bern, 2007/2011.
\newblock URL \url{https://arxiv.org/abs/0707.4643}.

\bibitem[Gontscharuk et~al.(2016)Gontscharuk, Landwehr, and
  Finner]{Gontscharuk_etal_2016}
Gontscharuk, Veronika, Landwehr, Sandra, and Finner, Helmut.
\newblock Goodness of fit tests in terms of local levels with special emphasis
  on higher criticism tests.
\newblock \emph{Bernoulli}, 22\penalty0 (3):\penalty0 1331--1363, 2016.

\bibitem[Groeneboom and Jongbloed(2014)]{Groeneboom_Jongbloed_2014}
Groeneboom, Piet and Jongbloed, Geurt.
\newblock \emph{Nonparametric estimation under shape constraints}, volume~38 of
  \emph{Cambridge Series in Statistical and Probabilistic Mathematics}.
\newblock Cambridge University Press, New York, 2014.
\newblock Estimators, algorithms and asymptotics.

\bibitem[Groeneboom et~al.(2001)Groeneboom, Jongbloed, and
  Wellner]{groeneboom2001}
Groeneboom, Piet, Jongbloed, Geurt, and Wellner, Jon~A.
\newblock Estimation of a convex function: Characterizations and asymptotic
  theory.
\newblock \emph{Ann. Statist.}, 29\penalty0 (6):\penalty0 1653--1698, 12 2001.
\newblock \doi{10.1214/aos/1015345958}.
\newblock URL \url{http://dx.doi.org/10.1214/aos/1015345958}.

\bibitem[Groeneboom et~al.(2008)Groeneboom, Jongbloed, and
  Wellner]{Groeneboom_etal_2008}
Groeneboom, Piet, Jongbloed, Geurt, and Wellner, Jon~A.
\newblock The support reduction algorithm for computing nonparametric function
  estimates in mixture models.
\newblock \emph{Scand. J. Statist.}, 35:\penalty0 385--399, 2008.

\bibitem[Liu and Wang(2018)]{Liu_Wang_2018}
Liu, Yu and Wang, Yong.
\newblock A fast algorithm for univariate log-concave density estimation.
\newblock \emph{Aust. N. Z. J. Stat.}, 60\penalty0 (2):\penalty0 258--275,
  2018.
\newblock ISSN 1369-1473.

\bibitem[Marshall and Olkin(1979)]{Marshall_Olkin_1979}
Marshall, Albert~W. and Olkin, Ingram.
\newblock \emph{Inequalities: theory of majorization and its applications},
  volume 143 of \emph{Mathematics in Science and Engineering}.
\newblock Academic Press, Inc. [Harcourt Brace Jovanovich, Publishers], New
  York-London, 1979.

\bibitem[McCullagh and Polson(2012)]{mccullagh2012}
McCullagh, Peter and Polson, Nicholas~G.
\newblock Tail inflation.
\newblock Preprint, 2012.

\bibitem[McCullagh and Polson(2017)]{McCullagh_Polson_2017}
McCullagh, Peter and Polson, Nicholas~G.
\newblock Statistical sparsity.
\newblock \emph{Biometrika}, 105\penalty0 (4):\penalty0 797--814, 2017.

\bibitem[{R Core Team}(2016)]{R2016}
{R Core Team}.
\newblock \emph{R: A Language and Environment for Statistical Computing}.
\newblock R Foundation for Statistical Computing, Vienna, Austria, 2016.
\newblock URL \url{https://www.R-project.org/}.

\bibitem[Samworth(2018)]{samworth2018}
Samworth, Richard~J.
\newblock Recent progress in log-concave density estimation.
\newblock \emph{Statist. Sci.}, 33\penalty0 (4):\penalty0 493--509, 11 2018.
\newblock \doi{10.1214/18-STS666}.
\newblock URL \url{https://doi.org/10.1214/18-STS666}.

\bibitem[Silverman(1982)]{silverman1982}
Silverman, Bernard~W.
\newblock On the estimation of a probability density function by the maximum
  penalized likelihood method.
\newblock \emph{Ann. Statist.}, 10\penalty0 (3):\penalty0 795--810, 09 1982.
\newblock \doi{10.1214/aos/1176345872}.
\newblock URL \url{http://dx.doi.org/10.1214/aos/1176345872}.

\bibitem[Sommer-Simpson(2019)]{Sommer-Simpson_2019}
Sommer-Simpson, Jasha.
\newblock Convergence of {D}{\"u}mbgen's algorithm for estimation of tail
  inflation.
\newblock Master's thesis, Department of Statistics, Univ. of Chicago, 2019.
\newblock arxiv:1906.04544.

\bibitem[{von Neumann}(1951)]{vonNeumann_1951}
{von Neumann}, John.
\newblock Various techniques used in connection with random digits.
\newblock \emph{J. Res. Nat. Bur. Stand. Appl. Math. Series}, 3:\penalty0
  36--38, 1951.

\bibitem[Walther(2002)]{Walther_2002}
Walther, Guenther.
\newblock Detecting the presence of mixing with multiscale maximum likelihood.
\newblock \emph{J. Amer. Statist. Assoc.}, 97\penalty0 (458):\penalty0
  508--513, 2002.
\newblock ISSN 0162-1459.
\newblock \doi{10.1198/016214502760047032}.

\bibitem[Wang(2018)]{Wang_2018}
Wang, Yong.
\newblock Computation of the nonparametric maximum likelihood estimate of a
  univariate log-concave density.
\newblock \emph{WIREs Computational Statistics}, 11\penalty0 (1):\penalty0
  e1452, 2018.

\end{thebibliography}

\clearpage

\appendix
\section{Technical details}
\label{sec:technical.details}

\subsection{Localised kink functions}
\label{subsec:Localised.kinks}

As mentioned at the end of Section~\ref{subsec:Characterization}, working with the kink  functions $V_\tau(x) = \xi (x - \tau)^+$ may be computationally inefficient and numerically problematic. For instance, by means of local search we obtain functions $\theta$ satisfying \eqref{eq:optimality.1} approximately, but not perfectly. As a result it may happen that $DL(\theta,V_\tau) > 0$ for some $\tau \in D(\theta)$ although this contradicts \eqref{eq:optimality.1}. Furthermore, the support of $V_\tau$ may contain several points $\sigma \in D(\theta)$, so the evaluation of $DL(\theta,V_\tau)$ would involve several integrals of an affine function times a log-affine function with respect to $P_o$. Hence we propose to replace the simple kink functions $V_\tau$ in \eqref{eq:optimality.2} with localised kink functions $V_{\tau,\theta} = V_\tau - \eta_{\tau,\theta}$ for some $\eta_{\tau,\theta} \in \V_{D(\theta)}$ such that\\[0.5ex]
(i) \ $\theta$ is affine on $\{x \in \XX : V_{\tau,\theta}(x) \ne 0\}$,\\
(ii) \ $\tau \mapsto V_{\tau,\theta}(x)$ is Lipschitz-continuous with constant $1$ for any $x \in \XX$,\\
(iii) \ $V_{\tau,\theta} \equiv 0$ if $\tau \in D(\theta)$.\\[0.5ex]
Then we redefine the auxiliary function $h_\theta$ and replace \eqref{eq:optimality.2} with
\begin{equation}
\label{eq:optimality.2'}
	h_\theta(\tau) := DL(\theta, V_{\tau,\theta}) \ \le \ 0
	\quad\text{for all} \ \tau \in \DD \setminus D(\theta) .
\end{equation}
Note that in case of \eqref{eq:optimality.1}, the two requirements \eqref{eq:optimality.2} and \eqref{eq:optimality.2'} are equivalent, because then $DL(\theta,V_{\tau,\theta}) = DL(\theta,V_\tau)$. We do assume that $P_\theta$ is a probability measure, even if \eqref{eq:optimality.1} is not satisfied perfectly.

To simplify subsequent explicit formulae, let us introduce the following auxiliary functions: For real numbers $a < b$ let
\[
	j_{10}(x; a, b) \ := \ 1_{[a < x \le b]} \, \frac{b - x}{b - a}
	\quad\text{and}\quad
	j_{01}(x; a, b) \ := \ 1_{[a < x \le b]} \, \frac{x - a}{b - a} ,
\]
so $j_{10}(x; a, b) + j_{01}(x; a, b) = 1_{[a < x \le b]}$. In addition we set $j_{01}(x; a,a) := j_{10}(x; a,a) := 0$.

In Setting~1 let $D(\theta) \cup \{x_1,x_n\} = \{\tau_1,\ldots,\tau_m\}$ with $m \ge 2$ points $\tau_1 < \cdots < \tau_m$ in $\{x_1,\ldots,x_n\}$. Then for $\tau_j \le \tau \le \tau_{j+1}$ with $1 \le j < m$,
\begin{align*}
	V_{\tau,\theta}(x)
	\ := \ &V_\tau(x)
		- \frac{\tau_{j+1} - \tau}{\tau_{j+1} - \tau_j} \, V_{\tau_j}(x)
		- \frac{\tau - \tau_j}{\tau_{j+1} - \tau_j} \, V_{\tau_{j+1}}(x)
	 \ = \ \begin{cases}
	 	0
			& \text{for} \ x \not\in [\tau_j,\tau_{j+1}] \\[1.5ex]
		\displaystyle
		\frac{(x - \tau_j)(\tau_{j+1} - \tau)}{\tau_{j+1} - \tau_j}
			& \text{for} \ x \in [\tau_j,\tau] \\[2ex]
		\displaystyle
		\frac{(\tau - \tau_j)(\tau_{j+1} - x)}{\tau_{j+1} - \tau_j}
			& \text{for} \ x \in [\tau,\tau_{j+1}]
		\end{cases} \\[1ex]
	 = \ &\frac{(\tau - \tau_j)(\tau_{j+1} - \tau)}{\tau_{j+1} - \tau_j}
		\, \bigl( j_{01}(x; \tau_j, \tau) + j_{10}(x; \tau, \tau_{j+1}) \bigr) .
\end{align*}
Figure~\ref{fig:LocalPert1} illustrates these localised kink functions $V_{\tau,\theta}$.

\begin{figure}
\centering
\includegraphics[width=0.7\textwidth]{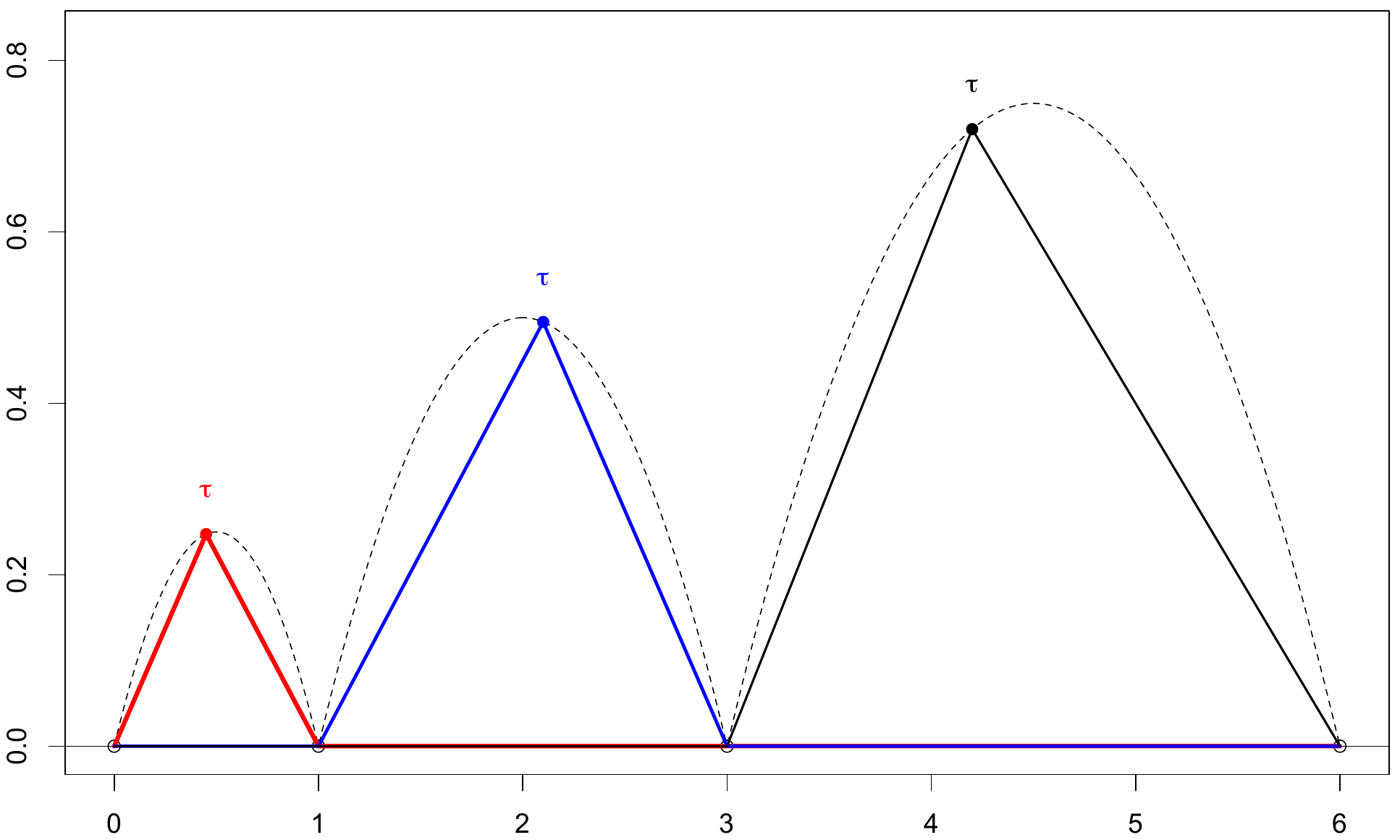}
\caption{Localised kink functions in Setting~1: For $D(\theta) \cup \{x_1,x_n\} = \{0,1,3,6\}$ one sees $V_{\tau,\theta}$ for three different values of $\tau$.}
\label{fig:LocalPert1}
\end{figure}

Now we consider Settings~2A-B. If $D(\theta) = \emptyset$, we set $V_{\tau,\theta} := V_\tau = (\cdot - \tau)^+$ and note that $\partial V_\tau(x)/\partial \tau = - 1_{[x > \tau]}$ for $x \ne \tau$. Otherwise, let $D(\theta) = \{\tau_1,\ldots,\tau_m\}$ with $m \ge 1$ points $\tau_1 < \cdots < \tau_m < x_n$, where $\tau_1 > x_1$ in Setting~2A and $\tau_1 \in \{0\} \cup (x_1,x_n)$ in Setting~2B. For $\tau \le \tau_1$ we define
\begin{align}
\nonumber
	V_{\tau,\theta}(x) \
	:= \ &V_\tau(x) - (\tau_1 - \tau) - V_{\tau_1}(x)
	 \ = \ \begin{cases}
	 	\tau - \tau_1
			& \text{for} \ x \le \tau \\
		x - \tau_1
			& \text{for} \ x \in [\tau,\tau_1] \\
		0
			& \text{for} \ x \ge \tau_1
		\end{cases} \\[0.5ex]
\label{eq:Vtautheta.2.l}
	= \ &(\tau - \tau_1) \bigl( 1_{[x \le \tau]} + j_{10}(x; \tau, \tau_1) \bigr)
\end{align}
and note that
\begin{equation}
\label{eq:Vtautheta.2.l'}
	\partial V_{\tau,\theta}(x)/\partial \tau
	\ = \ 1_{[x \le \tau]}
	\quad\text{for} \ x \ne \tau .
\end{equation}
For $\tau_j \le \tau \le \tau_{j+1}$ with $1 \le j < m$ we set
\begin{align}
\nonumber
	V_{\tau,\theta}(x) \
	:= \ &V_\tau(x)
		- \frac{\tau_{j+1} - \tau}{\tau_{j+1} - \tau_j} \, V_{\tau_j}(x)
		- \frac{\tau - \tau_j}{\tau_{j+1} - \tau_j} \, V_{\tau_{j+1}}(x)
	\ = \ \begin{cases}
	 	0
			& \text{for} \ x \not\in [\tau_j,\tau_{j+1}] \\[1.5ex]
		\displaystyle
		- \frac{(x - \tau_j)(\tau_{j+1} - \tau)}{\tau_{j+1} - \tau_j}
			& \text{for} \ x \in [\tau_j,\tau] \\[2ex]
		\displaystyle
		- \frac{(\tau - \tau_j)(\tau_{j+1} - x)}{\tau_{j+1} - \tau_j}
			& \text{for} \ x \in [\tau,\tau_{j+1}]
		\end{cases} \\[1ex]
\label{eq:Vtautheta.2.m}
	 = \ &(\tau - \tau_j)
	 	\bigl( 1_{[\tau_j < x \le \tau]}
			- j_{10}(x; \tau_j, \tau_{j+1})
			- j_{01}(x; \tau_j, \tau) \bigr) .
\end{align}
and note that
\begin{equation}
\label{eq:Vtautheta.2.m'}
	\partial V_{\tau,\theta}(x)/\partial \tau \
	= \ 1_{[\tau_j < x \le \tau]} - j_{10}(x; \tau_j, \tau_{j+1})
	\quad\text{for} \ x \ne \tau ,
\end{equation}
because $1_{[\tau_j < x \le \tau]}$ and $(\tau - \tau_j) j_{01}(x; \tau_j,\tau) = 1_{[\tau_j < x \le \tau]} (x - \tau_j)$ are locally constant in $\tau \ne x$. Finally, for $\tau > \tau_m$ we define
\begin{align}
\nonumber
	V_{\tau,\theta}(x) \
	:= \ &V_\tau(x) - V_{\tau_m}(x)
	\ = \ \begin{cases}
	 	0
			& \text{for} \ x \le \tau_m \\
		\tau_m - x
			& \text{for} \ x \in [\tau_m,\tau] \\
		\tau_m - \tau
			& \text{for} \ x \ge \tau
		\end{cases} \\[0.5ex]
\label{eq:Vtautheta.2.r}
	 = \ &(\tau - \tau_m) \big( -1_{[x > \tau]} - j_{01}(x; \tau_m, \tau) \bigr)
\end{align}
and note that
\begin{equation}
\label{eq:Vtautheta.2.r'}
	\partial V_{\tau,\theta}(x) \ = \ -1_{[x > \tau]}
	\quad\text{for} \ x \ne \tau .
\end{equation}
Figure~\ref{fig:LocalPert2} illustrates these localised kink functions $V_{\tau,\theta}$.

\begin{figure}
\centering
\includegraphics[width=0.7\textwidth]{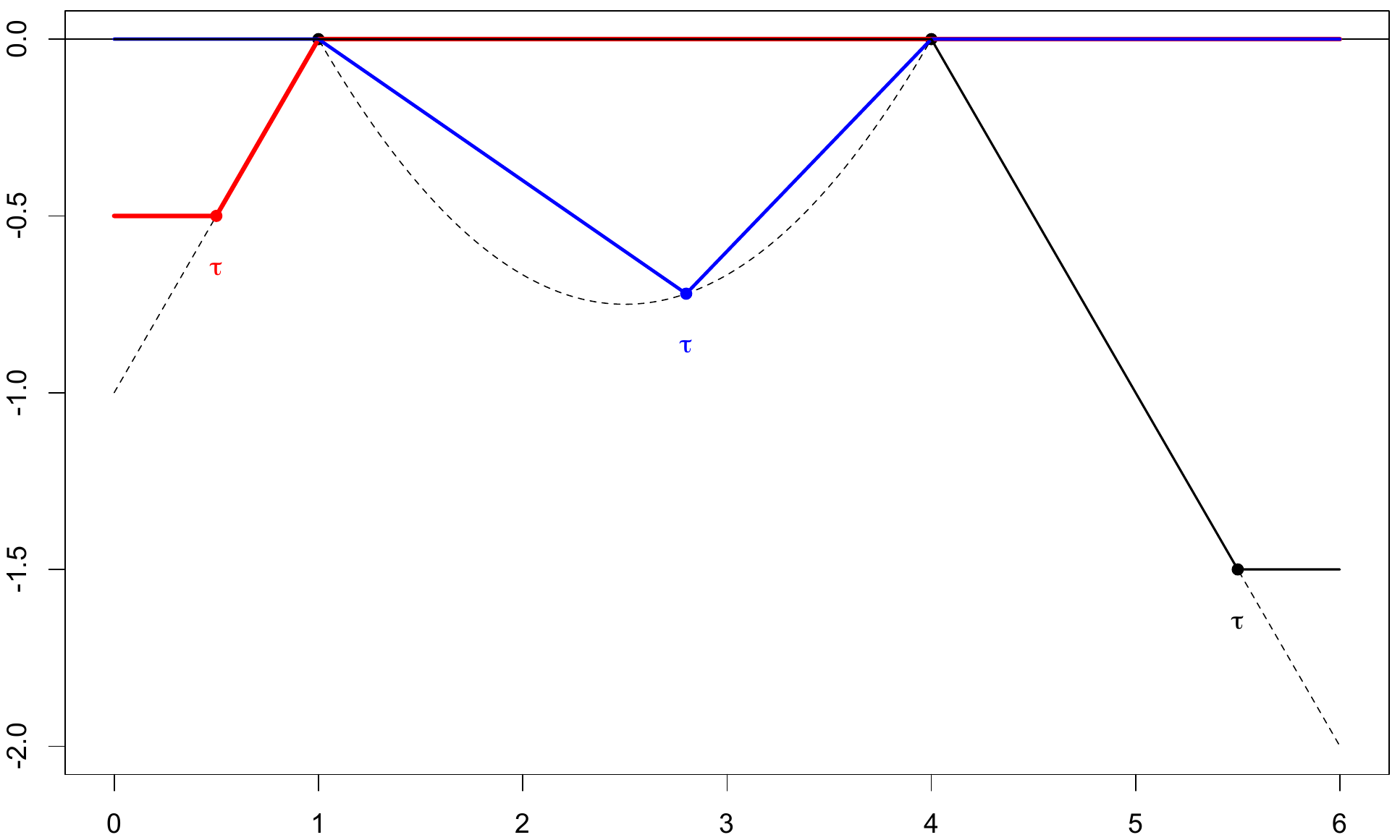}
\caption{Localised kink functions in Settings~2A-B: For $D(\theta) = \{1,4\}$ one sees $V_{\tau,\theta}$ for three different values of $\tau$.}
\label{fig:LocalPert2}
\end{figure}

When searching for local maxima of $h_\theta(\tau) := DL(\theta,V_{\tau,\theta})$ in case of $D(\theta) = \{\tau_1,\ldots,\tau_m\}$ as above, one should treat the $m+1$ intervals $(-\infty,\tau_1]$, $[\tau_j,\tau_{j+1}]$ with $1 \le j < m$ and $[\tau_m,\infty)$ separately, because $h_\theta$ equals $0$ but could be non-differentiable at points in $D(\theta)$. Hence one should look for maximizers of $h_\theta$ on the $\tilde{n} - 1$ intervals $[t_i, t_{i+1}]$, $1 \le i < \tilde{n}$, where $t_1 < \cdots < t_{\tilde{n}}$ are the different elements of $\{x_1,\ldots,x_n\} \cup \{\tau_1,\ldots,\tau_m\}$.

Now we provide explicit formulae for $h_\theta$ and its one-sided derivatives. One can easily derive from \eqref{eq:Vtautheta.2.l} and \eqref{eq:Vtautheta.2.l'} that for $\tau < \tau_1$,
\begin{align*}
	h_\theta'(\tau\,+) \
	&= \ (\hat{F} - F_\theta)(\tau) \quad\text{and} \\
	h_\theta(\tau) \
	&= \ (\tau - \tau_1) \Bigl( h_\theta'(\tau\,+)
		+ \int j_{10}(x; \tau,\tau_1) \, (\hat{P} - P_\theta)(dx) \Bigr) .
\end{align*}
For $1 \le j < m$ and $\tau_j \le \tau < \tau_{j+1}$, equations \eqref{eq:Vtautheta.2.m} and \eqref{eq:Vtautheta.2.m'} lead to
\begin{align*}
	h_\theta'(\tau\,+) \
	&= \ (\hat{P} - P_\theta)((\tau_j,\tau])
		- \int \ j_{10}(x; \tau_j, \tau_{j+1}) \, (\hat{P} - P_\theta)(dx)
		\quad\text{and} \\
	h_\theta(\tau) \
	&= \ (\tau - \tau_j) \Bigl( h_\theta'(\tau\,+)
		- \int j_{01}(x; \tau_j,\tau) \, (\hat{P} - P_\theta)(dx) \Bigr) .
\end{align*}
Finally, for $\tau \ge \tau_m$, it follows from \eqref{eq:Vtautheta.2.r} and \eqref{eq:Vtautheta.2.r'} that
\begin{align*}
	h_\theta'(\tau\,+) \
	&= \ (\hat{F} - F_\theta)(\tau)
		\ = \ - (\hat{P} - P_\theta)((\tau,\infty))
		\quad\text{and} \\
	h_\theta(\tau) \
	&= \ (\tau - \tau_m) \Bigl( h_\theta'(\tau\,+)
		- \int j_{01}(x; \tau_m,\tau) \, (\hat{P} - P_\theta)(dx) \Bigr) .
\end{align*}
The representation of $h_\theta(\tau)$ in terms of $h_\theta'(\tau\,+)$ is particularly convenient, because $h_\theta$ is evaluated only at its local maximizers, i.e.\ zeros of $h_\theta'$.

\subsection{Details for Setting~1}
\label{subsec:Details.1}

\paragraph{Auxiliary functions.}
For real numbers $x_1 < x_2$ and a linear function $\theta$ on $[x_1,x_2]$,
\[
	\int_{x_1}^{x_2} e_{}^{\theta(x)} \, dx
	\ = \ (x_2 - x_1) J \bigl( \theta(x_1), \theta(x_2) \bigr)
\]
with
\begin{equation}
	\label{Eq:AS_AuxiliaryJ}
	J(r,s) \ := \ \int_0^1 e_{}^{(1 - v)r + v s} \, dv
	\ = \ \begin{cases}
		\displaystyle
		\frac{e^s - e^r}{s - r} & \text{if} \ r \ne s , \\
		e_{}^s & \text{if} \ r = s .
	\end{cases}
\end{equation}
In general, for integers $a,b \ge 0$,
\[
	J_{ab}(r,s) := \frac{\partial^{a+b}}{\partial r^a \partial s^b} \, J(r,s)
	\ = \ \int_0^1 (1 - v)^a v^b e_{}^{(1 - v) r + vs} \, dv .
\]
Let $m := (r + s)/2$ and $\delta := (s - r)/2$, so $r = m - \delta$, $s = m + \delta$ and $s - r = 2\delta$. In case of $\delta \ne 0$ we may write
\[
	J(r,s) \ = \ e_{}^m \sinh(\delta)/\delta .
\]
Moreover, with $\Delta := s - r = 2\delta$, partial integration leads to the formulae
\begin{align*}
	J_{10}(r,s) \
	&= \ e_{}^r \int_0^1 (1 - v) e_{}^{\Delta v} \, dv
		\ = \ e_{}^r \Bigl( - \frac{1}{\Delta}
			+ \frac{e^{\Delta} - 1}{\Delta^2} \Bigr)
		&&= \ e_{}^m
			 \bigl( \sinh(\delta) - \delta e_{}^{-\delta} \bigr)
			 	/ (2\delta^2) , \\
	J_{20}(r,s) \
	&= \ e_{}^r \int_0^1 (1 - v)^2 e_{}^{\Delta v} \, dv
		\ = \ e_{}^r \Bigl( - \frac{1}{\Delta}
			- \frac{2}{\Delta^2} + \frac{2 (e^\Delta - 1)}{\Delta^3} \Bigr)
		&&= \ e_{}^m
			\bigl( \sinh(\delta)/\delta - (1 + \delta) e_{}^{-\delta} \bigr)
				/ (2\delta^2) , \\
	J_{11}(r,s) \
	&= \ e_{}^r \int_0^1 (1 - v)v e_{}^{\Delta v} \, dv
		\ = \ e_{}^r \Bigl( \frac{e^{\Delta} + 1}{\Delta^2}
			- \frac{2(e^\Delta - 1)}{\Delta^3} \Bigr)
		&&= \ e_{}^m
			\bigl( \cosh(\delta) - \sinh(\delta)/\delta \bigr)
				/ (2\delta^2) .
\end{align*}

If $|\delta|$ is close to $0$, the formulae above get problematic. Here is a reasonable approximation for small values of $|\delta|$: For integers $a, b \ge 0$ let $B_{ab} := \int_0^1 u^a (1 - u)^b \, du = a! b! / (a+b+1)!$, and let $U_{ab}$ be a random variable with distribution $\mathrm{Beta}(a+1,b+1)$, so
\begin{align*}
	\mu_{ab} := \Ex U_{ab} \
	&= \ \frac{a+1}{a+b+2} , \\
	\sigma_{ab}^2 := \mathrm{Var}(U_{ab}) \
	&= \ \frac{(a+1)(b+1)}{(a+b+2)^2 (a+b+3)} , \\
	\gamma_{ab} := \Ex \bigl( (U_{ab} - \mu_{ab})^3 \bigr) \
	&= \ \frac{2(a+1)(b+1)(b-a)}{(a+b+2)^3 (a+b+3) (a+b+4)} .
\end{align*}
Then
\[
	J_{ab}(r,s)
	\ = \ B_{ab} \Ex \exp \bigl( U_{ab} r + (1 - U_{ab}) s \bigr)
	\ = \ B_{ab} \exp \bigl( \mu_{ab} r + (1 - \mu_{ab}) s \bigr)
		\Ex \exp \bigl( (U_{ab} - \mu_{ab}) (r - s) \bigr) ,
\]
and
\begin{align*}
	\log \Ex \exp \bigl( (U_{ab} - \mu_{ab}) (r - s) \bigr) \
	&= \ \frac{\sigma_{ab}^2 (r - s)^2}{2} + \frac{\gamma_{ab} (r - s)^3}{6}
		+ O(|r - s|^4)
\end{align*}
as $|r - s| \to 0$. Hence
\begin{align*}
	J_{ab}(r,s) \ = \
	&\frac{a! b!}{(a+b)! (a+b+1)}
		\cdot \exp \Bigl( \frac{(a+1)r + (b+1) s}{a+b+2} \\
	& \qquad + \, \frac{(a+1)(b+1) (r-s)^2}{2 (a+b+2)^2 (a+b+3)}
			+ \frac{(a+1)(b+1)(b-a) (r-s)^3}{3 (a+b+2)^3 (a+b+3) (a+b+4)} \Bigr)
		\cdot \bigl( 1 + O(|r - s|^4) \bigr)
\end{align*}
as $|r - s| \to 0$. Specifically,
\begin{align*}
	J(r,s) \
	&\approx \ \exp \bigl( (r + s)/2 + (r - s)^2/24 \bigr) , \\
	J_{10}(r,s) \
	&\approx \ 2^{-1} \exp \bigl( (2r + s)/3 + (r - s)^2/36
		- (r - s)^3/810 \bigr) , \\
	J_{20}(r,s) \
	&\approx \ 3^{-1} \exp \bigl( (3r + s)/4 + 3(r - s)^2/160
		- (r - s)^3/960 \bigr) , \\
	J_{11}(r,s) \
	&\approx \ 6^{-1} \exp \bigl( (r + s)/2 + (r - s)^2/40 \bigr) .
\end{align*}
Numerical experiments show that the relative error of these approximations is less than $10^{-10}$ for $|r - s| \le 0.01$.

\paragraph{Local parametrizations.}
Let us fix arbitrary points $\tau_1 < \cdots < \tau_m$ in $\{x_1,\ldots,x_n\}$ with $\tau_1 = x_1$ and $\tau_m = x_n$. Any function $\theta : \R \to [-\infty,\infty)$ which is linear on each interval $[\tau_j, \tau_{j+1}]$, $1 \le j < m$, and satisfies $\theta \equiv - \infty$ of $\R \setminus [\tau_1,\tau_m]$ is uniquely determined by the vector $\thb = (\theta_j)_{j=1}^m := (\theta(\tau_j))_{j=1}^m \in \R^m$. Then $L(\theta) = L(\taub,\thb)$ with $L(\taub,\cdot) : \R^m \to \R$ given by
\begin{equation}
\label{eq:LocalL.Setting1}
	L(\taub,\thb)
	\ := \ \sum_{i=1}^n w_i \theta(x_i)
		- \sum_{j=1}^{m-1} (\tau_{j+1} - \tau_j) J(\theta_j, \theta_{j+1}) + 1
	\ = \ \sum_{j=1}^m \tilde{w}_j \theta_j
		- \sum_{j=1}^{m-1} (\tau_{j+1} - \tau_j) J(\theta_j, \theta_{j+1}) + 1
\end{equation}
with the auxiliary function $J(\cdot,\cdot)$ defined in~\eqref{Eq:AS_AuxiliaryJ} and the weights
\[
	\tilde{w}_j \
	:= \ 1_{[j = 1]} w_1 + \sum_{i=1}^n \Bigl(
		1_{[j > 1, \, x_i \le \tau_j]} \,
			\frac{(x_i - \tau_{j-1})^+}{\tau_j - \tau_{j-1}}
		+ 1_{[j < m, \, x_i > \tau_j]} \,
			\frac{(\tau_{j+1} - x_i)^+}{\tau_{j+1} - \tau_j} \Bigr) w_i .
\]
The function $L(\taub,\cdot)$ on $\R^m$ is twice continuously differentiable with negative definite Hessian matrix, see the next paragraph.

\paragraph{Gradient vector and Hessian matrix of $L(\taub,\thb)$ in \eqref{eq:LocalL.Setting1}.}
For fixed $\taub$ and as a function of $\thb \in \R^m$, $L(\taub,\thb)$ has gradient vector $\nabla L(\taub,\thb) =: \bs{g}(\taub,\thb)$ with components
\[
	g_j(\taub,\thb) \ = \ \tilde{w}_j
		- 1_{[j < m]} (\tau_{j+1} - \tau_j) J_{10}(\theta_j,\theta_{j+1})
		- 1_{[j > 1]} (\tau_j - \tau_{j-1}) J_{10}(\theta_j,\theta_{j-1})
\]
and negative Hessian matrix $- D^2 L(\taub,\thb) =: \bs{H}(\taub,\thb)$ with components
\begin{align*}
	H_{jj}(\taub,\thb) \
	&= \ 1_{[j < m]} (\tau_{j+1} - \tau_j) J_{20}(\theta_j,\theta_{j+1})
		+ 1_{[j > 1]} (\tau_j - \tau_{j-1}) J_{20}(\theta_j,\theta_{j-1}) , \\
	H_{j,j+1}(\taub,\thb) = H_{j+1,j}(\taub,\thb) \
	&= \ (\tau_{j+1} - \tau_j) J_{11}(\theta_j,\theta_{j+1}) , \\
	H_{jk}(\taub,\thb) \
	&= \ 0 \quad \text{if} \ |k - j| \ge 2 .
\end{align*}
Note also that
\[
	\bs{g}(\taub,\thb)^\top \bs{\delta}
	\ = \ \int_{[x_1,x_n]} \delta(x) \, (\hat{P}(dx) - e_{}^{\theta(x)}\,dx)
	\quad\text{and}\quad
	\bs{\delta}^\top \bs{H}(\taub,\thb) \bs{\delta}
	\ = \ \int_{[x_1,x_n]} \delta(x)^2 e_{}^{\theta(x)} \, dx ,
\]
the last equality showing positive definiteness of $\bs{H}(\taub,\thb)$.

\paragraph{Evaluating the directional derivative $DL(\theta,V_{\tau,\theta})$.}
If $\theta \in \V$ with $\{x_1,x_n\} \cup D(\theta)$ having elements $\tau_1 < \cdots < \tau_m$, then for $1 \le j < m$ and $\tau_j \le \tau \le \tau_{j+1}$,
\begin{align*}
	DL(\theta,V_{\tau,\theta}) \
	&= \ \sum_{i=1}^n V_{\tau,\theta}(x_i) w_i
		- \frac{(\tau - \tau_j)(\tau_{j+1} - \tau)}{\tau_{j+1} - \tau_j}
			\int_{\tau_j}^{\tau_{j+1}}
				\bigl( j_{01}(x; \tau_j, \tau) + j_{10}(x; \tau, \tau_{j+1}) \bigr)
					e_{}^{\theta(x)} \, dx \\
	&= \ \sum_{i=1}^n V_{\tau,\theta}(x_i) w_i
		- \frac{(\tau - \tau_j)(\tau_{j+1} - \tau)}{\tau_{j+1} - \tau_j}
			\bigl( (\tau - \tau_j) J_{10}(\theta_*,\theta_j)
				+ (\tau_{j+1} - \tau) J_{10}(\theta_*,\theta_{j+1}) \bigr)
\end{align*}
with
\[
	\theta_* \ := \ \theta(\tau)
	\ = \ \frac{(\tau_{j+1} - \tau) \theta_j + (\tau - \tau_j) \theta_{j+1}}
		{\tau_{j+1} - \tau_j} .
\]

\paragraph{Activating one constraint.}
Suppose that $m \ge 3$ in \eqref{eq:LocalL.Setting1}. If we activate the constraint at $\tau_{j_o}$, where $1 < j_o < m$, this amounts to replacing $(\tilde{w}_{j_o-1}, \tilde{w}_{j_o}, \tilde{w}_{j_o+1})$ with
\[
	\Bigl( \tilde{w}_{j_o-1} + \frac{\tau_{j_o+1} - \tau_{j_o}}{\tau_{j_o+1} - \tau_{j_o-1}}
			\, \tilde{w}_{j_o},
		\ 0, \
		\tilde{w}_{j_o+1} + \frac{\tau_{j_o} - \tau_{j_o-1}}{\tau_{j_o+1} - \tau_{j_o-1}}
			\, \tilde{w}_{j_o} \Bigr)
\]
and then removing the $j_o$-th components of $\bs{\tau}$ and $(\tilde{w}_j)_{j=1}^m$.

\subsection{Details for Setting~2A}
\label{subsec:Details.2A}

We provide explicit formulae for the special case of $P_o = \NN(0,1)$ with Lebesgue density $\phi$ and distribution function $\Phi$.

\paragraph{Auxiliary functions.}
The subsequent formulae follow from tedious but elementary algebra, the essential ingredients being
\[
	e^{\theta x} \phi(x) \ = \ e^{\theta^2/2} \phi(x - \theta)
	\quad\text{for} \ x,\theta \in \R
\]
and
\begin{align*}
	\int \phi(z) \, dz \ &= \ C + \Phi(z) , \\
	\int z \phi(z) \, dz \ &= \ C - \phi(z) , \\
	\int z^2 \phi(z) \, dz \ &= \ C - z \phi(z) + \Phi(z) .
\end{align*}

On the one hand, for a fixed number $a \in \R$ let
\begin{equation}
	\label{Eq:AS_AuxiliaryK}
	K(\theta_0, \theta_1) = K(\theta_0, \theta_1; a)
	\ := \ \int_a^\infty e_{}^{\theta_0 + \theta_1 (x - a)} \phi(x) \, dx .
\end{equation}
Then
\[
	K(\theta_0, \theta_1)
	\ = \ e_{}^{\theta_0 - \theta_1 a + \theta_1^2/2} \, \Phi(\theta_1 - a)
	\ = \ \frac{\partial K(\theta_0,\theta_1)}{\partial \theta_0} ,
\]
and explicit expressions for
\[
	K_\ell(\theta_0,\theta_1)
	\ := \ \frac{\partial^\ell K(\theta_0,\theta_1)}{\partial \theta_1^\ell}
	\ = \ \int_a^\infty
	 	(x - a)_{}^\ell e_{}^{\theta_0 + \theta_1 (x - a)} \phi(x) \, dx
\]
are given by
\begin{align*}
	K_1(\theta_0,\theta_1) \
	&= \ e_{}^{\theta_0 - \theta_1 a + \theta_1^2/2}
		\bigl( (\theta_1 - a) \Phi(\theta_1 - a) + \phi(\theta_1 - a) \bigr) , \\
	K_2(\theta_0,\theta_1) \
	&= \ e_{}^{\theta_0 - \theta_1 a + \theta_1^2/2}
		\Bigl( \bigl( 1 + (\theta_1 - a)^2 \bigr) \Phi(\theta_1 - a)
			+ (\theta_1 - a) \phi(\theta_1 - a) \Bigr) .
\end{align*}
Moreover,
\[
	\int_{-\infty}^a e_{}^{\theta_0 + \theta_1 (x - a)} \phi(x) \, dx
	\ = \ K(\theta_0, -\theta_1; -a) .
\]

On the other hand, for fixed real numbers $a < b$ let
\begin{equation}
	\label{Eq:AS_AuxiliaryJTheta}
	J(\theta_0,\theta_1) = J(\theta_0,\theta_1; a, b)
	\ := \ \int_a^b \exp \Bigl( \frac{b-x}{b-a} \, \theta_0
		+ \frac{x-a}{b-a} \, \theta_1 \Bigr) \phi(x) \, dx .
\end{equation}
With
\[
	\tilde{\theta}_0 \ := \ \frac{b\theta_0 - a \theta_1}{b - a} ,
	\quad
	\tilde{\theta}_1 \ := \ \frac{\theta_1 - \theta_0}{b - a}
	\quad\text{and}\quad
	\tilde{b} \ := \ b - \tilde{\theta}_1 ,
	\quad
	\tilde{a} \ := \ a - \tilde{\theta}_1
\]
we may write
\[
	J(\theta_0,\theta_1)
	\ = \ e_{}^{\tilde{\theta}_0 + \tilde{\theta}_1^2/2}
		\bigl( \Phi(\tilde{b}) - \Phi(\tilde{a}) \bigr).
\]
Furthermore, explicit expressions for
\[
	J_{\ell m}(\theta_0,\theta_1)
	\ := \ \frac{\partial^{\ell + m} J(\theta_0,\theta_1)}
		{\partial \theta_0^\ell \partial \theta_1^m}
	\ = \ \int_a^b \frac{(b - x)^\ell (x - a)^m}{(b - a)^{\ell+m}} \,
		\exp \Bigl( \frac{b - x}{b - a} \, \theta_0
			+ \frac{x - a}{b - a} \, \theta_1 \Bigr) \phi(x) \, dx
\]
for $\ell,m\in\{0,1,2\}$ with $1\le \ell + m\le 2$ are given by
\begin{align*}
	J_{10}(\theta_0,\theta_1) \
	&= \ e_{}^{\tilde{\theta}_0 + \tilde{\theta}_1^2/2} \,
		\frac{\tilde{b} \bigl( \Phi(\tilde{b}) - \Phi(\tilde{a}) \bigr)
			+ \phi(\tilde{b}) - \phi(\tilde{a})}{b - a} , \\
	J_{01}(\theta_0,\theta_1) \
	&= \ J_{10}(\theta_1,\theta_0; -b, -a) , \\
	J_{20}(\theta_0,\theta_1) \
	&= \ e_{}^{\tilde{\theta}_0 + \tilde{\theta}_1^2/2} \,
		\frac{(1 + \tilde{b}^2) \bigl( \Phi(\tilde{b}) - \Phi(\tilde{a}) \bigr)
			+ (\tilde{a} - 2 \tilde{b}) \phi(\tilde{a})
			+ \tilde{b} \phi(\tilde{b})}{(b - a)^2} , \\
	J_{11}(\theta_0,\theta_1) \
	&= \ e_{}^{\tilde{\theta}_0 + \tilde{\theta}_1^2/2} \,
		\frac{- (1 + \tilde{a}\tilde{b}) \bigl( \Phi(\tilde{b}) - \Phi(\tilde{a}) \bigr)
			+ \tilde{b} \phi(\tilde{a}) - \tilde{a} \phi(\tilde{b})}{(b - a)^2} , \\
	J_{02}(\theta_0,\theta_1) \
	&= \ e_{}^{\tilde{\theta}_0 + \tilde{\theta}_1^2/2} \,
		\frac{(1 + \tilde{a}^2) \bigl( \Phi(\tilde{b}) - \Phi(\tilde{a}) \bigr)
			+ (2 \tilde{a} - \tilde{b}) \phi(\tilde{b})
			- \tilde{a} \phi(\tilde{a})}{(b - a)^2} .
\end{align*}
In case of $\tilde{a} > 0$, the right hand side of the equation
\[
	\Phi(\tilde{b}) - \Phi(\tilde{a})
	\ = \ \Phi(-\tilde{a}) - \Phi(-\tilde{b})
\]
is numerically more accurate than its left-hand side. In connection with $J(\theta_0,\theta_1)$ we also use the the lower bound
\begin{align*}
	\log(\Phi(\tilde{b}) - \Phi(\tilde{a})) \
	&= \ - \frac{\tilde{m}^2}{2}
		+ \log \int_{-\tilde{d}}^{\tilde{d}} \exp(\tilde{m}z) \phi(z) \, dz
		\ \ge \ - \frac{\tilde{m}^2}{2}
			+ \log \bigl( \Phi(\tilde{d}) - \Phi(-\tilde{d}) \bigr)
\end{align*}
with $\tilde{m} := (\tilde{a} + \tilde{b})/2$ and $\tilde{d} := (\tilde{b} - \tilde{a})/2$. The bound follows from $\exp(\tilde{m}z) \ge 1 + \tilde{m} z$.

\paragraph{Local parametrizations.}
Let us fix any vector $\taub$ with $m \ge 1$ components $\tau_1 < \cdots < \tau_m$ in $(x_1,x_n)$. Any function $\theta$ which is linear on the intervals $\XX_0,\XX_1,\ldots,\XX_m$ specified in Lemma~\ref{lem:existence.uniqueness.2A} is uniquely determined by the vector
\[
	\thb = (\theta_j)_{j=0}^{m+1} \
	:= \ \bigl( \theta'(\tau_1\,-),
		\theta(\tau_1),\ldots,\theta(\tau_m),
		\theta'(\tau_m\,+) \bigr)^\top
	\ \in \ \R^{m+2} .
\]
Then $L(\theta)$ is given by
\begin{align}
\nonumber
	L(\taub,\thb) \
	:= \ &\sum_{i=1}^n w_i \theta(x_i)
		- \int_{\XX_0} e_{}^{\theta(x)} \, P_o(dx)
		- \sum_{j=1}^m \int_{\XX_j} e_{}^{\theta(x)} \, P_o(dx)
		+ 1 \\
\label{eq:LocalL.Setting2A}
	= \ &\sum_{j=0}^{m+1} \tilde{w}_j \theta_j
		- K(\theta_1, - \theta_0; - \tau_1)
		- \sum_{1 \le j < m} J(\theta_j,\theta_{j+1}; \tau_j, \tau_{j+1})
		- K(\theta_m, \theta_{m+1}; \tau_m) + 1 .
\end{align}
with the auxiliary functions $K(\cdot,\cdot;\cdot)$ and $J(\cdot,\cdot;\cdot,\cdot)$ introduced in~\eqref{Eq:AS_AuxiliaryK} and \eqref{Eq:AS_AuxiliaryJTheta} and the `weights'
\begin{align*}
	\tilde{w}_0 \
	&:= \ - \sum_{i=1}^n (\tau_1 - x_i)^+ w_i , \\
	\tilde{w}_1 \
	&:= \ \sum_{i=1}^n
		\min \Bigl( 1, \frac{(\tau_2 - x_i)^+}{\tau_2 - \tau_1} \Bigr) \, w_i , \\
	\tilde{w}_j \
	&:= \ \sum_{i=1}^n \Bigl(
		1_{[x_i \le \tau_j]} \,
			\frac{(x_i - \tau_{j-1})^+}{\tau_j - \tau_{j-1}}
		+ 1_{[x_i > \tau_j]} \,
			\frac{(\tau_{j+1} - x_i)^+}{\tau_{j+1} - \tau_j} \Bigr) w_i
		\quad\text{for} \ 1 < j < m , \\
	\tilde{w}_m \
	&:= \ \sum_{i=1}^n 
		\min \Bigl( 1, \frac{(x_i - \tau_{m-1})^+}{\tau_m - \tau_{m-1}} \Bigr) \, w_i, \\
	\tilde{w}_{m+1} \
	&:= \ \sum_{i=1}^n (x_i - \tau_m)^+ \, w_i .	
\end{align*}
In case of $m = 1$, the weight $\tilde{w}_1$ is just given by $\tilde{w}_1 = 1$.

The function $L(\taub,\cdot): \R^{m+2} \to \R$ is twice continuously differentiable with negative definite Hessian matrix, see the next paragraph.

\paragraph{Gradient vector and Hessian matrix for $L(\taub,\cdot)$ in \eqref{eq:LocalL.Setting2A}.}
In case of $m \ge 2$, the gradient $\bs{g}(\taub,\thb) = \bigl( g_j(\taub,\thb) \bigr)_{j=0}^{m+1}$ of $L(\taub,\cdot)$ equals
\[
	g_j(\taub, \thb) \ = \ \tilde{w}_j - \begin{cases}
		- K_1(\theta_1, - \theta_0; - \tau_1)
			& \text{if} \ j = 0 , \\
		K(\theta_1, - \theta_0; - \tau_1)
			+ J_{10}(\theta_1,\theta_2; \tau_1, \tau_2)
			& \text{if} \ j = 1 , \\
		J_{01}(\theta_{j-1},\theta_j; \tau_{j-1}, \tau_j)
			+ J_{10}(\theta_j, \theta_{j+1}; \tau_j, \tau_{j+1})
			& \text{if} \ 2 < j < m , \\
		J_{01}(\theta_{m-1},\theta_m; \tau_{m-1}, \tau_m)
			+ K(\theta_m, \theta_{m+1}; \tau_m)
			& \text{if} \ j = m , \\
		K_1(\theta_m, \theta_{m+1}; \tau_m)
			& \text{if} \ j = m+1 ,
	\end{cases}	
\]
while its negative Hessian matrix $\bs{H}(\taub,\thb) = \bigl( H_{jk}(\taub,\thb) \bigr)_{j,k=0}^{m+1}$ is given by
\begin{align*}
	H_{00}(\taub,\thb) \
	&= \, K_2(\theta_1, - \theta_0; - \tau_1) , \\
	H_{01}(\taub,\thb) = H_{10}(\taub,\thb) \
	&= \, - K_1(\theta_1, - \theta_0; - \tau_1) , \\
	H_{11}(\taub,\thb) \
	&= \, K(\theta_1, - \theta_0; - \tau_1)
		+ J_{20}(\theta_1, \theta_2; \tau_1, \tau_2) , \\
	H_{j,j+1}(\taub,\thb) = H_{j+1,j}(\taub,\thb) \
	&= \, J_{11}(\theta_j, \theta_{j+1}; \tau_j, \tau_{j+1})
		\ \, \text{for} \ 1 \le j < m , \\
	H_{jj}(\taub,\thb) \
	&= \, J_{02}(\theta_{j-1}, \theta_j; \tau_{j-1}, \tau_j)
		+ J_{20}(\theta_j, \theta_{j+1}; \tau_j, \tau_{j+1})
		\ \, \text{for} \ 1 < j < m , \\
	H_{mm}(\taub,\thb) \
	&= \, J_{02}(\theta_{m-1},\theta_m; \tau_{m-1},\tau_m)
		+ K(\theta_m,\theta_{m+1}; \tau_m) \\
	H_{m,m+1}(\taub,\thb) = H_{m+1,m}(\taub,\thb) \
	&= \, K_1(\theta_m, \theta_{m+1}; \tau_m) , \\
	H_{m+1,m+1}(\taub,\thb) \
	&= \, K_2(\theta_m,\theta_{m+1}; \tau_m), \\
	H_{jk}(\taub,\thb) \
	&= \, 0 \ \, \text{if} \ |j - k| \ge 2 .
\end{align*}

In case of $m = 1$ we get the simplified formulae
\[
	L(\taub,\thb) \ = \ \sum_{j=0}^2 \tilde{w}_j \theta_j
		- K(\theta_1, - \theta_0; - \tau_1) - K(\theta_1, \theta_2; \tau_1) + 1 ,
\]
\[
	g_j(\taub,\thb) \ = \ \tilde{w}_j - \begin{cases}
		- K_1(\theta_1, - \theta_0; - \tau_1)
			& \text{if} \ j = 0 , \\
		K(\theta_1, - \theta_0; - \tau_1) + K(\theta_1, \theta_2; \tau_1)
			& \text{if} \ j = 1 , \\
		K_1(\theta_1, \theta_2; \tau_1)
			& \text{if} \ j = 2 ,
	\end{cases}
\]
and
\begin{align*}
	H_{00}(\taub,\thb) \
	&= \ K_2(\theta_1, - \theta_0; - \tau_1) , \\
	H_{01}(\taub,\thb) = H_{10}(\taub,\thb) \
	&= \ - K_1(\theta_1, - \theta_0; - \tau_1) , \\
	H_{11}(\taub,\thb) \
	&= \ K(\theta_1, - \theta_0; - \tau_1) + K(\theta_1,\theta_2; \tau_2) \\
	H_{12}(\taub,\thb) = H_{21}(\taub,\thb) \
	&= \ K_1(\theta_1, \theta_2; \tau_1) , \\
	H_{22}(\taub,\thb) \
	&= \ K_2(\theta_1,\theta_2; \tau_1) .
\end{align*}

\paragraph{Evaluating $h_\theta(\tau) := DL(\theta,V_{\tau,\theta})$ and $h_\theta'(\tau\,+)$.}
Suppose first that $\theta(x) = \hat{\mu} x - \hat{\mu}^2/2$, so $P_\theta = \NN(\hat{\mu},1)$ and $D(\theta) = \emptyset$. Then one can show that
\begin{align*}
	h_\theta'(\tau\,+) \
	&= \ \hat{F}(\tau) - \Phi(\tau - \hat{\mu}) , \\
	h_\theta(\tau) \
	&= \ \tau h_\theta'(\tau\,+) - \int_{(-\infty,\tau]} x \, \hat{P}(dx)
		+ \hat{\mu} \Phi(\tau - \hat{\mu}) - \phi(\tau - \hat{\mu}) .
\end{align*}

Now suppose that $\theta$ is given by a vector $\taub$ of $m \ge 1$ points $\tau_1 < \cdots < \tau_m$ and a vector $\thb = (\theta_j)_{j=0}^{m+1}$ as in \eqref{eq:LocalL.Setting2A}. Then for $\tau < \tau_1$,
\begin{align*}
	h_\theta'(\tau\,+) \
	&= \ \hat{F}(\tau) - K(\theta_*, - \theta_0; - \tau) , \\
	h_\theta(\tau) \
	&= \ (\tau - \tau_1)
		\bigl( h_\theta'(\tau\,+) - J_{10}(\theta_*, \theta_1; \tau, \tau_1) \bigr)
		- \int 1_{[\tau < x \le \tau_1]}(\tau_1 - x) \hat{P}(dx) ,
\end{align*}
where $\theta_* := \theta(\tau) = \theta_1 + (\tau - \tau_1) \theta_0$. For $1 \le j < m$ and $\tau \in [\tau_j,\tau_{j+1})$,
\begin{align*}
	h_\theta'(\tau\,+) \
	&= \ \hat{P}((\tau_j,\tau])
		- J(\theta_j, \theta_*; \tau_j, \tau)
		- \int \ j_{10}(x; \tau_j, \tau_{j+1}) \, \hat{P}(dx)
		+ J_{10}(\theta_j,\theta_{j+1}; \tau_j,\tau_{j+1}) , \\
	h_\theta(\tau) \
	&= \ (\tau - \tau_j)
		\bigl( h_\theta'(\tau\,+) + J_{01}(\theta_j, \theta_*; \tau_j, \tau) \bigr)
		- \int 1_{[\tau_j < x \le \tau]} (x - \tau_j) \, \hat{P}(dx) ,
\end{align*}
where $\theta_* := \theta(\tau) = (\tau_{j+1} - \tau_j)^{-1} \bigl( (\tau_{j+1} - \tau) \theta_j + (\tau - \tau_j) \theta_{j+1} \bigr) = \theta_j + (\tau - \tau_j) \theta_j'$. Finally, for $\tau > \tau_m$,
\begin{align*}
	h_\theta'(\tau\,+) \
	&= \ K(\theta_*, \theta_{m+1}; \tau) - \hat{P}((\tau,\infty)) , \\
	h_\theta(\tau) \
	&= \ (\tau - \tau_m)
		\bigl( h_\theta'(\tau\,+) + J_{01}(\theta_m,\theta_*; \tau_m, \tau) \bigr)
		- \int 1_{[\tau_m < x \le \tau]} (x - \tau_m) \hat{P}(dx) ,
\end{align*}
where $\theta_* := \theta_m + (\tau - \tau_m) \theta_{m+1}$.

If $\tau$ is restricted to some interval $I$ not containing any observations $x_i$ or knots $\tau_j$, the latter expressions for $h_\theta'(\tau\,+)$ are constant in $\tau$ except for one term $K(\theta_*, - \theta_0; - \tau)$, $J(\theta_j, \theta_*; \tau_j, \tau)$ or $K(\theta_*, \theta_{m+1}; \tau)$. Hence finding $\tau$ such that $h_\theta'(\tau\,+) = 0$ leads to equations of the following type: For given real numbers $\theta_0, \theta_1, \tau_0$ and $c$, find $\tau \in \R$ such that
\begin{align}
\label{eq:inverse.1}
	K \bigl( \theta_0 + \theta_1(\tau - \tau_0), \pm \theta_1; \pm \tau \bigr) \
	&= \ c , \\
\label{eq:inverse.2}
	J \bigl( \theta_0, \theta_0 + \theta_1(\tau - \tau_0); \tau_0, \tau \bigr) \
	&= \ c ,
\end{align}
and check whether $\tau \in I$. Since $K \bigl( \theta_0 + \theta_1(\tau - \tau_0), \pm\theta_1; \pm \tau \bigr)$ equals $e_{}^{\theta_0 - \theta_1\tau_0 +\theta_1^2/2} \Phi(\mp (\tau - \theta_1))$, the unique solution of \eqref{eq:inverse.1} is given by
\[
	\tau \ = \ \theta_1 \mp \Phi^{-1}(e_{}^{-\theta_0 + \theta_1\tau_0 - \theta_1^2/2} c) ,
\]
provided that $c > 0$ and $c e^{-\theta_0 + \theta_1\tau_0 - \theta_1^2/2} < 1$; otherwise no solution exists. Likewise, since $J \bigl( \theta_0, \theta_0 + \theta_1(\tau - \tau_0); \tau_0, \tau \bigr)$ equals $e_{}^{\theta_0 - \theta_1\tau_0 + \theta_1^2/2} \bigl( \Phi(\tau - \theta_1) - \Phi(\tau_0 - \theta_1) \bigr)$, the unique solution of \eqref{eq:inverse.2} is given by
\[
	\tau \ = \ \theta_1
		+ \Phi^{-1} \bigl( \Phi(\tau_0 - \theta_1)
			+ e_{}^{-\theta_0 + \theta_1\tau_0 - \theta_1^2/2} c \bigr) ,
\]
provided that $0 < \Phi(\tau_0 - \theta_1) + c e^{-\theta_0 + \theta_1\tau_0 - \theta_1^2/2} < 1$; otherwise no solution exists.

\paragraph{Activating one constraint.}
Suppose that $m \ge 2$ in \eqref{eq:LocalL.Setting2A}. If even $m \ge 3$, and if we activate the constraint at $\tau_{j_o}$, where $1 < j_o < m$, the update of $\bs{\tau}$ and $(\tilde{w}_j)_{j=0}^{m+1}$ is essentially the same as in Setting~1. If we activate the constraint at $\tau_1$, this amounts to replacing $(\tau_1,\tau_2)$ and $(\tilde{w}_0, \tilde{w}_1, \tilde{w}_{j_o+1})$ with
\[
	(\tau_2)
	\quad\text{and}\quad
	\bigl( \tilde{w}_0 - (\tau_2 - \tau_1) \tilde{w}_1, \,
		\tilde{w}_1 + \tilde{w}_2 \bigr) ,
\]
respectively. Similary, activating the constraint at $\tau_m$ amounts to replacing $(\tau_{m-1}, \tau_m)$ with
\[
	(\tau_{m-1})
	\quad\text{and}\quad
	\bigl( \tilde{w}_{m-1} + \tilde{w}_m, \,
		\tilde{w}_{m+1} + (\tau_m - \tau_{m-1}) \tilde{w}_m \bigr) ,
\]
respectively.

\subsection{Details for Setting~2B}
\label{subsec:Details.2B}

We provide explicit formulae for the special case of $P_o$ being a gamma distribution with shape parameter $\alpha>0$ and rate parameter $\beta=1$, i.e. $P_o$ has density
\[
	p_o(x) \ = \ \Gamma(\alpha)^{-1}x^{\alpha-1}e^{-x},\qquad x>0.
\]
Note that the case of a gamma distribution with rate parameter $\beta\neq 1$ may be reduced to the case $\beta=1$ by multiplying all observations with $\beta$, then estimating the function $\theta$ by $\hat{\theta}_{\text{temp}}$ and finally setting $\hat{\theta}(x):=\hat{\theta}_{\text{temp}}(x/\beta)$.

\paragraph{Auxiliary functions.}
For $s>0$, the c.d.f.\ of a gamma distribution with shape $s$ and rate $1$ is the function $G_s:[0,\infty] \to [0,1]$ defined by
\[
	G_s(x) \ := \ \Gamma(s)^{-1} \int_{0}^x z^{s-1}e^{-z}\,dz,
\]
and, for $0\leq a< b \leq \infty$, we define the partial integral
\[
	G_s(a,b) \ := \ \Gamma(s)^{-1} \int_{a}^b z^{s-1}e^{-z}\,dz \ = \ G_s(b) - G_s(a).
\]

On the one hand, for a fixed number $c \in \R$ let
\[
	K(\theta_0,\theta_1) 
	\  = \ K(\theta_0,\theta_1;c)
	\ := \ \int_{c}^\infty e^{\theta_0+\theta_1(x - c)}p_o(x) \, dx .
\]
This is equal to $\infty$ in case of $\theta_1 \ge 1$. Otherwise, when $\theta_1<1$, let $\tilde{c} := (1 - \theta_1) c$. Then
\[
	K(\theta_0,\theta_1)
	\ = \ \frac{e^{\theta_0-\theta_1c}}{(1 - \theta_1)^\alpha} \,
		G_\alpha(\tilde{c},\infty)
	\ = \ \frac{\partial K(\theta_0,\theta_1)}{\partial\theta_0} ,
\]
and explicit expressions for
\[
	K_\ell(\theta_0,\theta_1) 
	\ := \ \frac{\partial^\ell K(\theta_0,\theta_1)}{\partial\theta_1^\ell}
	\  = \ \int_c^\infty (x-c)^\ell e^{\theta_0+\theta_1(x-c)} p_o(x) \, dx
\]
are given by
\begin{align*}
	K_1(\theta_0,\theta_1) \
	&= \ \frac{e^{\theta_0-\theta_1c}}{(1-\theta_1)^{\alpha+1}}
		\bigl( \alpha G_{\alpha+1}(\tilde{c},\infty)
			-\tilde{c} G_\alpha(\tilde{c},\infty) \bigr) , \\
	K_2(\theta_0,\theta_1) \
	&= \ \frac{e^{\theta_0-\theta_1c}}{(1-\theta_1)^{\alpha+2}}
		\bigl( \alpha(\alpha+1)  G_{\alpha+2}(\tilde{c},\infty)
			- 2\alpha \tilde{a} G_{\alpha+1}(\tilde{c},\infty)
			+ \tilde{c}^2 G_{\alpha}(\tilde{c},\infty) \bigr) .
\end{align*}

On the other hand, for fixed numbers $0\leq a<b < \infty$ let
\[
	J(\theta_0,\theta_1) = J(\theta_0,\theta_1;a,b)
	\ = \ \int_a^b \exp \Bigl( \frac{b-x}{b-a}\theta_0+\frac{x-a}{b-a}\theta_1 \Bigr)
			p_o(x) \, dx
	\ = \ \frac{e^{\tilde{\theta}_0}}{\Gamma(\alpha)}
		\int_a^b e_{}^{(\tilde{\theta}_1 - 1)x} x_{}^{\alpha-1} \, dx  ,	
\]
where
\[
	\tilde{\theta}_0 \ := \ \frac{b\theta_0-a\theta_1}{b-a}
	\quad\text{and}\quad
	\tilde{\theta}_1 \ := \ \frac{\theta_1-\theta_0}{b-a} .
\]
With $\tilde{a} := (1 - \tilde{\theta}_1) a$ and $\tilde{b} := (1 - \tilde{\theta}_1) b$ we may write
\[
	J(\theta_0,\theta_1)
	\ = \ \begin{cases}
		\displaystyle
		\frac{e^{\tilde{\theta}_0}
				G_\alpha(\tilde{a},\tilde{b})}{(1-\tilde{\theta}_1)^{\alpha}}
		&\text{if} \ \tilde{\theta} < 1 , \\[2ex]
		\displaystyle
		\frac{e^{\tilde{\theta}_0} (b^\alpha - a^\alpha)}{\Gamma(\alpha + 1)}
		&\text{if} \ \tilde{\theta} = 1 .
	\end{cases}
\]
Note that in our specific applications the slope parameter $\tilde{\theta}_1$ corresponds to the difference ratio $\bigl( \theta(b) - \theta(a) \bigr) /(b - a)$ of a function $\theta \in \V$. Thus it will be strictly smaller than $1$ as soon as $\theta \in \Theta$ and $L(\theta) > - \infty$. During a Newton step the latter conditions may be violated temporarily, so in case of $\tilde{\theta}_1 > 1$ we use the simple bound
\[
	J(\theta_0,\theta_1)
	\ \le \ \frac{e^{\tilde{\theta}_0 + (\tilde{\theta}_1 - 1)b} (b^\alpha - a^\alpha)}
		{\Gamma(\alpha + 1)} .
\]
In case of $\tilde{\theta}_1 < 1$, explicit expressions for
\[
	J_{\ell m}(\theta_0,\theta_1) 
	\ := \ \frac{\partial^{\ell+m}J(\theta_0,\theta_1)}
		{\partial\theta_0^\ell\partial\theta_1^m}
	\ = \ \int_a^b \frac{(b-x)^\ell(x-a)^m}{(b-a)^{\ell+m}}
		\exp \Bigl( \frac{b-x}{b-a}\theta_0 + \frac{x-a}{b-a}\theta_1 \Bigr)
			p_o(x) \, dx
\]
are given by
\begin{align*}
	J_{10}(\theta_0,\theta_1) \
	&= \ \frac{e^{\tilde{\theta}_0}}{(1-\tilde{\theta}_1)^{\alpha+1}} \,
		\frac{\tilde{b}G_\alpha(\tilde{a},\tilde{b})
			- \alpha G_{\alpha+1}(\tilde{a},\tilde{b})}{b-a} , \\
	J_{01}(\theta_0,\theta_1) \
	&= \ \frac{e^{\tilde{\theta}_0}}{(1-\tilde{\theta}_1)^{\alpha+1}} \,
		\frac{-\tilde{a}G_\alpha(\tilde{a},\tilde{b})
			+ \alpha G_{\alpha+1}(\tilde{a},\tilde{b})}{b-a} , \\
	J_{20}(\theta_0,\theta_1) \
	&= \ \frac{e^{\tilde{\theta}_0}}{(1-\tilde{\theta}_1)^{\alpha+2}} \,
		\frac{\tilde{b}^2 G_\alpha(\tilde{a},\tilde{b})
			- 2\alpha\tilde{b} G_{\alpha+1}(\tilde{a},\tilde{b})
			+ \alpha(\alpha+1) G_{\alpha+2}(\tilde{a},\tilde{b})}{(b-a)^2} , \\
	J_{11}(\theta_0,\theta_1) \
	&= \ \frac{e^{\tilde{\theta}_0}}{(1-\tilde{\theta}_1)^{\alpha+2}} \,
		\frac{-\tilde{a}\tilde{b}G_\alpha(\tilde{a},\tilde{b})
			+ \alpha(\tilde{a}+\tilde{b}) G_{\alpha+1}(\tilde{a},\tilde{b})
			- \alpha(\alpha+1) G_{\alpha+2}(\tilde{a},\tilde{b})}{(b-a)^2} , \\
	J_{02}(\theta_0,\theta_1) \
	&= \ \frac{e^{\tilde{\theta}_0}}{(1-\tilde{\theta}_1)^{\alpha+2}} \,
		\frac{\tilde{a}^2G_\alpha(\tilde{a},\tilde{b})
			- 2\alpha\tilde{a} G_{\alpha+1}(\tilde{a},\tilde{b})
			+ \alpha(\alpha+1) G_{\alpha+2}(\tilde{a},\tilde{b})}{(b-a)^2} .
\end{align*}

\paragraph{Local parametrizations.}
Let us fix an arbitrary vector $\taub$ with $m \ge 1$ components $0 \le \tau_1 < \cdots < \tau_m < x_n$. Any function $\theta : [0,\infty) \to \R$ which is constant on $[0,\tau_1]$ and linear on the intervals $\XX_1,\ldots,\XX_m$ specified in Lemma~\ref{lem:existence.uniqueness.2B} is uniquely determined by the vector $\thb = (\theta_j)_{j=1}^{m+1} := \bigl( \theta(\tau_1),\ldots,\theta(\tau_m), \theta'(\tau_m\,+) \bigr)^\top \in \R^{m+1}$. Then $L(\theta)$ is given by
\begin{align}
\nonumber
	L(\taub,\thb) \
	:= \ &\sum_{i=1}^n w_i \theta(x_i)
		- e_{}^{\theta_1} F_0(\tau_1)
		- \sum_{j=1}^m \int_{\XX_j} e_{}^{\theta_j + \theta_j' (x - \tau_j)} \, P_o(dx)
		+ 1 \\
\label{eq:LocalL.Setting2B}
	= \ &\sum_{j=1}^{m+1}\tilde{w}_j\theta_j - e^{\theta_1} G_\alpha(\tau_1)
	- \sum_{1\leq j < m} J(\theta_j,\theta_{j+1};\tau_j,\tau_{j+1})
	- K(\theta_m,\theta_{m+1};\tau_m)
	+ 1
\end{align}
with the auxiliary functions $G_\alpha(\cdot)$, $J(\cdot,\cdot;\cdot,\cdot)$ and $K(\cdot,\cdot;\cdot)$ introduced before and the weights
\begin{align*}
	\tilde{w}_1 \
	&:= \ \sum_{i=1}^n
		\min \Bigl( 1, \frac{(\tau_2 - x_i)^+}{\tau_2 - \tau_1} \Bigr) \, w_i , \\
	\tilde{w}_j \
	&:= \ \sum_{i=1}^n \Bigl(
		1_{[x_i \le \tau_j]} \,
			\frac{(x_i - \tau_{j-1})^+}{\tau_j - \tau_{j-1}}
		+ 1_{[x_i > \tau_j]} \,
			\frac{(\tau_{j+1} - x_i)^+}{\tau_{j+1} - \tau_j} \Bigr) w_i
		\quad\text{for} \ 1 < j < m , \\
	\tilde{w}_m \
	&:= \ \sum_{i=1}^n
		\min \Bigl( 1, \frac{(x_i - \tau_{m-1})^+}{\tau_m - \tau_{m-1}} \Bigr) w_i , \\
	\tilde{w}_{m+1} \
	&:= \ \sum_{i=1}^n (x_i - \tau_m)^+ \, w_i .	
\end{align*}
In case of $m = 1$, the weight $\tilde{w}_1$ is just given by $\tilde{w}_1 = 1$.

The function $L(\taub,\cdot): \R^{m+1} \to [-\infty,\infty)$ is continuous and concave. On the open set $\bigl\{ \thb \in \R^{m+1} : L(\taub,\thb) > - \infty \bigr\} = \bigl\{ \thb \in \R^{m+1} : \theta_{m+1} < 1 \bigr\}$ it is twice continuously differentiable with negative definite Hessian matrix, see the next paragraph.

\paragraph{Gradient vector and Hessian matrix for $L(\bs \tau,\cdot)$ in \eqref{eq:LocalL.Setting2B}.}
Let $\theta_{m+1} < 1$. In case of $m\geq 2$, the gradient $\bs g(\bs\tau,\bs\theta)=\left(g_j(\bs\tau,\bs\theta)\right)_{j=1}^{m+1}$ of $L(\bs\tau,\cdot)$ equals
\[
	g_j(\bs\tau,\bs\theta) \ = \ \tilde{w}_j -
	\begin{cases}
	e^{\theta_1} G_\alpha(\tau_1) + J_{10}(\theta_1,\theta_2;\tau_1,\tau_2)
		& \text{if} \ j = 1 , \\	
	J_{01}(\theta_{j-1},\theta_{j};\tau_{j-1},\tau_{j})
		+ J_{10}(\theta_{j},\theta_{j+1};\tau_{j},\tau_{j+1})
		& \text{if} \ 1 < j < m , \\
	J_{01}(\theta_{m-1},\theta_{m};\tau_{m-1},\tau_{m})
		+ K(\theta_m,\theta_{m+1};\tau_m)
		& \text{if} \ j = m , \\
	K_1(\theta_m,\theta_{m+1};\tau_m)
		& \text{if} \ j = m+1 ,
	\end{cases}
\]
while its negative Hessian matrix $\bs H(\bs\tau,\bs\theta)=\left(H_{jk}(\bs\tau,\bs\theta)\right)_{j,k=1}^{m+1}$ is given by
\begin{align*}
	H_{11}(\bs\tau,\bs\theta)
	&= e^{\theta_1} G_\alpha(\tau_1) + J_{20}(\theta_1,\theta_2;\tau_1,\tau_2) , \\
	H_{j,j+1}(\bs\tau,\bs\theta) = H_{j+1,j}(\bs\tau,\bs\theta)
	&= J_{11}(\theta_{j},\theta_{j+1};\tau_{j},\tau_{j+1})
		\ \ \text{for} \ 1 \le j < m , \\
	H_{jj}(\bs\tau,\bs\theta)
	&= J_{02}(\theta_{j-1},\theta_{j};\tau_{j-1},\tau_{j})
		+ J_{20}(\theta_{j},\theta_{j+1};\tau_{j},\tau_{j+1})
		\ \ \text{for} \ 1 < j < m , \\
	H_{mm}(\bs\tau,\bs\theta)
	&= J_{02}(\theta_{m-1},\theta_{m};\tau_{m-1},\tau_{m})
		+ K(\theta_m,\theta_{m+1};\tau_m) , \\
	H_{m,m+1}(\bs\tau,\bs\theta) = H_{m+1,m}(\bs\tau,\bs\theta)
	&= K_1(\theta_m,\theta_{m+1};\tau_m) , \\
	H_{m+1,m+1}(\bs\tau,\bs\theta)
	&= K_2(\theta_m,\theta_{m+1};\tau_m) , \\
	H_{jk}(\bs\tau,\bs\theta)
	&= 0 \ \ \text{if} \ |j-k| > 1 .
\end{align*}

In case of $m=1$ we get the simplified formulae
\[
	L(\bs\tau,\bs\theta) \ = \
	\sum_{j=1}^{2}\tilde{w}_j\theta_j - e^{\theta_1} G_\alpha(\tau_1)
	- K(\theta_1,\theta_{2};\tau_1) + 1 ,
\]
\[
	g_j(\bs\tau,\bs\theta) \ = \ \tilde{w}_j -
	\begin{cases}
		e^{\theta_1} G_\alpha(\tau_1)+K(\theta_1,\theta_{2};\tau_1)
			& \text{if} \ j = 1 , \\	
		K_1(\theta_1,\theta_{2};\tau_1)
			& \text{if} \ j = 2 ,
	\end{cases}
\]
and
\begin{align*}
	H_{11}(\bs\tau,\bs\theta) \
	&= \ e^{\theta_1} G_\alpha(\tau_1) + K(\theta_1,\theta_{2};\tau_1), \\
	H_{12}(\bs\tau,\bs\theta) = H_{21}(\bs\tau,\bs\theta) \
	&= \ K_1(\theta_1,\theta_{2};\tau_1),\\
	H_{22}(\bs\tau,\bs\theta) \
	&= \ K_2(\theta_1,\theta_{2};\tau_1).
\end{align*}

\paragraph{Evaluating $h_\theta(\tau):=DL(\theta,V_{\tau,\theta})$ and $h_\theta'(\tau+)$.}
Suppose first that $\theta \equiv 0$, so $D(\theta)=\emptyset$. Then one can show that
\begin{align*}
	h_\theta'(\tau\,+) \
	&= \ (\hat{F} - G_\alpha)(\tau) , \\
	h_\theta(\tau) \
	&= \ \tau h_\theta'(\tau\,+)
		+ \hat{\mu} - \int_{[0,\tau]} x \, \hat{P}(dx)
			- \alpha + \alpha G_{\alpha+1}(\tau) .
\end{align*}

Now suppose that $\theta$ is given by a vector $\bs\tau$ of $m\geq 1$ points $\tau_1<\cdots<\tau_m$ and a vector $\bs\theta=(\theta_j)_{j=1}^{m+1}$ as in \eqref{eq:LocalL.Setting2B}. Then
\[
	h_\theta(0) \
	= \ - \tau_1 \bigl( \hat{F}(\tau_1) - e^{\theta_1} G_\alpha(\tau_1) \bigr)
		+ \int 1_{[x \le \tau_1]} x \, \hat{P}(dx)
		- e^{\theta_1} \alpha G_{\alpha+1}(\tau_1) ,
\]
while for $0 \le \tau<\tau_1$
\begin{align*}
	h_\theta'(\tau\,+) \
	&= \ \hat{F}(\tau) - e^{\theta_1} G_\alpha(\tau) , \\
	h_\theta(\tau) \
	&= \ h_\theta(0) + \tau h_\theta'(\tau\,+)
		- \int 1_{[x\leq \tau]} x \,\hat{P}(dx)
		+ e^{\theta_1} \alpha G_{\alpha+1}(\tau) .
\end{align*}
For $1\le j < m$ and $\tau \in [\tau_j,\tau_{j+1})$,
\begin{align*}
	h_\theta '(\tau\,+) \
	&= \ \hat{P}((\tau_j,\tau])
		- J(\theta_j,\theta_{*};\tau_j,\tau)
		- \int j_{10}(x;\tau_j,\tau_{j+1})\,\hat{P}(dx)
		+ J_{10}(\theta_j,\theta_{j+1};\tau_j,\tau_{j+1}) , \\
	h_\theta(\tau) \
	&= \ (\tau-\tau_j)\left(h_\theta '(\tau+)
		+ J_{01}(\theta_j,\theta_*;\tau_j,\tau)\right)
		- \int 1_{[\tau_j<x\leq \tau]}(x-\tau_j)\,\hat{P}(dx) ,
\end{align*}
where $\theta_* := \theta(\tau) = (\tau_{j+1} - \tau_j)^{-1} \bigl( (\tau_{j+1}-\tau) \theta_j + (\tau-\tau_j) \theta_{j+1} \bigr) = \theta_j + (\tau-\tau_j)\theta_j'$. Finally, for $\tau >\tau_m$,
\begin{align*}
	h_\theta '(\tau\,+) \
	&= \ K(\theta_*,\theta_{m+1};\tau) - \hat{P}((\tau,\infty)) , \\
	h_\theta(\tau) \
	&= \ (\tau - \tau_m) \bigl( h_\theta '(\tau+)
			+ J_{01}(\theta_m,\theta_*;\tau_m,\tau) \bigr)
		- \int 1_{[\tau_m<x\leq \tau]} (x - \tau_m) \, \hat{P}(dx) ,
\end{align*}
where $\theta_* := \theta_m + (\tau - \tau_m) \theta_{m+1}$.

If $\tau$ is restricted to some interval $I$ not containing any observations $x_i$ or knots $\tau_j$, the expressions for $h_\theta '(\tau+)$ are constant in $\tau$ except for one term $e^{\theta_1} G_\alpha(\tau)$, $J(\theta_j,\theta_{*};\tau_j,\tau)$ or $K(\theta_*,\theta_{m+1};\tau)$. Hence finding $\tau$ such that $h_\theta '(\tau+)=0$ leads to equations of the following type: For given real numbers $\theta_0,\theta_1,\tau_0$ and $c$, find $\tau\in [0,\infty)$ such that
\begin{align}
e^{\theta_0} G_\alpha(\tau) \ & = \ c,\label{Eq:IsEquToC1}\\
J(\theta_0,\theta_0+\theta_1(\tau-\tau_0);\tau_0,\tau) \ & = \ c,\label{Eq:IsEquToC2}\\
K(\theta_0+\theta_1(\tau-\tau_0),\theta_{1};\tau)\ & = \ c,\label{Eq:IsEquToC3}
\end{align}
and check whether $\tau\in I$. The unique solution of \eqref{Eq:IsEquToC1} is given by
\[
	\tau \ = \ G_\alpha^{-1}(ce^{-\theta_0})
\]
with the quantile function $G_\alpha^{-1} : [0,1) \to [0,\infty)$ of $\mathrm{Gamma}(\alpha,1)$, provided that $0 \le ce^{-\theta_0} < 1$; otherwise no solution exists. It follows from $J(\theta_0,\theta_0+\theta_1(\tau-\tau_0);\tau_0,\tau) = (1-\theta_1)^{-\alpha} e^{\theta_0-\theta_1\tau_0} \bigl( G_\alpha((1-\theta_1)\tau) - G_\alpha((1-\theta_1)\tau_0) \bigr)$ that the unique solution of \eqref{Eq:IsEquToC2} is given by
\[
	\tau \ = \
	(1-\theta_1)^{-1} G_\alpha^{-1}
		\bigl( c(1-\theta_1)^\alpha e^{\theta_1\tau_0-\theta_0}
			+ G_\alpha((1-\theta_1)\tau_0) \bigr) ,
\]
provided that $0\le \theta_1 < 1$ and $0 \le c(1-\theta_1)^\alpha e^{\theta_1\tau_0-\theta_0}+G_\alpha\left((1-\theta_1)\tau_0\right) < 1$; otherwise no solution exists. Likewise it follows from $K(\theta_0+\theta_1(\tau-\tau_0),\theta_{1};\tau)
	= (1-\theta_1)^{-\alpha} e^{\theta_0-\theta_1\tau_0} \bigl( 1 - G_\alpha((1-\theta_1)\tau) \bigr)$ that the unique solution of \eqref{Eq:IsEquToC3} is given by
\[
	\tau \ = \ (1-\theta_1)^{-1} G_\alpha^{-1}
		\bigl( 1 - c(1-\theta_1)^\alpha e^{\theta_1\tau_0-\theta_0} \bigr) ,
\]
provided that $0\le \theta_1 < 1$ and $0 < c(1-\theta_1)^\alpha e^{\theta_1\tau_0-\theta_0} \le 1$; otherwise no solution exists.

\paragraph{Activating one constraint.}
The activation of one constraint is identical to Setting~2A, except that here is no weight $\tilde{w}_0$.

\paragraph{Data Simulation.}
Let $P_o = \mathrm{Gamma}(\alpha,\beta)$, and let $\theta \in \Theta$ such that $\gamma = \gamma(\theta) := \lim_{x \to \infty} \theta'(x\,+) < \beta$ and $\int f_\theta \, dP_o = 1$ with $f_\theta := e^\theta$. To simulate data from the density $f_\theta := e^\theta$ with respect to $P_o$, we use the acceptance rejection method of \cite{vonNeumann_1951}. We simulate independent random variables $Y \sim \text{Gamma}(\alpha, \beta - \gamma)$ and $U \sim \mathrm{Unif}[0,1]$. Note that $Y$ has density $h(x) := (1 - \gamma/\beta)^{-\alpha} e^{\gamma x}$ with respect to $P_o$ and that
\[
	(f_\theta/h)(x)
	\ = \ (f_\theta/h)(0) \, \exp \bigl( \theta(x) - \theta(0) - \gamma x \bigr)
\]
is monotone decreasing in $x \ge 0$. Hence the conditional distribution of $Y$, given that $U \le \exp( \theta(Y) - \theta(0) - \gamma Y \bigr)$ is equal to the desired distribution $P_\theta$. This leads to the following pseudocode for generating an independent sample $\bs{X}$ of size $n$ from $f_\theta$:
\begin{displaymath}
	\begin{array}{l}
	\hline
	i \leftarrow 0\\
	\text{while} \ i < n \ \text{do}\\
	\strut\quad	\text{simulate} \ Y \sim \mathrm{Gamma}(\alpha, \beta - \gamma) \\
	\strut\quad	\text{simulate} \ U  \sim \mathrm{Unif}([0,1]) \\
	\strut\quad \text{if} \ U \leq \exp\bigl( \theta(Y) -
	\theta(0) - \gamma Y \bigr) \ \text{then}\\
	\strut\qquad i \leftarrow i + 1\\
	\strut\qquad X_i \leftarrow Y\\
	\strut\quad \text{end if}\\
	\text{end while} \\
	\hline
	\end{array}
\end{displaymath}

\subsection{Further proofs}
\label{subsec:further.proofs}

\begin{proof}[\bf Continuity of $L$ on $(\V,\|\cdot\|)$ (Section~\ref{subsec:properties.L})]
In Setting~1, the assertion is obvious, so we prove it for Settings~2A-B. Recall that a sequence $(\theta_k)_k$ in $\V$ converges to a function $\theta \in \V$ with respect to $\|\cdot\|$ if and only if it converges uniformly on any bounded subset of $\XX$. Assuming this from now on, we want to show that $L(\theta_k) \to L(\theta)$ as $k \to \infty$. If $L(\theta) = - \infty$, then it follows from Fatou's lemma that
\[
	\limsup_{k \to \infty} L(\theta_k)
	\ = \ \int \theta \, d\hat{P}
		- \liminf_{k \to \infty} \int e_{}^{\theta_k} \, dP_o + 1
	\ \le \ L(\theta) = - \infty .
\]
If $L(\theta) > - \infty$, then $\int \exp \bigl( \theta(x) + \eps (1 + |x|) \bigr) \, P_o(dx) < \infty$ for sufficiently small $\eps > 0$, and for sufficiently large $k$, $\theta_k(x) \le \theta(x) + \eps (1 + |x|)$ for all $x \in \R$. Hence, it follows from dominated convergence that $L(\theta_k) \to L(\theta)$ as $k \to \infty$.
\end{proof}

\begin{proof}[\bf Proof of Remark~\ref{rem:from.L.to.theta}]
Let $(\theta_k)_k$ be a sequence in $\Theta \cap \V$ such that $L(\theta_k) \to L(\hat{\theta})$ but $\theta_k \not\to \hat{\theta}$ pointwise as $k \to \infty$. As in the proof of Lemmas~\ref{lem:existence.uniqueness.2A} and \ref{lem:existence.uniqueness.2B}, we may replace this sequence by a subsequence, if necessary, such that it converges to some function $\theta_* \in \Theta \setminus \{\hat{\theta}\}$ with respect to $\|\cdot\|$. Since $L$ is continuous, this implies that $L(\theta_k) \to L(\theta_*)$ as $k \to \infty$, whence $L(\theta_*) = L(\hat{\theta})$. Now, uniqueness of the maximizer of $L$ on $\Theta$ leads to the contradiction that $\theta_* = \hat{\theta}$.
\end{proof}

\begin{proof}[\bf Proof of Lemma~\ref{lem:convergence} for Setting~2B and Setting~1]
We only indicate the main changes in the proof for Setting~2A.

In Setting~2B, the constant $C_\ell$ may be replaced with $0$, and the set $\VV$ of basis functions consists of $v_0 \equiv 1$ and $V_\tau$, $\tau \in \DD$. This leads to $v_{\rm max}(x) = \max(1,x)$, and $\theta_{\rm max}(x) = C_o + C_r (x - x_n)^+$. Moreover, $\hat{\theta} - \theta = \alpha_0 + \sum_{\tau \in \DD} \beta_\tau V_\tau$ with $|\alpha_0| = \bigl| \hat{\theta}(0) - \theta(0) \bigr| \le 2C_o$, and
\[
	\sum_{\tau \in \DD} \beta_\tau^+ \ \le \ \hat{\theta}'(x_n) \ \le \ C_r,
	\quad
	\sum_{\tau \in \DD} \beta_\tau^- \ \le \ \theta'(x_n) \ \le \ C_r .
\]
Here $|DL(\theta,\eta_{\tau,\theta})| \le (1 + x_n) \sqrt{C_{\rm N} \delta_{\rm Newton}(\theta)}$, and this leads to obvious changes in the upper bound for $DL(\theta, \hat{\theta} - \theta)$.

In Setting~1, the main changes are as follows. We do not need the constants $C_\ell, C_r$, and integrals $\int \cdots \, P_o(dx)$ have to be replaced with integrals $\int_{x_1}^{x_n} \cdots \, dx$. Here $v_{\rm max}(x) = \max(1,x-x_1)$, and $\theta_{\rm max} \equiv C_o$. The difference $\hat{\theta} - \theta$ equals $\alpha_0 v_0 + \alpha_1 v_1 + \sum_{\tau \in \DD} \beta_\tau V_\tau$ with
\begin{align*}
	|\alpha_0| \ &= \ \bigl| \hat{\theta}(x_1) - \theta(x_1) \bigr|
		\ \le \ 2C_o , \\
	|\alpha_1| \ &= \ \bigl| \hat{\theta}'(x_1\,+) - \theta'(x_1\,+) \bigr|
		\ \le \ 4C_o/(x_2 - x_1) , \\
	\beta_\tau \ &= \ \hat{\theta}'(\tau\,-) - \hat{\theta}'(\tau\,+)
		- \bigl( \theta'(\tau\,-) - \theta'(\tau\,+) \bigr)
		\ \begin{cases}
		\le \ \ \hat{\theta}'(\tau\,-) - \hat{\theta}'(\tau\,+) , \\
		\ge \ - \bigl( \theta'(\tau\,-) - \theta'(\tau\,+) \bigr) .
	\end{cases}
\end{align*}
In particular,
\[
	\left.\begin{array}{c}
		\displaystyle
		\sum_{\tau \in \DD} \beta_\tau^+ 
		\ \le \ \hat{\theta}'(x_1\,+) - \hat{\theta}'(x_n\,-) \\[3ex]
		\displaystyle
		\sum_{\tau \in \DD} \beta_\tau^- 
		\ \le \ \theta'(x_1\,+) - \theta'(x_n\,-)
	\end{array}\!\!\right\}
	\ \le \ 2C_o/\min\{x_2 - x_1,x_n - x_{n-1}\} .
\]
Here we utilized the fact that $v'(x_i\,+) = v'(x_{i+1}\,-) = \bigl( v(x_{i+1}) - v(x_i) \bigr) / (x_{i+1} - x_i)$ for $v \in \V$ and $1 \le i < n$. Finally, $|DL(\theta,\eta_{\tau,\theta})| \le \sqrt{C_{\rm N} \delta_{\rm Newton}(\theta)}$, because $\eta_{\tau,\theta}$ is always a convex combination of two basis functions in $\VV \cap \V_{D(\theta)}$.
\end{proof}

\subsection{On the distribution of $T_{LR}$ under the null hypothesis}
\label{subsec:null.distributions}

For the goodness-of-fit tests with a given sample size $n$, we simulated $10^5 - 1$ times a sample $X_1,\ldots,X_n$ from $P_o$ and recorded the test statistic $T_{LR} = T_{LR}(X_1,\ldots,X_n)$ as well as the number $M = M(X_1,\ldots,X_n)$ of kinks of the estimator $\hat{\theta} = \hat{\theta}(\cdot \,|\, X_1,\ldots,X_n)$. The reference distribution $P_o$ was $\NN(0,1)$ in Setting~2A and $\chi_1^2$ in Setting~2B. In the latter setting, we also recorded the indicator $J = J(X_1,\ldots,X_n)$ that $\hat{\theta}$ has a kink at $0$, i.e.\ $\hat{\theta}'(0\,+) > 0$.

Table~\ref{tab:crit.values} contains critical values $\hat{\kappa}_{n,\alpha}$ for different sample sizes $n$ and different test levels $\alpha$. Tables~\ref{tab:distr.M.2A} and \ref{tab:distr.M.2B} contain the estimated distribution of the random number $M$ in Settings~2A and 2B, respectively. In the latter setting, Monte Carlo estimators of probabilities $P(J = 1, M \cdots)$ are listed as well.

\begin{table}
\[
	\begin{array}{clcccc}
	\multicolumn{5}{l}{\text{Setting~2A}} \\
	\hline
	n	&& \hat{\kappa}_{n,0.10}
		 & \hat{\kappa}_{n,0.05}
		 & \hat{\kappa}_{n,0.01}
		 & \text{time (ms)} \\
	\hline\hline
	100  && 2.923 & 3.763 & 5.653 & 3.087 \\
	\hline
	400  && 3.298 & 4.179 & 6.133 & 4.282 \\
	\hline
	1000 && 3.531 & 4.434 & 6.473 & 5.880 \\
	\hline
	2000 && 3.682 & 4.613 & 6.678 & 8.355 \\
	\hline
	\end{array}
	\qquad
	\begin{array}{clcccc}
	\multicolumn{5}{l}{\text{Setting~2B}} \\
	\hline
	n	&& \hat{\kappa}_{n,0.10}
		 & \hat{\kappa}_{n,0.05}
		 & \hat{\kappa}_{n,0.01}
		 & \text{time (ms)} \\
	\hline\hline
	100  && 1.228 & 1.863 & 3.378 & 1.795 \\
	\hline
	400  && 1.481 & 2.160 & 3.751 & 2.736 \\
	\hline
	1000 && 1.622 & 2.317 & 3.879 & 4.228 \\
	\hline
	2000 && 1.736 & 2.418 & 4.128 & 6.676 \\
	\hline
	\end{array}
\]
\caption{Some estimated critical values for goodness-of-fit tests and mean running time per sample from $P_o$.}
\label{tab:crit.values}
\end{table}

\begin{table}
\[
	\begin{array}{clccccccc}
	\hline
	n    && 0     & 1     & 2     & 3     & 4     & 5     & >5    \\
	\hline
	100  && 0.164 & 0.324 & 0.296 & 0.154 & 0.050 & 0.011 & 0.002 \\
	\hline
	400  && 0.100 & 0.258 & 0.301 & 0.208 & 0.095 & 0.030 & 0.008 \\
	\hline
	1000 && 0.075 & 0.217 & 0.290 & 0.231 & 0.123 & 0.047 & 0.017 \\
	\hline
	2000 && 0.059 & 0.187 & 0.277 & 0.245 & 0.146 & 0.062 & 0.025 \\
	\hline
	\end{array}    
\]
\caption{Estimators of $P(M = m)$, $0 \le m \le 5$, and $P(M > 5)$ in Setting~2A.}
\label{tab:distr.M.2A}
\end{table}

\begin{table}
\[
	\begin{array}{clcccccc}
	\hline
	n    && 0     & 1     & 2     & 3     & 4     & >4    \\
	\hline
	100  && 0.360 & 0.445 & 0.165 & 0.028 & 0.002 & 0.000 \\
	     &&(0.000)&(0.069)&(0.029)&(0.006)&(0.000)&(0.000)\\
	\hline
	400  && 0.292 & 0.432 & 0.216 & 0.053 & 0.007 & 0.001 \\
	     &&(0.000)&(0.050)&(0.030)&(0.010)&(0.001)&(0.000)\\
	\hline
	1000 && 0.252 & 0.419 & 0.244 & 0.072 & 0.012 & 0.001 \\
	     &&(0.000)&(0.040)&(0.031)&(0.010)&(0.002)&(0.000)\\
	\hline
	2000 && 0.229 & 0.403 & 0.263 & 0.086 & 0.017 & 0.002 \\
	     &&(0.000)&(0.034)&(0.030)&(0.011)&(0.002)&(0.000)\\
	\hline
	\end{array}    
\]
\caption{Estimators of $P(M = m)$, $0 \le m \le 4$, and $P(M > 4)$ in Setting~2B. In brackets are the estimators of $P(J = 1, M \ldots)$.}
\label{tab:distr.M.2B}
\end{table}

\end{document}